\documentclass[aps,prb,twocolumn, superscriptaddress, showpacs]{revtex4}
\usepackage{amsmath}
\usepackage{amssymb}
\usepackage{graphicx}
\usepackage{mathtools}
\usepackage{relsize}
\usepackage{lmodern}
\usepackage{slantsc}
\usepackage{scalefnt}
\usepackage{subfigure}
\usepackage{dsfont}
\usepackage{ifpdf}
\usepackage{pgffor}
\usepackage{verbatim}
\usepackage{tikz}
\usetikzlibrary{shapes.geometric}
\usepackage{color}
\usetikzlibrary{arrows,matrix,calc,scopes,decorations.markings}
\allowdisplaybreaks[1]

\newtheorem{theorem}{Theorem}
\newtheorem{proposition}[theorem]{Proposition}
\newenvironment{proof}[1][Proof]{\noindent\textbf{#1.} }{\ \rule{0.5em}{0.5em}}
\newcommand{\be}{\begin{equation}}
\newcommand{\ee}{\end{equation}}
\newcommand{\bse}{\begin{subequations}}
\newcommand{\ese}{\end{subequations}}
\newcommand{\ket}[1]{|{#1}\rangle}

\newcommand{\bbra}[1]{\left\langle{#1}\right|}
\newcommand{\Z}{\mathbb{Z}}
\newcommand{\ii}{\mathrm{i}}

\newcommand{\Hil}{\mathcal{H}}

\newcommand{\str}{\mathcal{S}}
\newcommand{\ttr}{\mathcal{T}}
\newcommand{\bpm}{\begin{pmatrix}}
\newcommand{\epm}{\end{pmatrix}}
\newcommand{\bmm}{\begin{matrix}}
\newcommand{\emm}{\end{matrix}}
\newcommand{\elf}[4]{\alpha\left([#1#2],[#2#3],[#3#4]\right)}
\newcommand{\dlt}[3]{\delta_{[#1#2]\cdot[#2#3]\cdot[#3#1]}}
\newcommand{\Blangle}{\Biggl\langle\bmm} 
\newcommand{\BLvert}{\Biggl\vert\bmm} 
\newcommand{\Bvert}{\emm\Biggr\vert\bmm} 
\newcommand{\Brangle}{\emm\Biggr\rangle}
\newcommand{\defeq}{\stackrel{\mathrm{def}}{=}}
\newcommand{\cb}[2]{#1\underset{\!\scalebox{1}{$#2$}}{\square}}
\makeatletter
\newcommand*{\Relbarfill@}{\arrowfill@\Relbar\Relbar\Relbar}
\newcommand*{\xeq}[2][]{\ext@arrow 0055\Relbarfill@{#1}{#2}}
\makeatother

\tikzset{->-/.style={decoration={
  markings,
  mark=at position .5 with {\arrow{>}}},postaction={decorate}}}
\tikzset{-<-/.style={decoration={
  markings,
  mark=at position .5 with {\arrow{<}}},postaction={decorate}}}

\newcommand{\tetrahedron}[7]{
\begin{tikzpicture}[scale=1,baseline]
\coordinate (a) at(0,0);
\coordinate (b) at (0.8,-0.75);
\coordinate (c) at (2.5,0);
\coordinate (d) at (1.25,2);
\coordinate (e) at ($ (a) ! 0.3 ! (c) $);
\coordinate (f) at ($ (a) ! 0.45 ! (c) $);
\draw (a) -- (b) -- (c) -- (d) -- cycle;
\draw (b) -- (d);
\draw (a) -- (e);
\draw (f) -- (c);
\draw[-<,>=latex,line width=0.01pt] (a) -- ($ (a) ! 0.5 ! (b) $);
\draw[-<,>=latex,line width=0.01pt] (b) -- ($ (b) ! 0.5 ! (c) $);
\draw[-<,>=latex,line width=0.01pt] (c) -- ($ (c) ! 0.5 ! (d) $);
\draw[<-,>=latex,line width=0.01pt] (f) -- (c);
\draw[-<,>=latex,line width=0.01pt] (b) -- ($ (b) ! 0.5 ! (d) $);
\draw[-<,>=latex,line width=0.01pt] (a) -- ($ (a) ! 0.5 ! (d) $);

\node[left] at(a) {\scalebox{0.7}{$#1$}};
\node[below] at(b) {\scalebox{0.7}{$#2$}};
\node[right] at(c) {\scalebox{0.7}{$#3$}};
\node[above] at(d) {\scalebox{0.7}{$#4$}};

\node[below] at ($ (a) ! 0.5 ! (b) $) {\scalebox{0.7}{$#5$}};
\node[below] at ($ (b) ! 0.5 ! (c) $) {\scalebox{0.7}{$#6$}};
\node[right] at ($ (c) ! 0.5 ! (d) $) {\scalebox{0.7}{$#7$}};
\node[below] at (f) {\scalebox{0.7}{$#5#6$}};
\node[right] at ($ (b) ! 0.5 ! (d) $) {\scalebox{0.7}{$#6#7$}};
\node[left] at ($ (a) ! 0.5 ! (d) $) {\scalebox{0.7}{$#5#6#7$}};
\end{tikzpicture}
}

\newcommand{\tetrahedronSplit}[9]{
\begin{tikzpicture}[scale=1,baseline]
\coordinate (a) at(0,0);
\coordinate (b) at (0.8,-0.75);
\coordinate (c) at (2.5,0);
\coordinate (d) at (1.25,2);
\coordinate (e) at ($ (a) ! 0.3 ! (c) $);
\coordinate (f) at ($ (a) ! 0.45 ! (c) $);
\coordinate (g) at ($ (a) ! 0.5 ! (b) $);
\coordinate (h) at ($ (c) ! 0.5 ! (d) $);
\coordinate (n) at ($ (g) ! 0.5 ! (h) $); 
\draw (a) -- (b) -- (c) -- (d) -- cycle;
\draw (b) -- (d);
\draw (a) -- (e);
\draw (f) -- (c);
\draw[-<,>=latex,line width=0.01pt] (a) -- ($ (a) ! 0.5 ! (b) $);
\draw[-<,>=latex,line width=0.01pt] (b) -- ($ (b) ! 0.5 ! (c) $);
\draw[-<,>=latex,line width=0.01pt] (c) -- ($ (c) ! 0.5 ! (d) $);
\draw[<-,>=latex,line width=0.01pt] (f) -- (c);
\draw[-<,>=latex,line width=0.01pt] (b) -- ($ (b) ! 0.5 ! (d) $);
\draw[-<,>=latex,line width=0.01pt] (a) -- ($ (a) ! 0.5 ! (d) $);

\node[left] at(a) {\scalebox{0.7}{$#1$}};
\node[below] at(b) {\scalebox{0.7}{$#2$}};
\node[right] at(c) {\scalebox{0.7}{$#3$}};
\node[above] at(d) {\scalebox{0.7}{$#4$}};
\node[circle,fill=black,outer sep=0pt, inner sep=0.6pt] at(n) {};
\node[below] at ($ (a) ! 0.5 ! (b) $) {\scalebox{0.7}{$#5$}};
\node[below] at ($ (b) ! 0.5 ! (c) $) {\scalebox{0.7}{$#6$}};
\node[right] at ($ (c) ! 0.5 ! (d) $) {\scalebox{0.7}{$#7$}};
\node at (0.6,-0.15) {\scalebox{0.7}{$#5#6$}};
\node[left] at ($ (b) ! 0.5 ! (d) $) {\scalebox{0.7}{$#6#7$}};
\node[left] at ($ (a) ! 0.5 ! (d) $) {\scalebox{0.7}{$#5#6#7$}};

\draw[dashed] (a) -- (n);
\draw[dashed] (b) -- (n);
\draw[dashed] (c) -- (n);
\draw[dashed] (d) -- (n);

\ifnum #8=1 {
  \node[right, outer sep=0pt, inner sep=1pt] at($ (n) ! 0.07 ! (d) $) {\scalebox{0.7}{$#1'$}};
  \draw[->,>=latex,line width=0.01pt] ($ (a) ! 0.4 ! (n) $) -- ($ (a) ! 0.5 ! (n) $) node[above] {\scalebox{0.7}{$#9$}};
  \draw[->,>=latex,line width=0.01pt] ($ (b) ! 0.4 ! (n) $) node[right] {\scalebox{0.7}{$#9#5$}} -- ($ (b) ! 0.5 ! (n) $) ;
  \draw[->,>=latex,line width=0.01pt] ($ (c) ! 0.4 ! (n) $) -- ($ (c) ! 0.5 ! (n) $) node[above] {\scalebox{0.7}{$#9#5#6$}};
  \draw[->,>=latex,line width=0.01pt] ($ (d) ! 0.4 ! (n) $) -- ($ (d) ! 0.5 ! (n) $) node[right,outer sep=0pt, inner sep=1pt] {\scalebox{0.6}{$#9#5#6#7$}};}
\fi
\ifnum #8=2 {
  \node[right, outer sep=0pt, inner sep=1pt] at($ (n) ! 0.07 ! (d) $) {\scalebox{0.7}{$#2'$}};
  \draw[-<,>=latex,line width=0.01pt] ($ (a) ! 0.4 ! (n) $) -- ($ (a) ! 0.5 ! (n) $) node[above, outer sep=0pt, inner sep=1pt] {\scalebox{0.7}{$#5#9^{-1}$}};
  \draw[->,>=latex,line width=0.01pt] ($ (b) ! 0.4 ! (n) $) node[right] {\scalebox{0.7}{$#9$}} -- ($ (b) ! 0.5 ! (n) $) ;
  \draw[->,>=latex,line width=0.01pt] ($ (c) ! 0.4 ! (n) $) -- ($ (c) ! 0.5 ! (n) $) node[above] {\scalebox{0.7}{$#9#6$}};
  \draw[->,>=latex,line width=0.01pt] ($ (d) ! 0.4 ! (n) $) -- ($ (d) ! 0.5 ! (n) $) node[right,outer sep=0pt, inner sep=1pt] {\scalebox{0.6}{$#9#6#7$}};}
\fi
\ifnum #8=3 {
  \node[right, outer sep=0pt, inner sep=1pt] at($ (n) ! 0.07 ! (d) $) {\scalebox{0.7}{$#3'$}};
  \draw[-<,>=latex,line width=0.01pt] ($ (a) ! 0.4 ! (n) $) -- ($ (a) ! 0.5 ! (n) $) node[above, outer sep=0pt, inner sep=1pt] {\scalebox{0.7}{$#5#6#9^{-1}$}};
  \draw[-<,>=latex,line width=0.01pt] ($ (b) ! 0.4 ! (n) $) node[right, outer sep=0pt, inner sep=1pt] {\scalebox{0.7}{$#6#9^{-1}$}} -- ($ (b) ! 0.5 ! (n) $) ;
  \draw[->,>=latex,line width=0.01pt] ($ (c) ! 0.4 ! (n) $) -- ($ (c) ! 0.5 ! (n) $) node[above] {\scalebox{0.7}{$#9$}};
  \draw[->,>=latex,line width=0.01pt] ($ (d) ! 0.4 ! (n) $) -- ($ (d) ! 0.5 ! (n) $) node[right,outer sep=0pt, inner sep=1pt] {\scalebox{0.6}{$#9#7$}};}
\fi
\ifnum #8=4 {
  \node[right, outer sep=0pt, inner sep=1pt] at($ (n) ! 0.07 ! (d) $) {\scalebox{0.7}{$#4'$}};
  \draw[-<,>=latex,line width=0.01pt] ($ (a) ! 0.4 ! (n) $) -- ($ (a) ! 0.5 ! (n) $) node[above, outer sep=0pt, inner sep=1pt] {\scalebox{0.7}{$#5#6#7#9^{-1}$}};
  \draw[-<,>=latex,line width=0.01pt] ($ (b) ! 0.4 ! (n) $) node[right, outer sep=0pt, inner sep=1pt] {\scalebox{0.7}{$#6#7#9^{-1}$}} -- ($ (b) ! 0.5 ! (n) $) ;
  \draw[-<,>=latex,line width=0.01pt] ($ (c) ! 0.4 ! (n) $) -- ($ (c) ! 0.5 ! (n) $) node[above] {\scalebox{0.7}{$#7#9^{-1}$}};
  \draw[->,>=latex,line width=0.01pt] ($ (d) ! 0.4 ! (n) $) -- ($ (d) ! 0.5 ! (n) $) node[right,outer sep=0pt, inner sep=1pt] {\scalebox{0.6}{$#9$}};}
\fi
\end{tikzpicture}
}

\newcommand{\oneTriangle}[4][0]{
\begin{tikzpicture}[scale=1,baseline]
\draw (0,0) node[left] {\scalebox{0.7}{$#2$}} -- (1,0) node[right] {\scalebox{0.7}{$#3$}} -- (60:1) node[above] {\scalebox{0.7}{$#4$}} -- cycle;
\ifnum #1=1 {
   \draw[-<,>=latex,line width=0.01pt] (0,0) -- (0.5,0) node[below] {\scalebox{0.7}{$[#2#3]$}};
   \draw[-<,>=latex,line width=0.01pt] (0,0) -- (60:0.5) node[left] {\scalebox{0.7}{$[#3#4]$}};
   \draw[->,>=latex,line width=0.01pt] (0,0) +(60:1) -- +(30:0.866) node[right] {\scalebox{0.7}{$[#2#4]$}};}
\else {
   \draw[-<,>=latex,line width=0.01pt] (0,0) -- (0.5,0) ;
   \draw[-<,>=latex,line width=0.01pt] (0,0) -- (60:0.5) ;
   \draw[->,>=latex,line width=0.01pt] (0,0) +(60:1) -- +(30:0.866);}  \fi
\end{tikzpicture}
}

\newcommand{\oneTriangleInKitaevModel}{
\begin{tikzpicture}[scale=1,baseline]
 \draw (0,0) node[left] {\scalebox{0.7}{$v_1$}} -- (1,0) node[right] {\scalebox{0.7}{$v_2$}} -- (60:1) node[above] {\scalebox{0.7}{$v_3$}} -- cycle;
   \draw[-<,>=latex,line width=0.01pt] (0,0) -- (0.5,0) node[below] {\scalebox{0.7}{$g_1$}};
   \draw[->,>=latex,line width=0.01pt] (0,0) -- (60:0.5) node[left] {\scalebox{0.7}{$g_3$}};
   \draw[->,>=latex,line width=0.01pt] (0,0) +(60:1) -- +(30:0.866) node[right] {\scalebox{0.7}{$g_2$}};
\end{tikzpicture}
}

\newcommand{\twoTriangles}[5]{
\begin{tikzpicture}[scale=0.8,baseline]
\coordinate (b) at (0,0);
\coordinate (c) at (30:1);
\coordinate (d) at (0,1);
\coordinate (a) at (150:1);

\draw (a) node[left] {\scalebox{0.7}{$#1$}} -- (b) node[below] {\scalebox{0.7}{$#2$}} -- (c) node[right] {\scalebox{0.7}{$#3$}} -- (d) node[above] {\scalebox{0.7}{$#4$}} -- cycle;
\draw[-<,>=latex,line width=0.01pt] (a) -- ($ (a) ! 0.5 ! (b) $);
\draw[-<,>=latex,line width=0.01pt] (b) -- ($ (b) ! 0.5 ! (c) $);
\draw[-<,>=latex,line width=0.01pt] (c) -- ($ (c) ! 0.5 ! (d) $);
\draw[-<,>=latex,line width=0.01pt] (a) -- ($ (a) ! 0.5 ! (d) $);

\ifnum #5=0 {
   \draw (b)  -- (d);
   \draw[-<,>=latex,line width=0.01pt] (b) -- ($ (b) ! 0.5 ! (d) $);}
\else {
   \draw (a)  -- (c);
   \draw[-<,>=latex,line width=0.01pt] (a) -- ($ (a) ! 0.5 ! (c) $);}
\fi
\end{tikzpicture}
}

\newcommand{\threeTriangles}[5]{
\begin{tikzpicture}[scale=0.8,baseline]
\coordinate (c) at (0,0);
\coordinate (d) at (0,1);
\coordinate (a) at (210:1);
\coordinate (b) at (-30:1);

\draw (a) node[left] {\scalebox{0.7}{$#1$}} -- (b) node[right] {\scalebox{0.7}{$#2$}} -- (c) node[below] {\scalebox{0.7}{$#3$}} -- (d) node[above] {\scalebox{0.7}{$#4$}} -- (a) -- (c);
\draw (b) -- (d);
\draw[-<,>=latex,line width=0.01pt] (a) -- ($ (a) ! 0.5 ! (b) $);
\draw[-<,>=latex,line width=0.01pt] (a) -- ($ (a) ! 0.5 ! (d) $);
\draw[-<,>=latex,line width=0.01pt] (b) -- ($ (b) ! 0.5 ! (d) $);

\ifnum #5=1 {
\draw[-<,>=latex,line width=0.01pt] (b) -- ($ (b) ! 0.5 ! (c) $);
\draw[->,>=latex,line width=0.01pt] (c) -- ($ (c) ! 0.5 ! (d) $);
\draw[-<,>=latex,line width=0.01pt] (a) -- ($ (a) ! 0.5 ! (c) $);}
\fi
\ifnum #5=2 {
\draw[-<,>=latex,line width=0.01pt] (b) -- ($ (b) ! 0.5 ! (c) $);
\draw[-<,>=latex,line width=0.01pt] (c) -- ($ (c) ! 0.5 ! (d) $);
\draw[-<,>=latex,line width=0.01pt] (a) -- ($ (a) ! 0.5 ! (c) $);}
\fi
\ifnum #5=3 {
\draw[->,>=latex,line width=0.01pt] (b) -- ($ (b) ! 0.5 ! (c) $);
\draw[-<,>=latex,line width=0.01pt] (c) -- ($ (c) ! 0.5 ! (d) $);
\draw[-<,>=latex,line width=0.01pt] (a) -- ($ (a) ! 0.5 ! (c) $);}
\fi
\ifnum #5=4 {
\draw[->,>=latex,line width=0.01pt] (b) -- ($ (b) ! 0.5 ! (c) $);
\draw[-<,>=latex,line width=0.01pt] (c) -- ($ (c) ! 0.5 ! (d) $);
\draw[->,>=latex,line width=0.01pt] (a) -- ($ (a) ! 0.5 ! (c) $);}
\fi
\end{tikzpicture}
}

\newcommand{\BvTri}[5]{
\begin{tikzpicture}[scale=1,baseline]
\coordinate (a) at(0,0);    
\coordinate (b) at (0.8,-0.75);  
\coordinate (c) at (2.5,0);      
\coordinate (d) at (1.25,2);     
\coordinate (e) at ($ (a) ! 0.3 ! (c) $);
\coordinate (f) at ($ (a) ! 0.45 ! (c) $);
\coordinate (g) at ($ (a) ! 0.5 ! (b) $);
\coordinate (h) at ($ (c) ! 0.5 ! (d) $);
\coordinate (n) at ($ (g) ! 0.5 ! (h) $); 
\draw (a) -- (b) -- (c) -- (d) -- cycle;
\draw (b) -- (d);
\draw (a) -- (e);
\draw (f) -- (c);
\draw[-<,>=latex,line width=0.01pt] (a) -- ($ (a) ! 0.5 ! (b) $);
\draw[-<,>=latex,line width=0.01pt] (b) -- ($ (b) ! 0.5 ! (c) $);
\draw[-<,>=latex,line width=0.01pt] (c) -- ($ (c) ! 0.5 ! (d) $);
\draw[<-,>=latex,line width=0.01pt] (f) -- (c);
\draw[-<,>=latex,line width=0.01pt] (b) -- ($ (b) ! 0.5 ! (d) $);
\draw[-<,>=latex,line width=0.01pt] (a) -- ($ (a) ! 0.5 ! (d) $);

\node[left] at(a) {\scalebox{0.7}{$#1$}};
\node[below] at(b) {\scalebox{0.7}{$#2$}};
\node[right] at(c) {\scalebox{0.7}{$#4$}};
\node[above] at(d) {\scalebox{0.7}{$#3$}};
\node[circle,fill=black,outer sep=0pt, inner sep=0.6pt] at(n) {};
\draw[dashed] (a) -- (n);
\draw[dashed] (b) -- (n);
\draw[dashed] (c) -- (n);
\draw[dashed] (d) -- (n);

\ifnum #5=1 {
  \node[right, outer sep=0pt, inner sep=1pt] at($ (n) ! 0.07 ! (d) $) {\scalebox{0.7}{$#1'$}};
  \draw[->,>=latex,line width=0.01pt] ($ (a) ! 0.4 ! (n) $) -- ($ (a) ! 0.5 ! (n) $) ;
  \draw[->,>=latex,line width=0.01pt] ($ (b) ! 0.4 ! (n) $) -- ($ (b) ! 0.5 ! (n) $) ;
  \draw[->,>=latex,line width=0.01pt] ($ (c) ! 0.4 ! (n) $) -- ($ (c) ! 0.5 ! (n) $) ;
  \draw[->,>=latex,line width=0.01pt] ($ (d) ! 0.4 ! (n) $) -- ($ (d) ! 0.5 ! (n) $) ;}
\fi
\ifnum #5=2 {
  \node[right, outer sep=0pt, inner sep=1pt] at($ (n) ! 0.07 ! (d) $) {\scalebox{0.7}{$#2'$}};
  \draw[-<,>=latex,line width=0.01pt] ($ (a) ! 0.4 ! (n) $) -- ($ (a) ! 0.5 ! (n) $) ;
  \draw[->,>=latex,line width=0.01pt] ($ (b) ! 0.4 ! (n) $) -- ($ (b) ! 0.5 ! (n) $) ;
  \draw[->,>=latex,line width=0.01pt] ($ (c) ! 0.4 ! (n) $) -- ($ (c) ! 0.5 ! (n) $) ;
  \draw[->,>=latex,line width=0.01pt] ($ (d) ! 0.4 ! (n) $) -- ($ (d) ! 0.5 ! (n) $) ;}
\fi
\ifnum #5=3 {
  \node[right, outer sep=0pt, inner sep=1pt] at($ (n) ! 0.07 ! (d) $) {\scalebox{0.7}{$#3'$}};
  \draw[-<,>=latex,line width=0.01pt] ($ (a) ! 0.4 ! (n) $) -- ($ (a) ! 0.5 ! (n) $) ;
  \draw[-<,>=latex,line width=0.01pt] ($ (b) ! 0.4 ! (n) $) -- ($ (b) ! 0.5 ! (n) $) ;
  \draw[->,>=latex,line width=0.01pt] ($ (c) ! 0.4 ! (n) $) -- ($ (c) ! 0.5 ! (n) $) ;
  \draw[->,>=latex,line width=0.01pt] ($ (d) ! 0.4 ! (n) $) -- ($ (d) ! 0.5 ! (n) $) ;}
\fi
\ifnum #5=4 {
  \node[right, outer sep=0pt, inner sep=1pt] at($ (n) ! 0.07 ! (d) $) {\scalebox{0.7}{$#4'$}};
  \draw[-<,>=latex,line width=0.01pt] ($ (a) ! 0.4 ! (n) $) -- ($ (a) ! 0.5 ! (n) $) ;
  \draw[-<,>=latex,line width=0.01pt] ($ (b) ! 0.4 ! (n) $) -- ($ (b) ! 0.5 ! (n) $) ;
  \draw[->,>=latex,line width=0.01pt] ($ (c) ! 0.4 ! (n) $) -- ($ (c) ! 0.5 ! (n) $) ;
  \draw[-<,>=latex,line width=0.01pt] ($ (d) ! 0.4 ! (n) $) -- ($ (d) ! 0.5 ! (n) $) ;}
\fi
\end{tikzpicture}
}

\newcommand{\TriVertex}[3]{
  {
  \begin{tikzpicture}[scale=0.6,baseline]
  \tikzstyle{every node}=[font=\small]
  \draw  [->,>=latex] (0,0) -- (1,0.5);
  \draw  [->,>=latex] (2,0) -- (1,0.5);
  \draw  [->,>=latex] (1,1.5) -- (1,0.5);
  \node at(0.4,0.6) {$#1$};
  \node at(1.8,0.6) {$#2$};
  \node at(1.3,1.2) {$#3$};
  \end{tikzpicture}
  }}

\newcommand{\YYa}[5]{
  {
  \begin{tikzpicture}[scale=#1,baseline]
  \tikzstyle{every node}=[font=\small]
  \draw  (-1,-1) -- (0,0);
  \draw  (1,-1) -- (0,0);
  \draw  (0,0) -- (0,2);
  \draw  (0,2) -- (-1,3);
  \draw  (0,2) -- (1,3);
  \node[gray] at (-1,1) {$#2$};
  \node[gray] at (1,1) {$#3$};
  \node[gray] at (0,-1) {$#4$};
  \node[gray] at (0,3) {$#5$};
  \end{tikzpicture}
  }}

\newcommand{\YYb}{
  {
  \begin{tikzpicture}[scale=0.5,baseline]
  \tikzstyle{every node}=[font=\small]
  \draw  (-1,-1) -- (0,0);
  \draw  (1,-1) -- (0,0);
  \draw[->,>=latex]  (0,0) -- (0,2);
  \draw  (0,2) -- (-1,3);
  \draw  (0,2) -- (1,3);
  \node[gray] at (-1.5,1) {$v_1$};
  \node[gray] at (1.5,0.8) {$v_2$};
  \node at(1.1,1.5) {$[v_1v_2]$};
  \node[gray] at (0,-1) {$v_3$};
  \node[gray] at (0,3) {$v_4$};
  \end{tikzpicture}
  }}

\newcommand{\YYc}{
  {
  \begin{tikzpicture}[scale=0.5,baseline]
  \tikzstyle{every node}=[font=\small]
  \draw  (-1,-1) -- (0,0);
  \draw  (1,-1) -- (0,0);
  \draw[<-,>=latex]  (0,0) -- (0,2);
  \draw  (0,2) -- (-1,3);
  \draw  (0,2) -- (1,3);
  \node[gray] at (-1.5,1) {$v_1$};
  \node[gray] at (1.5,1.2) {$v_2$};
  \node at (1.1,0.5) {$[v_2v_1]$};
  \node[gray] at (0,-1) {$v_3$};
  \node[gray] at (0,3) {$v_4$};
  \end{tikzpicture}
  }}

\newcommand{\HorGaugeTransformYY}[4]{
  {
  \begin{tikzpicture}[scale=#1,baseline]
  \tikzstyle{every node}=[font=\small]
  \draw [->,>=latex,line width=0.01pt] (-1,-1) -- (0,0);
  \draw [->,>=latex,line width=0.01pt] (-1,1) -- (0,0);
  \draw [->,>=latex,line width=0.01pt] (0,0) -- (2,0);
  \draw [->,>=latex,line width=0.01pt] (2,0) -- (3,-1);
  \draw [->,>=latex,line width=0.01pt] (2,0) -- (3,1);
  \node[gray] at (-1,0) {$#2$};
  \node[gray] at (3,0) {$#3$};
  \node at (1,1) {$#4$};
  \end{tikzpicture}
  }}

\newcommand{\torusgraph}{
\begin{tikzpicture}[scale=1,baseline]
\coordinate (p1) at (0,0);
\coordinate (p2) at (0,1);
\coordinate (p3) at (1,0);
\coordinate (p4) at (1,1);
\draw [->-,>=latex,line width=0.02pt] (p2) -- (p1);
\draw [->-,>=latex,line width=0.02pt] (p4) -- (p3);
\draw [->-,>=latex,line width=0.02pt] (p3) -- (p1);
\draw [->-,>=latex,line width=0.02pt] (p4) -- (p2);
\draw [->-,>=latex,line width=0.02pt] (p4) -- (p1);
\node at (-0.15,-0.15) {1};
\node at (-0.15,1.15) {2};
\node at (1.15,-0.15) {3};
\node at (1.15,1.15) {4};
\node at (0.5,-0.15) {$g$};
\node at (0.5,1.15) {$g$};
\node at (-0.15,0.5) {$h$};
\node at (1.15,0.5) {$h$};
\end{tikzpicture}
}

\newcommand{\torusgraphST}[9]{
\begin{tikzpicture}[scale=1,baseline]
\coordinate (p1) at (0,0);
\coordinate (p2) at (0,1);
\coordinate (p3) at (1,0);
\coordinate (p4) at (1,1);
\draw [->-,>=latex,line width=0.02pt] (p2) -- (p1);
\draw [->-,>=latex,line width=0.02pt] (p4) -- (p3);
\draw [->-,>=latex,line width=0.02pt] (p3) -- (p1);
\draw [->-,>=latex,line width=0.02pt] (p4) -- (p2);

\ifnum #9=0 {
    \ifnum #1<#7
       \draw [->-,>=latex,line width=0.02pt] (p4) -- (p1);
    \else   \draw [-<-,>=latex,line width=0.02pt] (p4) -- (p1);
    \fi }
\else {
    \ifnum #3<#5
       \draw [->-,>=latex,line width=0.02pt] (p3) -- (p2);
    \else   \draw [-<-,>=latex,line width=0.02pt] (p3) -- (p2);
    \fi }
\fi
\node at (-0.15,-0.15) {$#1#2$};
\node at (-0.15,1.15) {$#3#4$};
\node at (1.15,-0.15) {$#5#6$};
\node at (1.15,1.15) {$#7#8$};
\end{tikzpicture}
}

\newcommand{\SLtransformationSone}{
\begin{tikzpicture}[scale=0.8,baseline]
\coordinate (p1) at (0,0);
\coordinate (p2) at (0,1);
\coordinate (p3) at (1,0);
\coordinate (p4) at (1,1);
\draw [->,>=latex,line width=0.02pt] (-1.5,0) -- (1.5,0);
\draw [->,>=latex,line width=0.02pt] (0,-0.5) -- (0,1.5);
\draw [->-,>=latex,line width=0.02pt] (p2) -- (p1);
\draw [->-,>=latex,line width=0.02pt] (p4) -- (p3);
\draw [->-,>=latex,line width=0.02pt] (p3) -- (p1);
\draw [->-,>=latex,line width=0.02pt] (p4) -- (p2);
\draw [->-,>=latex,line width=0.02pt] (p4) -- (p1);
\node at (-0.15,-0.15) {1};
\node at (-0.15,1.15) {2};
\node at (1.15,-0.15) {3};
\node at (1.15,1.15) {4};
\end{tikzpicture}
}

\newcommand{\SLtransformationStwo}{
\begin{tikzpicture}[scale=0.8,baseline]
\coordinate (p1) at (0,0);
\coordinate (p2) at (-1,0);
\coordinate (p3) at (0,1);
\coordinate (p4) at (-1,1);
\draw [->,>=latex,line width=0.02pt] (-1.5,0) -- (1.5,0);
\draw [->,>=latex,line width=0.02pt] (0,-0.5) -- (0,1.5);
\draw [->-,>=latex,line width=0.02pt] (p2) -- (p1);
\draw [->-,>=latex,line width=0.02pt] (p4) -- (p3);
\draw [->-,>=latex,line width=0.02pt] (p3) -- (p1);
\draw [->-,>=latex,line width=0.02pt] (p4) -- (p2);
\draw [->-,>=latex,line width=0.02pt] (p4) -- (p1);
\node at (0.15,-0.15) {1};
\node at (-1.15,-0.15) {2};
\node at (0.15,1.15) {3};
\node at (-1.15,1.15) {4};
\end{tikzpicture}
}

\newcommand{\SLtransformationTone}{
\begin{tikzpicture}[scale=0.8,baseline]
\coordinate (p1) at (0,0);
\coordinate (p2) at (0,1);
\coordinate (p3) at (1,0);
\coordinate (p4) at (1,1);
\draw [->,>=latex,line width=0.02pt] (-0.5,0) -- (2.5,0);
\draw [->,>=latex,line width=0.02pt] (0,-0.5) -- (0,1.5);
\draw [->-,>=latex,line width=0.02pt] (p2) -- (p1);
\draw [->-,>=latex,line width=0.02pt] (p4) -- (p3);
\draw [->-,>=latex,line width=0.02pt] (p3) -- (p1);
\draw [->-,>=latex,line width=0.02pt] (p4) -- (p2);
\draw [->-,>=latex,line width=0.02pt] (p4) -- (p1);
\node at (-0.15,-0.15) {1};
\node at (-0.15,1.15) {2};
\node at (1.15,-0.15) {3};
\node at (1.15,1.15) {4};
\end{tikzpicture}
}

\newcommand{\SLtransformationTtwo}{
\begin{tikzpicture}[scale=0.8,baseline]
\coordinate (p1) at (0,0);
\coordinate (p2) at (1,1);
\coordinate (p3) at (1,0);
\coordinate (p4) at (2,1);
\draw [->,>=latex,line width=0.02pt] (-0.5,0) -- (2.5,0);
\draw [->,>=latex,line width=0.02pt] (0,-0.5) -- (0,1.5);
\draw [->-,>=latex,line width=0.02pt] (p2) -- (p1);
\draw [->-,>=latex,line width=0.02pt] (p4) -- (p3);
\draw [->-,>=latex,line width=0.02pt] (p3) -- (p1);
\draw [->-,>=latex,line width=0.02pt] (p4) -- (p2);
\draw [->-,>=latex,line width=0.02pt] (p4) -- (p1);
\node at (-0.15,-0.15) {1};
\node at (0.8,1.15) {2};
\node at (1.15,-0.15) {3};
\node at (2.15,1.15) {4};
\end{tikzpicture}
}

\newcommand{\Ygraph}[4][1]{
\begin{tikzpicture}[scale=0.6,baseline]
  \draw [->,>=stealth',line width=0.01pt] (30:0.1) -- (0,0) ; 
    \draw (30:1) -- (0,0) ; 
  \draw [->,>=stealth',line width=0.01pt] (150:0.1) -- (0,0); 
    \draw (150:1) -- (0,0); 
  \node at(-0.5,0.01) {\scalebox{0.8}{$#2$}};
  \node at(0.5,0.01) {\scalebox{0.8}{$#4$}};
  \ifnum #1=1 {
   \draw [->,>=stealth',line width=0.01pt] (0,-1/20) -- (0,0); 
   \node at(-0.2,-0.5) {\scalebox{0.8}{$#3$}};
  }
  \else{
    \draw [<-,>=stealth',line width=0.01pt] (0,-1/2) -- (0,0); 
    \node at(-0.4,-0.5) {\scalebox{0.8}{$-#3$}};
    }
  \fi
    \draw (0,-1) -- (0,0); 
  \end{tikzpicture}
  }

\newcommand{\Hgraph}[3][1]{
  \begin{tikzpicture}[scale=0.6,baseline]
    \coordinate (c) at (0,0);
    \coordinate (l) at (-0.7, 0);
    \coordinate (r) at ($ (c) ! -1 ! (l) $);
    \coordinate (ul) at (-0.95,0.75);
    \coordinate (lr) at ($ (c) ! -1 ! (ul) $);
    \coordinate (ll) at (-0.95,-0.75);
    \coordinate (ur) at ($ (c) ! -1 ! (ll) $);
    \draw (l) -- (r) ; 
    \draw (ul) -- (l); 
    \draw (ll) -- (l); 
    \draw (lr) -- (r); 
    \draw (ur) -- (r); 
    \ifnum #1=1 {
    \node[right] at($ (ul) ! .3 ! (l) $) {\scalebox{0.7}{$#2_1$}};
    \node[right] at($ (ll) ! .2 ! (l) $) {\scalebox{0.7}{$#2_2$}};
    \node[left] at($ (lr) ! .2 ! (r) $) {\scalebox{0.7}{$#2_3$}};
    \node[left] at($ (ur) ! .3 ! (r) $) {\scalebox{0.7}{$#2_4$}};
    \node[below] at($ (c) ! .15 ! (0,0.5) $) {\scalebox{0.7}{$#3$}};}
    \fi
    \draw[<-,>=stealth', line width=0.01pt] (l) -- (r) ;
    \draw[->,>=stealth', line width=0.01pt] (ul) -- (l);
    \draw[->,>=stealth', line width=0.01pt] (ll) -- (l);
    \draw[->,>=stealth', line width=0.01pt] (lr) -- (r);
    \draw[->,>=stealth', line width=0.01pt] (ur) -- (r);
   \end{tikzpicture}
  }

 \newcommand{\Xgraph}[3][1]{
  \begin{tikzpicture}[scale=0.6,baseline]
    \coordinate (c) at (0,0);
    \coordinate (u) at (0, 0.6);
    \coordinate (d) at ($ (c) ! -1 ! (u) $);
    \coordinate (ul) at (-0.85,0.85);
    \coordinate (lr) at ($ (c) ! -1 ! (ul) $);
    \coordinate (ll) at (-0.85,-0.85);
    \coordinate (ur) at ($ (c) ! -1 ! (ll) $);
    \draw (u) -- (d) ; 
    \draw (ul) -- (u); 
    \draw (ll) -- (d); 
    \draw (lr) -- (d); 
    \draw (ur) -- (u); 
    \ifnum #1=1 {
    \node[below] at($ (ul) ! .3 ! (u) $) {\scalebox{0.7}{$#2_1$}};
    \node[above] at($ (ll) ! .2 ! (d) $) {\scalebox{0.7}{$#2_2$}};
    \node[above] at($ (lr) ! .2 ! (d) $) {\scalebox{0.7}{$#2_3$}};
    \node[below] at($ (ur) ! .2 ! (u) $) {\scalebox{0.7}{$#2_4$}};
    \node[right] at($ (-0.5,0) ! 0.8 ! (c) $) {\scalebox{0.7}{$#3$}}; }
    \fi
    \draw[<-,>=stealth', line width=0.01pt] (u) -- (d) ;
    \draw[->,>=stealth', line width=0.01pt] (ul) -- (u);
    \draw[->,>=stealth', line width=0.01pt] (ll) -- (d);
    \draw[->,>=stealth', line width=0.01pt] (lr) -- (d);
    \draw[->,>=stealth', line width=0.01pt] (ur) -- (u);
   \end{tikzpicture}
  }

\newcommand{\Psix}[3][1]{
\begin{tikzpicture}[scale=0.8]
\node[name=s, regular polygon, regular polygon sides=6, minimum size=1cm, outer sep=0pt ,draw] at (0,0) {}; 
\foreach \anchor/\x/\y /\xx/\yy /\b in
{corner 1/0.17/0.17*1.732/-0.11/0.18/1, corner 2/-0.17/0.17*1.732/0.07/0.18/2, corner 3/-0.34/0/-0.15/-0.18/3, corner 4/-0.17/-0.17*1.732/-0.22/-0.05/4, corner 5/0.17/-0.17*1.732/0.2/-0.05/5, corner 6/0.34/0/0.15/-0.18/6}
{
 \draw[shift=(s.\anchor)] (0,0) -- (\x,\y) node at(\xx,\yy) {$#2_{\text{\scalebox{0.7}{$\b$}}}$};
 \ifnum #1=1
 \draw[shift=(s.\anchor),<-,>=stealth', line width=0.01pt] (s.\anchor) -- (\x,\y);
 \fi
 }
\foreach \anchor/\xx/\yy /\a in
{side 1/0/-0.18/1, side 2/-0.18/0.05/2, side 3/0.15/0.05/3, side 4/0/-0.18/4, side 5/-0.18/0.05/5, side 6/0.15/0.05/6}
 \draw[shift=(s.\anchor)]  node at(\xx,\yy) {$#3_{\text{\scalebox{0.7}{$\a$}}}$};
\ifnum #1=1{
  \foreach \anchorr/\anchorf in
   {corner 1/corner 2, corner 2/corner 3, corner 3/corner 4, corner 4/corner 5, corner 5/corner 6, corner 6/corner 1}
   \draw[shift=(s.\anchorr), ->, >=stealth', line width=0.01pt]  (s.\anchorr) -- (s.\anchorf);}
 \else {
  \foreach \anchorb/\anchorw in
   {corner 1/corner 2, corner 3/corner 4, corner 5/corner 6} {
   \node[fill=black, circle, minimum size=2.5, inner sep=0, outer sep=0, draw] at(s.\anchorb) {};
   \node[fill=white, circle, minimum size=2.5, inner sep=0, outer sep=0, draw] at(s.\anchorw) {};}
}
\fi
\end{tikzpicture}
}

\begin{document}

\title{Twisted Quantum Double Model of Topological Phases in Two Dimensions}

\author{Yuting Hu}
\email{yuting@physics.utah.edu}
\affiliation{Department of Physics and Astronomy,
University of Utah, Salt Lake City, UT 84112, USA}

\author{Yidun Wan}
\email{ywan@meso.t.u-tokyo.ac.jp}
\affiliation{Department of Applied Physics, Graduate School of Engineering, University of Tokyo,\\
Hongo 7-3-1, Bunkyo-ku, Tokyo 113-8656, Japan}
\author{Yong-Shi Wu}
\email{wu@physics.utah.edu} \affiliation{Key State Laboratory of 
Surface Physics, Department of Physics and \\
Center for Field Theory and Particle Physics, Fudan University,
Shanghai 200433, China} \affiliation{Department of Physics and
Astronomy, University of Utah, Salt Lake City, UT 84112, USA}

\date{\today}

\begin{abstract}
We propose a new discrete model---the twisted quantum double model---of 2D topological phases based on a finite group $G$ and a $3$--cocycle $\alpha$ over $G$. The detailed properties of the ground states are studied, and we find that the ground--state subspace can be characterized in terms of the twisted quantum double $D^{\alpha}(G)$ of $G$. When $\alpha$ is the trivial $3$--cocycle, the model becomes Kitaev's quantum double model based on the finite group $G$, in which the elementary excitations are known to be classified by the quantum double $D(G)$ of $G$. Our model can be viewed as a Hamiltonian extension of the Dijkgraaf--Witten topological gauge theories to the discrete graph case with gauge group being a finite group. We also demonstrate a duality between a large class of Levin-Wen string-net models and certain twisted quantum double models, by mapping the string--net $6j$ symbols to the corresponding $3$--cocycles.
The paper is presented in a way such that it is accessible to a wide range of physicists.\end{abstract}
\maketitle

\section{Introduction}\label{sec:intro}
The study of possible phases of matter has gone beyond Landau's paradigm of symmetry breaking for decades, which leads to the discovery of topological phases of matter. Among all possible topological phases, there are a class of them that are believed to bear intrinsic \emph{topological order}\cite{Wen1989}, in which they display features such as robust ground state degeneracies, Abelian or non--Abelian braiding
(anyonic) statistics of quasi--particle excitations, and in many cases protected edge excitations. The classic examples of these phases include the (fractional) quantum Hall states, $\Z_2$ spin liquids, chiral spin liquids, and $p+\ii p$ superconductors\cite{Tsui1982,Laughlin1983,Read1991,Wen1991,Moessner2001,Kalmeyer1987,Wen1989a,Read2000,Gurarie2007}. The physical characteristics of topological phases urge the search of the mathematical structures that classify the topological phases. It is then natural to resort to certain theoretical models that can yield various topological phases.

There is a very general framework---the string--net models\cite{Levin2004}, also known as the Levin--Wen models---supplying exactly soluble models that incorporates a large class of intrinsically topological phases, notably those preserving time--reversal symmetry. Although it is believed that tensor category theory is the mathematical framework that underlies these models, a general classification of these models---in particular of the topological phases they describe---is yet to be found.

The intrinsically, topologically ordered systems are roughly speaking those gapped quantum phases of matter that involve long range entanglement (LRE). In contrast, there are gapped quantum phases of matter that involve short range entanglement (SRE), which, when symmetry is unbroken, give rise to nontrivial phases, called symmetry--protected topological (SPT) phases\cite{Gu2009,Pollmann2012}, such as the Haldane phase on one--dimensional spin chain\cite{Haldane1983} and topological insulators\cite{Kane2005,Kane2005a,Bernevig2006,Moore2007,Fu2007,Qi2008}. Characteristic properties of these phases are usually non--degenerate ground states and, if the system has a boundary, nontrivial edge excitations.

Very recently, however, it is discovered that a specific SPT phase, namely an Ising spin model with a gauged  $\Z_2$ symmetry, admits a dual LRE phase described by a string net model \cite{Gu2012}. This remarkable duality is then conjectured\cite{Gu2012} to exist between a general SPT phase with discrete, gauged symmetry $G$ and a string net model with fusion rules also given by the product rule of $G$. Soon after, this conjecture is confirmed in Ref\cite{Hung2012}, which henceforth implies that the classification of a large class of SPT phases provided by group cohomology in 2+1 dimensions via $H^3(G,U(1))$ described in Ref\cite{Chen2011e} (We remark that Ref\cite{Lu2012} offers a field theoretic approach that obtain the same classification.) indirectly provide classifications of the corresponding string net models.

This classification of string--net models seems feasible, as the building blocks of these models, namely the $6j$--symbols may fall into equivalence classes that are related to the $3$--cocycles in the cohomology group $H^3(G,U(1))$ of the symmetry group $G$ of the model\cite{Hung2012}. Nevertheless, in the string--net models that have been studied so far, the $6j$ symbols are assumed to respect the full tetrahedral symmetry, which may be too restrictive for a description of  topological phases. Namely, as pointed out in Ref\cite{Hung2012}, the topological phases described by the Levin--Wen model with tetrahedral symmetry may not account for all topological phases classified by $H^3(G,U(1))$.

This has motivated us to propose a new class of discrete models for 2D topological phases, called the twisted quantum double model for reasons to be clear later, whose construction involves a 3-cocycle, an element in the group cohomology group $H^3(G,U(1))$. More precisely we consider a model on a planar graph of triangles, each edge of which is graced with a group element of a finite group $G$. The Hamiltonian of the model has matrix elements constructed by a $3$--cocycle $\alpha$ belonging to the cohomology group $H^3(G,U(1))$ of $G$. We require that $\alpha$ satisfies only the $3$--cocycle condition $\delta\alpha=1$ where $\delta$ is the coboundary operator, which under the circumstance of this paper is actually the pentagon identity in disguise. Owing to the absence of extra conditions put in by hand on $\alpha$, all solutions to the $3$--cocycle condition but one---namely the trivial $3$--cocycle---do not respect the tetrahedral symmetry. In other words, any element of $H^3(G,U(1))$ defines an instance of our model.

We study our model in detail by placing it on a torus. In terms of $3$--cocycles we construct, for the ground states of our model, explicitly three topological observables, namely, the ground state degeneracy (GSD), the $\str$ and $\ttr$ operators that are a representation of the generators of the modular group $SL(2,\Z)$. This construction is a new result of ours, purely based on our model and in terms of the $3$-cocycles of $G$, without using group representation theory. These topological observables on the ground states lead to a set of topological numbers, respectively formed by the GSD, elements of $\str$-matrix, and topological spin. We show that these topological numbers depend on the cohomology classes $[\alpha]\in H^3(G,U(1))$. Moreover, equivalent $3$--cocycles define equivalent twisted quantum double models, in the sense that their Hamiltonians can be continuously deformed into each other. We present a few characteristic properties of the topological numbers, which may help to resolve this open question in future work. On top of these abstract constructions, we work out a few concrete examples for certain finite groups, Abelian and non--Abelian.

We also discourse on how our model relates to topological field theories and models  of topological phases. It turns out that our model is a reasonable Hamiltonian extension of the Dijkgraaf--Witten theory\cite{Dijkgraaf1989,Dijkgraaf1990,Propitius1995} of topological Chern--Simons gauge theory in three dimensions, as we can identify the ground states of our model defined by an $[\alpha]\in H^3(G,U(1))$ on the boundary of a three-manifold with the gauge--invariant boundary states of the Dijkgraaf--Witten theory defined by the same $[\alpha]$ in the bulk, which then equates the GSD of our model with the partition function of the corresponding Dijkgraaf--Witten theory. Since three-dimensional topological Chern-Simons theory corresponds to two-dimensional rational conformal field theory (RCFT)\cite{Dijkgraaf1990}, a connection between our model and RCFT is thus established. In particular, the GSD of our model with group $G$ agrees with the number of primary fields in the RCFT that an orbifold by the symmetry group $G$ of a holomorphic CFT.

We demonstrate that our twisted quantum double model reduces precisely to Kitaev's quantum double model in the special case where the defining $3$-cocycle is trivial. The nontrivial $3$-cocycles in our model may twist the usual group algebra $\mathbb{C}[G]$ into a twisted group algebra, which mainly motivates the name of our model.

As our model is motivated by the Levin--Wen model, we demonstrate a duality between a large class of Levin--Wen string--net models and certain twisted quantum double models, by mapping the string--net $6j$ symbols to the corresponding $3$--cocycles.

We would like to insert as an aside here that we minimized the complexity of the mathematics in this paper without sacrificing the preciseness and comprehensibility of our presentation. For instance, although group cohomology is a key word of this paper, but we assume zero prior knowledge of it because we define and present the $n$--cocycles as merely $U(1)$ functions that satisfy an algebraic condition. As such, we believe the paper is accessible to a wide range of physicists and mathematicians.

Our paper is organized as follows. In Section \ref{sec:model} we construct our new model of topological phases. Section \ref{sec:topoOb} is devoted to the general setting for the topological observables. In Section \ref{sec:GSD} we compute the ground state degeneracy (GSD) on a torus and study the corresponding topological degrees of freedom. Section \ref{sec:fractionaltopologicalnumbers} furnishes the construction of two more topological observables that give rise to fractional  topological numbers. We present a classification of the topological numbers in our model in In Section \ref{sec:classification}. Section \ref{sec:examples} offers concrete examples of our model for a number of finite groups. The next three Sections \ref{sec:Kitaev}, \ref{sec:relation2DW}, and \ref{sec:LW} relate our model respectively to Kitaev's quantum double model, Dijkgraaf--Witten topological gauge theory, and Levin--Wen string--net model. The final section (Section \ref{sec:disc}) concludes with remarks and outlook.
Appendix \ref{app:HnGU1} introduces very briefly the group cohomology of finite groups, while the other appendices collect proofs of various statements in the paper.

\section{The model}\label{sec:model}
In this section, we shall construct our model in $(2+1)$--dimension, as an exactly--soluble Hamiltonian on the Hilbert space spanned by planar graphs consisting of triangles whose edges are graced with group elements in certain finite group.
\subsection{Basic Ingredients}\label{subsec:basic}

The model is defined on a two--dimensional graph $\Gamma$ consisting of triangles only (Fig. \ref{fig:GraphConfiguration}). Such a  graph does not have any open edge and may be thought as a simplicial triangulation of certain two-dimensional Riemannian surface, e.g., a sphere; however, in this model, we shall take the graph as abstract without referring to its topological background except when we compare the model with other models, such as Dijkgraaf--Witten discrete topological gauge theories. Note that Fig. \ref{fig:GraphConfiguration} is a crop of one such graph, so the open edges in the figure are not really open. We enumerate the vertices of $\Gamma$ by any ordered set of labels. The enumerations of the vertices we choose is irrelevant as long as their relative order remains consistent during the calculation.
\begin{figure}[!ht]\centering
  \includegraphics[scale=0.6]{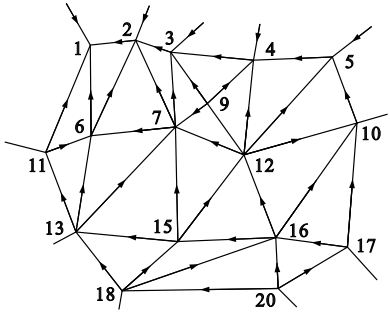}
  \caption{A portion of a graph that represent the basis vectors in the Hilbert space. Each edge carries an arrow and is assigned a group element denoted by $[ab]$ with $a<b$.}
  \label{fig:GraphConfiguration}
\end{figure}

The model is characterized by a triple $(H,G,\alpha)$, which can be denoted by $H_{G,\alpha}$ for short. The first in the triple is the Hamiltonian $H$. The second ingredient $G$  is  a finite group. Each edge of $\Gamma$ is graced with a group element of $G$. The Hilbert space is spanned by the configurations of group elements
on the edges of $\Gamma$. Each edge (see Fig. \ref{fig:GraphConfiguration}) carries an arrow that goes from the vertex with a larger label to the one with a smaller label.
To each edge $e$ of the graph $\Gamma$, we assign
a group element $g_e \in G$, and all possible assignments form
the basis vectors of the Hilbert space.
\begin{equation}
  \label{IndexByLink}
  \left\{g_1,g_2,...,g_E\right\}
\end{equation}
where $E$ is the total number
of edges in $\Gamma$.

It convenient to denote both an edge and the group element on the edge by simply $[ab]$ with $a<b$ the two boundary vertices of the edge. It is understood that $[ba]=[ab]^{-1}$. The inner product of the Hilbert space is the obvious one:
\begin{align}
  \label{eq:innerProduct}
  \Blangle
  \oneTriangle[1]{a'}{b'}{c'}
  \Bvert
  \oneTriangle[1]{a}{b}{c}
  \Brangle
  =& \delta_{[ab][a'b']}\delta_{[bc][b'c']}\delta_{[ac][a'c']}\nonumber\\
  &\dots\; ,
\end{align}
where only one triangle in $\Gamma$ is drawn, and the \textquotedblleft$\dots$" omits the $\delta$--functions on all other triangles that are not shown. Note that three group elements on the three sides of a triangle, e.g., the $[ab],[bc]$ and $[ac]$ on the RHS of Eq. (\ref{eq:innerProduct}), are independent of each other in general, i.e., $[ab]\cdot[bc]\neq[ac]$.
From now on, we shall neglect the group elements on the edges but keep only the vertex labels when we draw a basis vector.
\begin{figure}[h!]
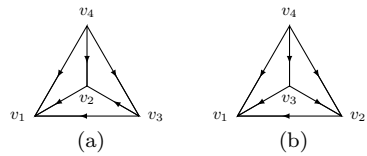

\centering
\subfigure[]{\threeTriangles{v_1}{v_3}{v_2}{v_4}{3} \label{fig:3cocycleA}}
\subfigure[]{\threeTriangles{v_1}{v_2}{v_3}{v_4}{2} \label{fig:3cocycleB}}
\caption{(a) The defining graph of the $3$--cocycle $\alpha([v_1v_2], [v_2v_3],[v_3v_4])$. (b) For $\alpha([v_1v_2], [v_2v_3],[v_3v_4])^{-1}$.}
\label{fig:3cocycle}
\end{figure}

The third element is a \textit{normalized $3$--cocycle} $\alpha\in H^3(G,U(1))$, i.e., a function
$\alpha:G^3\rightarrow U(1)$ that satisfies the \textit{3-cocycle
condition}
\begin{align}
  \label{3CocycleCondition}
  &\alpha(g_1,g_2,g_3)\alpha(g_0\cdot g_1,g_2,g_3)^{-1}\times\\
  &\quad\alpha(g_0,g_1\cdot g_2,g_3)\alpha(g_0,g_1,g_2\cdot g_3)^{-1} \alpha(g_0,g_1,g_2)=1\nonumber
\end{align}
for all $g_i \in G$, and satisfies the \textit{normalization condition}
\begin{align}
  \label{NormalizationCondition}
  \alpha(1,g,h)=\alpha(g,1,h)=\alpha(g,h,1)=1,
\end{align}
whenever $g,h\in G$ are arbitrary. A basic and brief introduction to cohomology groups $H^n(G,U(1))$ of finite groups is found in Appendix \ref{app:HnGU1}. We emphasize that this normalization condition is not an \textit{ad hoc} condition we imposed as an extra on the $3$--cocycles; rather, it is a natural condition that any group $3$--cocycle can satisfy for the following reason. A $3$--cocycle $\alpha$ is in fact an equivalence class of the $3$--cocycles that can be scaled into each other by merely a $3$--coboundary $\delta\beta$, where $\beta$ is a $2$--cochain. It can be shown that for any equivalence class of $3$--cocycles, there always exists a representative that meets the normalization condition in Eq. (\ref{NormalizationCondition}), which is in turn justified.

Note that every group has a trivial 3-cocycle $\alpha_0\equiv 1$ on the entire $G$. One can define a $3$--cocycle on any subgraph composed of three triangles, which share a vertex and any two of which share an edge. Consider Fig. \ref{fig:3cocycleA} as an example: The four vertices are in the order $v_1<v_2<v_3<v_4$; we define the $3$--cocycle for this subgraph by taking its three variables from left to right to be the three group elements, $[v_1v_2]$, $[v_2v_3]$ and $[v_3v_4]$, which are along the path from the least vertex $v_1$ to the greatest vertex $v_4$ passing $v_2$ and $v_3$ in order; hence, the $3$--cocycle reads $\alpha([v_1v_2],[v_2v_3],[v_3v_4])$.
If one lifts the vertex $v_2$ in Fig. \ref{fig:3cocycleA} above the paper plane, the three triangles turns out to be on the surface of a tetrahedron. In this sense, one can think of the $3$--cocycle as associated with a tetrahedron as well, which is useful when the graph is really interpreted as the triangulation of a Riemannian surface.

On the other hand, if one switches the vertices $v_2$ and $v_3$ in Fig. \ref{fig:3cocycleA}, one obtains Fig. \ref{fig:3cocycleB}, which defines the inverse $3$--cocycle $\alpha([v_1v_2],[v_2v_3],[v_3v_4])^{-1}$. Whether a graph defines a $3$--cocycle $\alpha$ or the inverse $\alpha^{-1}$ depends on the orientation of the four vertices in the graph by the following rule. One first reads off a list of the three vertices counter--clockwise from any of the three triangles of the defining graph of the $3$--cocycle, e.g., $(v_2,v_3,v_4)$ from Fig. \ref{fig:3cocycleA} and $(v_3,v_2,v_4)$ from Fig. \ref{fig:3cocycleB}. One then append the remaining vertex to the beginning of the list,  e.g., $(v_1,v_2,v_3,v_4)$ from Fig. \ref{fig:3cocycleA} and $(v_1,v_3,v_2,v_4)$ from Fig. \ref{fig:3cocycleB}. If the list can be turned into ascending order by even permutations, such as $(v_1,v_2,v_3,v_4)$ from Fig. \ref{fig:3cocycleA}, one has an $\alpha$ but an $\alpha^{-1}$ otherwise, as by $(v_1,v_3,v_2,v_4)$ from Fig. \ref{fig:3cocycleB}.

We would like to warn the reader of some abuse of language in the rest of the paper. For example, when we say ``a $3$--cocycle", we may refer to a class $[\alpha]$, a representative $\alpha$, or the evaluation of $\alpha$ on a tetrahedron. For another example, although there is abstractly only one $3$--cocycle condition as in Eq. (\ref{3CocycleCondition}), we may sometimes mean $3$--cocycle conditions by the evaluation of the condition on different tetrahedra. But all such and such should not cause any confusion contextually.
\subsection{The Hamiltonian}\label{subsec:Hamiltonian}
The $3$--cocycles will appear in the matrix elements of the model's Hamilton defined as follows.
\be\label{eq:Hamiltonian}
  H=-\sum_v A_v-\sum_f B_f,
\ee
where $B_f$ is the face operator defined at each triangular face $f$, and $A_v$ is the vertex operator defined on each vertex $v$. As we shall see later, this Hamiltonian is formally the same as and generalizes that of the Kitaev model\cite{Kitaev2003a,Kitaev2006}, where an operators $A_v$ behaves as a gauge transformation on the group elements respectively on the edges meeting at $v$, and a $B_f$ detects whether the flux through face $f$ is zero. This kind of Hamiltonians generically feature ground states that are gauge invariant and bear zero flux everywhere. We now elaborate more on these operators.

The action of $B_f$ on a basis vector is
\begin{align}
  \label{eq:actionOfBf}
  B_f\BLvert \oneTriangle{v_1}{v_2}{v_3} \Brangle
  =\delta_{[v_1v_2]\cdot[v_2v_3]\cdot[v_3v_1]}
  \BLvert \oneTriangle{v_1}{v_2}{v_3}\Brangle .
\end{align}
The discrete
delta function $\delta_{[v_1v_2]\cdot[v_2v_3]\cdot[v_3v_1]}$ is unity if ${[v_1v_2]\cdot[v_2v_3]\cdot[v_3v_1]=1 }$, where $1$ is the identity element in $G$, and 0 otherwise. Note again that here, the ordering of $v_1,v_2$, and $v_3$ does not matter because of the identities
$\delta_{[v_1v_2]\cdot[v_2v_3]\cdot[v_3v_1]}
=\delta_{[v_3v_1]\cdot[v_1v_2]\cdot[v_2v_3]}$ and
$\delta_{[v_1v_2]\cdot[v_2v_3]\cdot[v_3v_1]}
=\delta_{\{[v_1v_2]\cdot[v_2v_3]\cdot[v_3v_1]\}^{-1}}
=\delta_{[v_3v_1]^{-1}\cdot[v_2v_3]^{-1}\cdot[v_1v_2]^{-1}}
=\delta_{[v_1v_3]\cdot[v_3v_2]\cdot[v_2v_1]}$.
In other words,
in any state on which $B_f=1$ on a triangular face $f$, the three group degrees of freedom
around $v$ is related by a \textit{chain
rule}:
\begin{equation}
\label{ChainRule}
[v_1v_3]=[v_1v_2]\cdot[v_2v_3]
\end{equation}
for any enumeration
$v_1,v_2,v_3$ of the three vertices of the face $f$.

The operator $A_v$ is a summation
\begin{equation}
  \label{eq:Av}
  A_v=\frac{1}{|G|}\sum_{[vv']=g\in G}A_v^g,
\end{equation}
which deserves explanation. The value $|G|$ is the order of the group $G$. The operator $A_v^g$ acts on
a vertex $v$ with a group element $g\in G$ by replacing $v$ by a new enumeration $v'$ that is less than $v$ but greater than all the enumerations that are less than $v$ in the original set of enumerations before the action of the operator, such that $[v'v]=g$. $A_v^g$ does not affect any vertex other than $v$ but introduces a $U(1)$ phase, composed of $3$--cocycles determined by $v'$ and all the vertices adjacent to $v$ before the action, to the resulted state. In a dynamical language, $v'$ is understood as on the next \textquotedblleft time" slice, and there is an edge $[v'v]\in G$ in the $(2+1)$ dimensional \textquotedblleft spacetime" picture. Consider a trivalent vertex as an example (see Eq. (\ref{eq:Avg})). Without loss of generality, we assume that the enumerations of the four vertices are in the order $v_1<v_2<v_3<v_4$. The basis vector on the LHS of \eqref{eq:Avg} is specified by six group elements,
$[v_1v_3]$, $[v_2v_3]$, $[v_3v_4]$, $[v_1v_4]$, $[v_2v_1]$, and $[v_2v_4]$.
The action of $A_{v_3}^g$ on this state reads
\begin{align}
  \label{eq:Avg}
  &A_{v_3}^g\BLvert \threeTriangles{v_1}{v_2}{v_3}{v_4}{2} \Brangle
  \nonumber\\
  =&\delta_{[v'_3v_3],g}
  \alpha\left([v_1v_2],[v_2v'_3],[v'_3v_3]\right)
  \alpha\left([v_2v'_3],[v'_3v_3],[v_3v_4]\right)
  \nonumber\\
  &\times
  \alpha\left([v_1v'_3],[v'_3v_3],[v_3v_4]\right)^{-1}
  \BLvert \threeTriangles{v_1}{v_2}{v'_3}{v_4}{2} \Brangle,
\end{align}
where on the RHS, the new enumerations are in the order $v_1<v_2<v'_3< v_3<v_4$, and
the following \textit{chain rule} of group elements on the edges holds.
\be
\begin{aligned}
\label{eq:ChainRuleInBph}
&[v_1{v'_3}]=[v_1v_3]\cdot[v_3{v'_3}],\\
&[v_2{v'_3}]=[v_2v_3]\cdot[v_3{v'_3}],\\
&[{v'_3}v_4]=[{v'_3}v_3]\cdot[v_3v_4].
\end{aligned}
\ee

The phase factor consisting of three $3$--cocycles on the RHS of Eq. (\ref{eq:Avg}) encodes the non--vanishing matrix elements of $B^{v'_3}_{v_3}$, namely
\be\label{eq:AvMatrix}
\begin{aligned}
&\left(A_{v_3}^g\right)^{[v_1v_3][v_2v_3][v_3v_4]}_{[v_1v'_3] [v_2v'_3][v'_3v_4]}([v_1v_2],[v_2v_3],[v_1v_3])\\
=&  \alpha\left([v_1v_2],[v_2v'_3],[v'_3v_3]\right)
  \alpha\left([v_2v'_3],[v'_3v_3],[v_3v_4]\right)\\
  &\times
  \alpha\left([v_1v'_3],[v'_3v_3],[v_3v_4]\right)^{-1}.
\end{aligned}
\ee
For each vertex on the LHS of Eq. (\ref{eq:Avg}), we group
its three neighboring enumerations together with
the new enumeration $v'_3$ in the ascending order.
Hence, we have $(v_1,v_2,v'_3,v_3)$ for the lower vertex,
$(v_1,v'_3,v_3,v_4)$ for the upper left vertex, and
$(v_2,v'_3,v_3,v_4)$ for the upper right one, and then assign three 3-cocycles respectively to the three vertices:
  $\alpha\left([v_1v_2],[v_2v'_3],[v'_3v_3]\right)$,
  $\alpha\left([v_2v'_3],[v'_3v_3],[v_3v_4]\right)$,
  and
  $\alpha\left([v_1v'_3],[v'_3 v_3],[v_3v_4]\right)^{-1}$.
The chirality of a $3$--cocyle, or in other words, whether a vertex contributes a $3$--cocycle
$\alpha$ or the inverse $\alpha^{-1}$, is based on the following criteria. We write down a triple for the three neighboring enumerations around each vertex in the counterclockwise
direction and append $v'_3$ to the front, namely,
$(v'_3,v_1,v_2,v_3)$ for the lower vertex,
$(v'_3,v_1,v_3,v_4)$ for the upper left one,
and $(v'_3,v_2,v_4,v_3)$ for the upper right one.
If it takes (odd) even number of steps to permute a list to the ascending order,
the vertex contributes (the inverse of) the corresponding 3-cocycle
in the action.

The matrix elements in Eq. (\ref{eq:AvMatrix}) can be better motivated and understood in the following way. One may think that the graph evolves in "time" under the driven of the Hamiltonian. Focusing on the vertex operator only and considering the $A_{v_3}^g$ in Eq. (\ref{eq:Avg}), the action of the operator creates a new "time" slice by replacing the original vertex $v_3$ by $v'_3$ and connects the two vertices in the "time" direction. This scenario is shown in Fig. \ref{fig:BvTrivalent}, which is made three--dimensional ($2+1$) to illustrate the "spacetime" picture and the relation between our model and Dijkgraaf--Witten discrete topological gauge theory to be addressed in Section \ref{sec:relation2DW}.
\begin{figure}[h!]
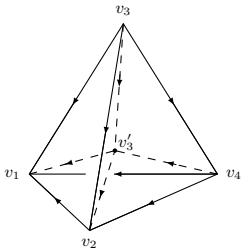

\centering
\BvTri{v_1}{v_2}{v_3}{v_4}{3}
\caption{The topology of the action of $A_{v_3}^g$.}
\label{fig:BvTrivalent}
\end{figure}

As in Fig. \ref{fig:BvTrivalent}, we can view the original three triangles on the LHS of Eq. (\ref{eq:Avg}) as a tetrahedron $v_1v_2v_3v_4$  and the three new triangles as another tetrahedron $v_1v_2v'_3v_4$, of which the vertex $v'_3$ lies inside $v_1v_2v_3v_4$ because of the ordering $v'_3<v_3$. Since $v'_3$ and $v_3$ are connected, there are three more tetrahedra generated effectively by the action of the vertex operator, namely $v_1v_2v'_3v_3$, $v_2v'_3v_3v_4$, and $v_1v'_3v_3v_4$. It looks like that the original tetrahedron is split into four tetrahedra. This splitting of tetrahedron implies the three chain rules in Eq. (\ref{eq:ChainRuleInBph}), which then enables us to endow the three tetrahedra $v_1v_2v'_3v_3$, $v_2v'_3v_3v_4$, and $v_1v'_3v_3v_4$ respectively with  the three $3$--cocycles $\alpha\left([v_1v_2], [v_2v'_3],[v'_3v_3]\right)$, $\alpha\left([v_2v'_3],[v'_3v_3], [v_3v_4]\right)$, and $\alpha\left([v_1v'_3],[v'_3 v_3],[v_3v_4] \right)^{-1}$, following the rule shown in Fig. \ref{fig:3cocycle}.

The operator $A^{g}_{v_3}$ in Eq. (\ref{eq:Avg}) is just an identity operator if $[v'_3v_3]=1$, i.e., the identity in $G$. In fact, according to Eq. (\ref{eq:Avg}), we have the following matrix element
\be\label{eq:Bph=1}
\begin{aligned}
&\alpha\left([v_1v_2],[v_2v'_3],1 \right)\alpha\left([v_2v'_3],1,[{v_3}{v_4}]\right)\\
&\times\alpha\left([v_1v'_3],1,[{v_3}{v_4}]\right)^{-1},
\end{aligned}
\ee
which is unity, by the \textit{normalization condition} \eqref{NormalizationCondition}.

The vertex operator in Eq. \eqref{eq:Avg} can naturally
extends its definition from a trivalent vertex to a vertex of any valence higher than three. The number of $3$--cocyles in the phase factor brought by the action of $A_v^g$ on a vertex is equal to the valence of the vertex.
The chirality of each $3$--cocycle in the phase factor follows the criteria described in the previous paragraph.
It is clear that $A^{g=1}_v\equiv \mathbb{I}$ by the discussion above.

It can be shown that all $B_f$ and $A_v$ are projection operators and
commute with each other (see Appendix A).
As a result, the ground states and all
elementary excitations are
thus simultaneous eigenvectors of all these local operators. Moreover,
the elementary excitations are identified as local quasi--particles
that are classified
by the the representations of the local operators.

We shall call our model \textbf{twisted quantum double model} for reasons to be explained in Section \ref{subsec:topoDof}.
\subsection{Equivalent Models}\label{subsec:equivModel}
Now that a $3$--cocycle defines a twisted quantum double model,
one may wonder since since a $3$--cocycle represents a whole equivalence class, whether two equivalent $3$--cocycles, i.e., two representatives of the same equivalent class, define the same model. Let us consider two Hamiltonians $H_{G,\alpha}$ and $H_{G,\alpha'}$, respectively defined by two equivalent $3$--cocycle $\alpha$ and $\alpha'$ that are related by the $3$--coboundary $\delta\beta$ of a normalized $2$--cochain $\beta:G^2\mapsto U(1)$ that satisfy $\beta(x,e)=1=\beta(e,x)$ for all $x \in G$,
\be\label{eq:equiv3cocycle}
\begin{aligned}
\alpha'(g_0,g_1,g_2) &=\delta\beta(g_0,g_1,g_2)\alpha(g_0,g_1,g_2)\\
&=\frac{\beta(g_1,g_2)\beta(g_0,g_1g_2)}{\beta(g_0g_1,g_2)\beta(g_0,g_1)} \alpha(g_0,g_1,g_2),
\end{aligned}
\ee
where $g_i\in G$, and $\delta$ is the coboundary operator. As each $3$--cocycle is defined on three triangles (or equally a tetrahedron) such as in Fig \ref{fig:3cocycle}, each $2$--cochain $\beta$ can be thought as defined on a triangle. Hence, Eq. (\ref{eq:equiv3cocycle}) can be viewed as a local "gauge" transformation on $\alpha$.

We now check the relation between $H_{G,\alpha'}$ and $H_{G,\alpha}$. It suffices to check only the vertex operators $A_v^g(\alpha')$ and $A_v^g(\alpha)$ because the face operators $B_f$ have merely $\delta$--functions as its matrix elements and are thus inert under the transformation in Eq. (\ref{eq:equiv3cocycle}). Without loss of generality, we consider again the vertex operator on a trivalent vertex, as that in Eq. (\ref{eq:Avg}).
By Eq. (\ref{eq:equiv3cocycle}), We immediately obtain the following.
\begin{align}
  \label{eq:AvEquiv3Cocycle}
  &A_3^g(\alpha')\BLvert \threeTriangles{1}{2}{3}{4}{2} \Brangle
  \nonumber\\
  =&\tfrac{\alpha'\left([12],[23'],[3'3]\right)
  \alpha'\left([23'],[3'3],[34]\right)}{\alpha'\left([13'],[3'3],[34]\right)}   \BLvert \threeTriangles{1}{2}{3'}{4}{2} \Brangle\nonumber\\
  =&\tfrac{\beta([12],[23])\beta([13],[34])}{\beta([23],[34])}\times \tfrac{\alpha\left([12],[23'],[3'3]\right)
  \alpha\left([23'],[3'3],[34]\right)}{\alpha\left([13'],[3'3],[34]\right)}
  \nonumber\\
  &\times
  \tfrac{\beta([23'],[3'4])}{\beta([12],[23'])\beta([13'],[3'4])}
  \BLvert \threeTriangles{1}{2}{3'}{4}{2} \Brangle,
\end{align}
where the $\delta$--function $\delta_{[3'3],g}$ is omitted for simplicity.
The second term consisting of three $\alpha$'s is precisely the matrix element of $A^g_3(\alpha)$. If we move the first fraction of $\beta$ in the second equality of the above equation to the LHS, we readily see that the action of $A^g_3(\alpha')$ on the rescaled state
\[
  \frac{\beta([23],[34])}{\beta([12],[23])\beta([13],[34])} \BLvert \threeTriangles{1}{2}{3}{4}{2} \Brangle
\]
matches perfectly the action of $A^g_3(\alpha)$ on the original state. The above rescaling is clearly a local $U(1)$ phase, which can be boiled down to the following local $U(1)$ transformation on the basis states of triangles:
\be\label{eq:basisTransformEquiv3Cocycle}
\BLvert\oneTriangle{a}{b}{c}\Brangle\mapsto\beta([ab],[bc])^{\varepsilon(a,b,c)} \BLvert \oneTriangle{a}{b}{c}\Brangle,
\ee
where $\varepsilon(a,b,c)$ is a sign, which equals $+1$ if the enumerations $a<b<c$ are clockwise on the triangle and $-1$ otherwise.
In this new basis, $A^g_v(\alpha')$ has the same matrix elements and thus the same spectrum as those of $A^g_v(\alpha)$ in the old basis.

There is a continuous deformation between any two 3--cocycles related by $\alpha'=\alpha\delta\beta$. Define a 2--cochain $\beta^{(t)}(x,y)=\beta(x,y)^t$, with
$0 \leq t \leq 1$, then $\alpha^{(t)}=\alpha\delta\beta^{(t)}$ is equivalent to $\alpha$ for all $0 \leq t \leq 1$, with $\alpha^{(0)}=\alpha$ and $\alpha^{(1)}=\alpha'$.
The corresponding transformation in Eq. (\ref{eq:basisTransformEquiv3Cocycle}) with $\beta$ replaced by $\beta^{(t)}$ is a continuous local $U(1)$ transformation; hence,
there is no phase transition in the one--parameter family of  systems with the the Hamiltonian $H_{G,\alpha^{(t)}}$ from $0\leq t \leq1$.
Thus we can conclude that the Hamiltonians $H_{G,\alpha'}$ and $H_{G,\alpha}$ due to two equivalent $3$--cocycles $\alpha'$ and $\alpha$  indeed describe the same topological phase.

\section{Topological observables and symmetries}\label{sec:topoOb}

In Hydrodynamics, topological properties of fluid, such as the stability and interactions of currents and fluxes, can be systematically studied by the diffeomorphism symmetry group acting on the fluid\cite{ArnoldHydrobook1998}. 
Analogously, the topological properties, in particular the topological observables and interactions (fusions) of the topological excitations, i.e., topological charges (currents), fluxes, and dyonic states of charge and fluxes, of a discrete model of topological phases can be systematically studied by the discrete version of diffeomorphisms, which we shall call the mutation symmetry transformations of the graph.

The symmetry we will be dealing with in this model are the mutations of the graph that preserve the spatial topology but not necessarily the local graph structure. A Hermitian operator is  \textit{a topological observable} if it is invariant under the these mutation transformations.

In most physical systems, the mutation (or diffeomorphism) symmetry does not exist. Nevertheless, in the discrete models of topological phases proposed by Kitaev\cite{Kitaev2003a,Kitaev2006}, and Levin and Wen\cite{Levin2004}, the mutation transformations to be constructed explicitly do have the space of the ground states of these models as invariant subspaces. Hence, we can use any topological observable, which is invariant under these mutation transformations, to characterize, at least partially, the topological phases in these models. One such topological observable is ground state degeneracy (GSD).

In this section, we construct the mutation transformations in our model and show that they are unitary symmetry transformations on the ground states. Then we shall define and see, as an immediate consequence, that the GSD of our model is indeed a topological observable.

All $B_f$ and $A_v$ are mutually commuting projection operators, as proven in Appendix A. Thus the ground states are the simultaneous $+1$ eigenvectors of all $B_f$, $A_v$. Define the ground state projection operator
\begin{equation}
  \label{eq:GSDprojector}
  P^0_{\Gamma}=(\prod_{f\in\Gamma}B_f)(\prod_{v\in\Gamma} A_v)
\end{equation}
and then the subspace of the ground states is
\begin{equation}
  \label{GroundStateSubspace}
  \mathcal{H}^0_{\Gamma}=\left\{|\Phi\rangle\Bigl|\Bigr. P_{\Gamma} |\Phi\rangle=|\Phi\rangle\right\}
\end{equation}

Usually, symmetry transformations in a lattice model do not affect the lattice structure and are thus well-defined on a fixed lattice. The mutation moves in our model, however, take one graph to another. Since each graph $\Gamma$ is endowed with  a Hilbert space $\Hil_{\Gamma}$ and the Hamiltonian defined in Eq. (\ref{eq:Hamiltonian}), the Hilbert space and the Hamiltonian may be subject to changes under the mutation moves.

It is known that we can always transform a triangular graph $\Gamma$ to another one $\Gamma'$ that triangulates the same Riemannian surface by a composition of the following elementary Pachner moves\cite{Pachner1978,Pachner1987}:
\begin{align}
& f_1:\; \bmm\twoTriangles{}{}{}{}{0}\emm\mapsto\bmm\twoTriangles{}{}{}{}{1}\emm\;\\
& f_2:\; \bmm\oneTriangle{}{}{}\emm\mapsto\bmm\threeTriangles{}{}{}{}{1}\emm\\
& f_3:\; \bmm\threeTriangles{}{}{}{}{1}\emm\mapsto\bmm\oneTriangle{}{}{}\emm,
\end{align}
which are the generators of all mutation transformations.

Each mutation generator $f_i:\Gamma\rightarrow\Gamma'$ induces a linear transformation $T_i:\Hil_{\Gamma}\rightarrow\Hil_{\Gamma'}$:
\begin{align}
  &T_1\BLvert \twoTriangles{v_1}{v_2}{v_3}{v_4}{0}\Brangle\label{eq:T1move}\\
  =&\sum_{[v_1v_3]\in G}\elf{v_1}{v_2}{v_3}{v_4}
  \BLvert \twoTriangles{v_1}{v_2}{v_3}{v_4}{1} \Brangle\nonumber
\end{align}
\begin{align}
  &T_2 \BLvert \oneTriangle{v_1}{v_2}{v_3} \Brangle\label{eq:T2move}\\
  =& \sum_{\substack{[v_1q],[v_2q],\\ [v_3q]\in G}}\elf{q}{v_1}{v_2}{v_3}
  \BLvert \threeTriangles{v_1}{v_2}{q}{v_3}{4} \Brangle\nonumber
\end{align}
\begin{align}
  &T_3\BLvert \threeTriangles{v_1}{v_3}{v_2}{v_4}{3}\Brangle\nonumber\\
  =& \elf{v_1}{v_2}{v_3}{v_4}
  \BLvert \oneTriangle{v_1}{v_3}{v_4}\Brangle\label{eq:T3move}
\end{align}

We now explain how we determine the linear properties of these operators.

For $T_1$, we enumerate the  four vertices in Eq. (\ref{eq:T1move}) by $v_1<v_2<v_3<v_4$. The action of $T_1$ does not change the degrees of freedom on the external edges, namely $[v_1v_2]$, $[v_2v_3]$, $[v_3v_4]$, and $[v_4v_1]$, but only changes the $[v_2v_4]$ on the internal edge to $[v_1v_3]$ on the internal edge in the new graph. The new group element $[v_1v_3]$ runs over all group elements in $G$. The operator $T_1$ also yields a $U(1)$ phase $\elf{v_1}{v_2}{v_3}{v_4}^{ \varepsilon(v_1,v_2,v_3,v_4)}$ on the basis vector in the new Hilbert space. The exponent $\varepsilon$ is a sign function that assigns an exponent $+1$ or $-1$ to the $3$--cocycle $\alpha$ according to following rule. One first picks up either of the two triangles before the action of $T_1$, then notes down as a list its three vertices counter--clockwise, e.g., either $(v_1,v_2,v_4)$ or $(v_2,v_3,v_4)$ as in Eq. (\ref{eq:T1move}). One then appends the remaining vertex to the list from the left, such as either $(v_3,v_1,v_2,v_4)$ or $(v_1,v_2,v_3,v_4)$. If it takes even permutations to shuffle the list to completely ascending order, $\varepsilon=1$, which is the case in Eq. (\ref{eq:T1move}), and $\varepsilon=-1$ otherwise.

As to $T_2$ defined in Eq. (\ref{eq:T2move}), we suppose the order of vertices is $v_1<v_2<v_3$. The action of $T_2$ creates three triangles separated by three new edges that carry respectively three new group elements. We enumerate this new vertex by $q$, which is set to be less than $v_1$, such that the three new group elements are $[qv_1]$, $[qv_2]$, and $[qv_3]$, which are then averaged out in order not to enlarge the Hilbert space.

The remaining factor in $T_2$ is also a phase, which is  in the form $\elf{v_1}{v_2}{v_3}{q}^{\varepsilon(q,v_1,v_2,v_3)}$, where the exponent is a sign depending on the orientation of the three triangles on the RHS of the equation. We determine the sign by first noting down the list of the three vertices \textit{\textbf{clockwise}} from any of the three triangles of the basis graph after the action of $T_2$, such as $(q,v_3,v_2)$ from the RHS of Eq. (\ref{eq:T2move}), then appending the remaining vertex to the beginning of the list, such as $(v_1,q,v_3,v_2);$ if the list can be turned into ascending order by even permutations, $\varepsilon=1$, which is the case in Eq. (\ref{eq:T2move}), and otherwise $\varepsilon=-1$. In general, the enumeration $q$ of the new vertex in Eq. (\ref{eq:T2move}) can have any order relative the  enumerations of the three old vertices; however, we assume $q$ is the smallest therein for simplicity.

As opposed to $T_2$, $T_3$ shrinks three triangles to a one, as in Eq. (\ref{eq:T3move}), the $3$--cocycle on the RHS of which is in general $\elf{v_1}{v_2}{v_3}{v_4}^{ \varepsilon(v_1,v_2,v_3,v_4)}$, where the exponent is a sign depending on the orientation of the three triangles on the LHS of the equation. This sign is determined this way: One first reads off the list of the three vertices \textit{\textbf{counter--clockwise}} from any of the three triangles of the original basis graph, such as $(v_2,v_3,v_4)$ from the LHS of Eq. (\ref{eq:T3move}), then appends the remaining vertex to the beginning of the list,  such as $(v_1,v_2,v_3,v_4);$ if the list can be turned into ascending order by even permutations, $\varepsilon=1$, otherwise $\varepsilon=-1$. To make life easier, in Eq. (\ref{eq:T3move}), we consider only one case.

Now we show that \textit{the mutation transformations generated by $T_1$, $T_2$, and $T_3$ are unitary symmetry transformations on the ground states.} In particular, $T_1$ is a unitarity of the entire $B_f=1$ subspace of the Hilbert space, in the sense that $\Hil^{B_f=1}_{\Gamma}\cong T_1(\Hil^{B_f=1}_{\Gamma})$. We denote the subspace of ground states of the Hilbert space $\Hil_{\Gamma}$ on a graph $\Gamma$ by $\Hil^0_{\Gamma}$. The proof consists of the following two steps.

\noindent (i). \textit{Mutation transformation preserve the space of  ground states.}

That is, if $T$ is a mutation transformation between two Hilbert spaces $\Hil_{\Gamma}$ to $\Hil_{\Gamma'}$, and if $|\Phi\rangle\in\Hil^0_{\Gamma}$, then $T|\Phi\rangle\in\Hil^0_{\Gamma'}$.

It suffices to show that $T_i P^0 _{\Gamma}=P^0_{\Gamma'}T_i$ for each mutation generator $T_i$ and each state in $\Hil^{B_f=1}$, where $P_{\Gamma}$ and $P^0_{\Gamma'}$ are the projectors respectively onto $\Hil^0_{\Gamma}$ and $\Hil^0_{\Gamma'}$ (see Appendix B). Note that however, we have $H' T \neq T H$ in general.

\noindent (ii). \textit{Mutations are unitary on ground states.}

By unitary we mean:
If $T$ is a mutation transformation between two Hilbert
spaces $\Hil_{\Gamma}$ to $\Hil_{\Gamma'}$, and if $|\Phi\rangle, |\Psi\rangle\in\Hil^0_{\Gamma}$, then
\begin{equation}
  \label{unitary}
  \left\langle T\Phi\right.\left|T\Psi\right\rangle
  =\left\langle \Phi\right.\left|\Psi\right\rangle.
\end{equation}

It is sufficient to check Eq. (\ref{unitary}) for $T_1$,$T_2$, and $T_3$ only, as seen in  Appendix B.

Consequently, there is a bijection between the ground states on any two graphs related by the mutation moves. Since two such graphs have the same spatial topology, the dimension of the ground state Hilbert space, i.e., the GSD of our model, is a a topological invariant and well--defined topological observable. Hence, our GSD can be taken as the trace of the ground state projector in Eq. (\ref{eq:GSDprojector}), which is as we have just seen a Hermitian operator that is invariant under the mutations, namely,
\be\text{GSD}=\text{tr}(P^0_{\Gamma}),\ee
where trace can be taken on any one of the graphs that are connected by the mutation moves, but the result is obviously independent of this choice.

\section{the ground state degeneracy and Topological degrees of freedom}\label{sec:GSD}

Ground state degeneracy (GSD) partially characterizes a topological phase. The nontrivial feature is that the GSD depends only on the spatial topology of the system. Two topological phases having different GSDs must be considered different.

Another important characteristic of topological phases is the emergent fractional quantum numbers of the elementary excitations in these phases and the fractional statistics of the quasiparticles of these elementary excitations. The relation with GSD is that the GSD is equal to the number of species of the quasiparticles of the elementary excitations.

The significance of the GSD lies in the degrees of freedom that are capable of distinguishing the degenerate ground states. The topological dependence of the GSD originates in that these degenerate degrees of freedom are global. An interesting question then arises: How to characterize these global degrees of freedom? Answering this question will enable us to (1) discern between two different topological phases that have the same topological dependence of GSD and (2) understand better the relationship between the global degrees of freedom in the degenerate ground states and the emergent fractional quantum numbers of the elementary excitations.

In what follows we calculate the GSD of our model on a torus and then analyze the global degrees of freedom in the degenerate ground states.

\subsection{Ground state degeneracy on a torus}\label{subsec:GSDtorus}
The topological invariance of the GSD of our model enables us to compute the GSD on the simplest triangle graph that triangulates the surface on which the model is defined.

In the case of finite groups, the GSD of our model on a 2--sphere is always unity because a  2--sphere has a trivial topology, in the sense that its fundamental group is trivial. This fact can be checked by following the  approach to be presented shortly in this section. This is a common feature of all known models of topological phases.

A torus is the simplest closed surface with a non--trivial topology. Fig. \ref{fig:torus} depicts the  simplest
triangle graph that triangulates a torus.
\begin{figure}[h!]
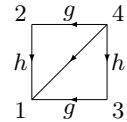

\centering
\torusgraph
\caption{Triangulation of a torus, with $g,h\in G$.}
\label{fig:torus}
\end{figure}

This graph has two triangle faces and only one vertex. But for the sake of assigning the 3--cocycles in $A_v^x$ easily, we use $1,2,3$ and $4$ to enumerate the sole vertex. This is perfectly fine because the boundary condition automatically merge the differently labeled vertices into one. We identify the boundary  edge $[12]$ with $[34]$, and $[13]$ with $[24]$. It is tricky to notice that the four enumerations can not be arbitrary. In Fig. \ref{fig:torus}, the orientations of the two boundary edges are consistently taken from higher enumerations to lower enumerations.

The subspace $\Hil^{B_f=1}$ is spanned by the basis vectors
\be
  \left\{\left|g,h\right\rangle\right.\left|g,h\in G,gh=hg\right\}
\ee
corresponding to the assignment of $[13]=g,[12]=h$ and $[14]=gh=hg$ in the above graph.

Since there is only one vertex in Fig. \ref{fig:torus}, we simplify the notation of  $A_v^x$ at this mere vertex by $A^x$, the action of which, according to its definition in Eq. \eqref{eq:Avg}, is
\begin{align}
  \label{eq:torusAxEnumeration}
  &A^x \left|[13],[12]\right\rangle
  \nonumber\\
  =&\left(\elf{1}{3}{4'}{4}\elf{1}{2}{4'}{4}^{-1}\right)
    \nonumber\\
  &\elf{1}{3'}{3}{4'}^{-1}
  \elf{1}{2'}{2}{4'}
  \nonumber\\
  &\left(\elf{1'}{1}{3'}{4'}\elf{1'}{1}{2'}{4'}^{-1}\right)
  \nonumber\\
  &\left|[1'3'],[1'2']\right\rangle,
\end{align}
where the rule $[1'1]=[2'2]=[3'3]=[4'4]=x$ holds for the new enumerations $1',2',3'$ and $4'$. The coefficient consisting of the six 3--cocycles are determined as follows. We obtain $\left(\elf{1}{3}{4'}{4}\elf{1}{2}{4'}{4}^{-1}\right)$ at enumeration $4$, then replace $4$ by $4'$ to determine the rest of the factors. Next we obtain $'\elf{1}{3'}{3}{4'}^{-1}$ at enumeration $3$ and then again replace $3$ by $3'$. Having repeated similar steps at enumerations $2$ and $1$, we arrive at the above formula. One may derive a seemingly different coefficient by following a different path, e.g., $1\rightarrow 2\rightarrow 3\rightarrow 4$. But because of the topological invariance, the new coefficient can be brought to precisely the same as that in Eq. \eqref{eq:torusAxEnumeration} by applying 3--cocycle conditions, as one can check.

Now we write down the action explicitly in terms of the group elements $g,h$ and $x$.
\begin{align}
  \label{eq:torusAx}
  &A^x\left|g,h\right\rangle
  \nonumber\\
  =&\alpha(g,hx^{-1},x)\alpha(h,gx^{-1},x)^{-1}
  \nonumber\\
  &\alpha(gx^{-1},x,hx^{-1})^{-1}
  \alpha(hx^{-1},x,gx^{-1})\\
  &\alpha(x,gx^{-1},xhx^{-1})\alpha(x,hx^{-1},xgx^{-1})^{-1}
  \left|xgx^{-1},xhx^{-1}\right\rangle\nonumber
\end{align}

One can verify the multiplication law $A^xA^y=A^{xy}$ by 3--cocycle conditions \eqref{3CocycleCondition}, which agrees with the result in Appendix \ref{app:algAvBf}. The ground-state projector is thus
\be
\label{eq:torusProjector}
P^0=\frac{1}{|G|}\sum_x A^x
\ee

Taking a trace of the ground--state projector \eqref{eq:torusProjector} computes the GSD,
\begin{align}
  \label{eq:GSDalpha}
  \text{GSD}=&\text{tr}(\frac{1}{|G|}\sum_x A^x)
  \nonumber\\
  =&\sum_{g,h}\delta_{gh,hg} \left\langle g,h\right|A^x\left|g,h\right\rangle
  \nonumber\\
  =&\frac{1}{|G|}\sum_{h,g,x}\delta_{gh,hg}\delta_{hx,xh}\delta_{xg,gx}
  \alpha(g,hx^{-1},x)
    \nonumber\\
  &\alpha(h,gx^{-1},x)^{-1}
  \alpha(gx^{-1},x,hx^{-1})^{-1}
  \alpha(hx^{-1},x,gx^{-1})
  \nonumber\\
  &\alpha(x,gx^{-1},h)\alpha(x,hx^{-1},g)^{-1}
\end{align}
where the trace is evaluated in the subspace $\mathcal{H}^{B_f=1}$.

The seemingly complicated summation of the six $3$--cocycles in Eq. (\ref{eq:GSDalpha}) can actually be simplified in many cases due to a hidden simple mathematical structure. To see this, we shall first explore in the next subsection the algebraic structure in Eq. \eqref{eq:torusAx}, after which we come back to the simplification of the GSD.

\subsection{Topological degrees of freedom}\label{subsec:topoDof}
We now proceed to extract the algebraic structure in Eq. \eqref{eq:torusAx} and explore the classification of the topological degrees of freedom in the ground states, so as to reveal the deep mathematical significance of the GSD yet not fully discussed in the previous subsection.

To this end, we rewrite Eq. (\ref{eq:torusAx}) as follows by applying appropriate 3-cocycle conditions (see Appendix C for the derivation).
\begin{align}
  \label{torusAxRewrite}
   &A^x\left|g,h\right\rangle
  \nonumber\\
  =&
  \frac{
  \alpha(g,x^{-1},xhx^{-1})\alpha(x^{-1},xhx^{-1},xgx^{-1})
  }{\alpha(x^{-1},xgx^{-1},xhx^{-1})}
  \nonumber\\
  &
  \frac{ \alpha(h,g,x^{-1})}{
  \alpha(g,h,x^{-1})\alpha(h,x^{-1},xgx^{-1})
  }
  \left|xgx^{-1},xhx^{-1}\right\rangle.
\end{align}

By defining a new function in terms of 3--cocycles as
\begin{align}
  \label{eq:betafunction}
  \beta_a(b,c)\defeq
  \frac{\alpha(a,b,c)\alpha(b,c,c^{-1}b^{-1}abc)}
  {\alpha(b,b^{-1}ab,c)},
\end{align}
$\forall a,b,c\in G$, and plugging it into Eq. (\ref{torusAxRewrite}), we obtain
\be\label{eq:torusAxBeta}
\begin{aligned}
   A^x\left|g,h\right\rangle =
  &\frac{\beta_g(x^{-1},xhx^{-1})}{\beta_g(h,x^{-1})}
  \left|xgx^{-1},xhx^{-1}\right\rangle\\
  &=\eta^g(h,x)\left|xgx^{-1},xhx^{-1}\right\rangle,
\end{aligned}
\ee
where we define
\be\label{eq:eta}
\eta^g(h,x)=\frac{\beta_g(x^{-1},xhx^{-1})}{\beta_g(h,x^{-1})},
\ee
for any given $g\in G$ and $h\in Z_g=\{x\in G|xg=gx\}$, the centralizer subgroup for $g\in G$. Let $h=g$ in the above definition, we have
\be\label{eq:etaggIs1}
\eta^g(g,x)=\frac{\beta_g(x^{-1},xgx^{-1})}{\beta_g(g,x^{-1})}=1,
\ee
for all $g,x\in G$, which can be quickly checked by directly using the $3$--cocycle condition. Interestingly, if $x\in Z_{g,h}=\{x\in G|xg=gx,xh=hx\}$, we see that the $U(1)$ number
\be\label{eq:rhoEta}
\rho^g(h,x)\defeq\eta^g(h,x)\Big|_{Z_{g,h}}=\frac{\beta_g(x^{-1},h)}{\beta_g(h,x^{-1})}
\ee
is actually a $1$--dimensional representation of the subgroup $Z_{g,h}\subseteq G$. This is because $\rho^g(h,x)\rho^g(h,y)=\rho^g(h,xy)$, which is a consequence of $A^x A^y=A^{xy}$ on the ground states.

It follows from Eq. \eqref{eq:torusAxBeta} that the ground states are spanned by the vectors
\begin{align}
  \label{eq:GGquotConj}
  \left\{\frac{1}{|G|}\sum_{x\in G} \eta^g(h,x)
  \left|xgx^{-1},xhx^{-1}\right\rangle\Biggl|\Biggr.g\in G,\ h\in Z_g\right\}.
\end{align}
This tempts one to think that counting the GSD amounts to counting the elements in $\text{Hom}\left(\pi_1(T^2),G\right)/conj$, where the $conj$ in the quotient is the conjugacy equivalence: $(g,h)\sim(xgx^{-1},xhx^{-1})$ for any $x$. This is in general not true, however, as one may over--count the states because in Eq. \eqref{eq:GGquotConj}, the terms that are summed over for some $g$ and $h$ may actually vanish, causing the corresponding states non--existing, as we now classify by studying the algebraic structure of
the function $\beta_a$ defined in Eq. (\ref{eq:betafunction}).

Using the $3$--cocycle condition of $\alpha$, one can show that the  function $\beta_a$ is in fact a normalized, \textbf{twisted} $2$--cocycle satisfying twisted $2$--cocycle condition,
\begin{align}
  \label{eq:betacondition}
  \widetilde{\delta}\beta_a(x,y,z)\equiv\frac{\beta_{x^{-1}ax}(y,z)\beta_a(x,yz)}{\beta_a(xy,z)\beta_a(x,y)}=1,
\end{align}
and the normalization condition
\begin{align}
  \label{eq:betanormalization}
  \beta_a(x,e)=\beta_a(e,x)=1,
\end{align}
for all $a,x,y,z\in{G}$. The $\widetilde{\delta}\beta_a$ is called the twisted 3--coboundary of $\beta_a$.

Furthermore, when its variables are restricted to the centralizer $Z_a$ of $a\in G$, $\beta_a$ clearly reduces to a normalized, usual $2$--cocycle over $Z_a$, which obeys the usual 2--cocycle condition,
\be  \label{eq:betaconditionUsual}
  \beta_{a}(y,z)\beta_a(xy,z)^{-1}\beta_a(x,yz)\beta_a(x,y)^{-1}=1,
\ee
for all $x,y,z\in Z_a$.

The function $\beta_a$ is closely related to the projective representations
of $Z_a$. In fact, each $\beta_a$ classifies a class of projective representations called $\beta_a$--\textbf{representations} $\widetilde{\rho}:Z_a{\rightarrow}\text{GL}\left(Z_a\right)$ obeying
\be\label{eq:betarepresentation}
\widetilde{\rho}(x)\widetilde{\rho}(y)=\beta_a(x,y)\widetilde{\rho}(xy)
\ee
It is evident that the normalization condition corresponds to $\widetilde{\rho}(e)\widetilde{\rho}(x)=\widetilde{\rho}(x)\widetilde{\rho}(e) =\widetilde{\rho}(x)$, while the $2$--cocycle condition in Eq. \eqref{eq:betacondition} corresponds to the associativity $\widetilde{\rho}(x)\left(\widetilde{\rho}(y)\widetilde{\rho}(z)\right) =\left(\widetilde{\rho}(x)\widetilde{\rho}(y)\right)\widetilde{\rho}(z)$. In particular, if the $3$--cocycles that define $\beta_a$ are the trivial one, i.e., $\alpha=\alpha^0\equiv 1$,  then $\beta_a=1$, reducing $\tilde{\rho}$ to the usual linear representations of $Z_a$.

In this paper, we are interested in the classification of $\beta_a$--representations of $Z_a$ with fixed $a\in G$. Here we record a few important properties of this kind of representations.

An element $g\in Z_a$ is $\beta_a$\textbf{--regular} if $\beta_a(g,h)=\beta_a(h,g)$ for all $h\in Z_{a,g}\subseteq Z_a$. Moreover, $g$ is $\beta_a$--regular if and only if all its conjugates are so, which can be verified by the 3--cocycle condition in Eq. \eqref{3CocycleCondition}. That is,  $g$ is $\beta_a$--regular $\Leftrightarrow$ $[g]$ is $\beta_a$--regular. We call $[g]$ a $\beta_a$\textbf{--regular conjugacy class}. In particular, each $g$ is always $\beta_g$--regular.

Let us denote all the conjugacy classes of $G$ by $C^A$ and number of such classes by $r(G)$. Since for any $a,b\in C^A$, $Z_a\cong Z_b$, it is convenient to denote these isomorphic centralizers by $Z^A$, obtained by any representative of the class. We henceforth collect any chosen set of representatives of all $C^A$ by simply $R_C=\{g^A\in C^A|A=1\dots r(G)\}$.

For $a\in Z^A$, let the number of $\beta_a$--regular conjugacy classes in $Z^A$ be $r(Z^A,\beta_a)$. Clearly, we have
\be\label{eq:rGbetaLrG}
r(Z^A,\beta_a)\leq r(Z^A).
\ee
It is known that the number of inequivalent irreducible $\beta_a$--representations  of $Z^A$ is equal to $r(Z^A,\beta_a)$. In particular, in the case where $\beta_a=1$ because of the trivial 3--cocycle $\alpha^0$, we arrive at the familiar result that the number of the irreducible linear representations equals the number of conjugacy classes. Eq. \eqref{eq:rGbetaLrG} states that irreducible $\beta_a$--representations of $Z_a$ are fewer than the irreducible liner representations.

The topological degrees of freedom is related to this classification of projective representations of $Z^A$. To show this, we reexpress the GSD in Eq. \eqref{eq:GSDalpha} as
\be
  \label{eq:GSDbeta}
\begin{aligned}
  \text{GSD}=&
  \frac{1}{|G|}\sum_{g,h,x\in G}\delta_{gh,gh}\delta_{gx,xg}\delta_{hx,xh}\eta^g(h,x)\\
  =&  \frac{1}{|G|}\sum_{g\in G}\sum_{h\in Z_g}\sum_{x\in Z_{g,h}}\rho^g(h,x)
  ,
\end{aligned}
\ee
which can be further simplified, by the identity
\begin{align}
  \label{eq:SumBetaToDelta}
  \frac{1}{|Z_{g,h}|}\sum_{x\in Z_{g,h}}
  \rho^g(h,x)
  =
  \left\{
  \begin{array}{ll}
    1, & h\text{ is }\beta_g\text{ regular}\\
    0, & \text{otherwise}
  \end{array}
  \right.
\end{align}
where $|Z_{g,h}|$ is the order of the subgroup $Z_{g,h}$ with fixed $g,h\in G$.

Here is the proof of Eq. \eqref{eq:SumBetaToDelta}. As shown below Eq. (\ref{eq:rhoEta}), the phase $\rho^g(h,x)$ in Eq. (\ref{eq:rhoEta})  is a $1$--dimensional representation of $Z_{g,h}$; it is the trivial representation $\rho^0=\mathds{1}$ if $h$ is $\beta_g$--regular and is otherwise a non--trivial irreducible representation (i.e., different from the identity representation). By the orthonormal condition
\be
\frac{1}{|Z_{g,h}|}\sum_{x\in Z_{g,h}}{\rho^g}^{(j)}(h,x)
  =\delta_{j,0},
\ee
where $j=0$ corresponds to the trivial representation and $j\neq 0$ a non-trivial irreducible representation, we obtain Eq. \eqref{eq:SumBetaToDelta}.

Equation (\ref{eq:SumBetaToDelta}) renders Eq. \eqref{eq:GSDbeta} as
\begin{align}
  \label{eq:GSDconjugacy}
  \text{GSD}=&\sum_{g\in G}\sum_{h\in Z_g}\frac{|Z_{g,h}|}{|G|}\times
  \left\{
  \begin{array}{ll}
    1, & h\text{ is }\beta_g\text{ regular}\\
    0, & \text{otherwise}
  \end{array}
  \right.
  \nonumber\\
  =&\sum_{h\in Z^A,A}\frac{|Z^A|}{|G|}\frac{|Z_{g^A,h}|}{|Z^A|}\times
  \left\{
  \begin{array}{ll}
    1, & h\text{ is }\beta_{g^A}\text{--regular}\\
    0, & \text{otherwise}
  \end{array}
  \right.
  \nonumber\\
  =&\sum_A r(Z^A,\beta_{g^A}).
\end{align}
In the last equality use is made of that $|G|/|Z^A|=|C^A|$.

According to the relationship between the number of $\beta_g$--regular conjugacy classes of $Z_g$ and the number of $\beta_g$--representations of $Z_g$ as discussed above, the GSD can take the form
\be\label{eq:GSDrepresentations}
  \text{GSD}
  =\sum_{A}\#(\beta_{g^A}\text{--representations of }Z^A),
\ee
where $\#$ stands for \textquotedblleft the number of".

\subsection{Ground States Basis}\label{subsec:GSbasis}

As promised in the previous subsection, we have simplified GSD evaluation in Eq. (\ref{eq:GSDalpha}) to counting the relevant projective representations. Computing the GSD of our model on a torus amounts to counting the irreducible projective $\beta_{g^A}$--representations of each conjugacy class $C^A$, then sum it over $C^A$ in $G$.

Hence, the ground states on a torus  can be labeled by pairs $(g^A,h)$ with $g^A$ running over $R_C$ and $h$ running over a set of $\beta_{g^A}$--regular conjugacy class representatives of $Z^A$. Equivalently, the ground states can also be labeled by pairs $(A,\mu)$ with $A=1\dots r(G)$ and $\mu$ labeling $\widetilde{\rho}^{g^A}_{\mu}$, which are the irreducible $\beta_{g^A}$--representations of $Z^A$. We posit that the basis vectors $\ket{A,\mu}$ can be defined as:
\be
  \label{eq:AmuBasis}
  \ket{A,\mu}
  =\frac{1}{\sqrt{|G|}}\sum_{g\in C^A,h\in Z_g}\,
  \widetilde{\chi}^g_{\mu}(h)\,\ket{g,h},
\ee
where $\widetilde{\chi}^g_{\mu}(h)=\text{tr}\widetilde{\rho}^g_{\mu}(h)$ is the projective character defined as usual by the trace of the representation, and $|G|$ the order of $G$. Since the centralizers $Z_g$ are isomorphic for all $g$ in a conjugacy class $C^A$, so are the set of irreducible $\beta_g$--representations of $Z_g$ for all $g\in C^A$. Therefore the same label $\mu$ works for all $\beta_g$--representations. We detail the construction of the isomorphism among the irreducible $\beta_g$--representations for $g\in C^A$ in Appendix \ref{app:solModularTrans}. The projective characters $\widetilde{\chi}^g_{\mu}(h)$ satisfy the following relation under simultaneous conjugation of $g$ and $h$:
\be
\label{eq:CharacterInConjugacyClasses}
\widetilde{\chi}^{xgx^{-1}}_{\mu}(xhx^{-1})= \eta^g(h,x)\widetilde{\chi}^g_{\mu}(h).
\ee
for all $x\in G$.
Keep in mind that if $h\in Z^A$ but $h\neq g$, then $h\notin C^A$. Practically, for each conjugacy class $C^A$ with its representative element $g^A$, if we find a $\beta_{g^A}$--representation $\widetilde{\rho}^{g^A}_{\mu}$ of $Z^A$, we can construct the $\beta_g$--representations for all other elements $g$ of $C^A$. Throughout this paper, we take the representations such that the relation \eqref{eq:CharacterInConjugacyClasses} are always satisfied.

Note that in general, the projective characters are not functions of conjugacy classes because the fact that $\widetilde{\rho}^g_{\mu}(a)\widetilde{\rho}^g_{\mu}(b)=\beta_{g}(a,b)\widetilde{\rho}^g_{\mu}(ab)$ yields the relation  $\widetilde{\chi}^g_{\mu}(xhx^{-1})= \left[\beta_{g}(hx^{-1},x)/\beta_{g}(x,hx^{-1})\right]\widetilde{\chi}^g_{\mu}(h)$. Nevertheless, the orthogonality and completeness relations of these projective characters still hold, namely, for all $\beta_{g}$--regular elements $a,b$ in $G$,
\be
\begin{aligned}
  \label{eq:characterrelation}
  &\frac{1}{|Z_g|}\sum_{h\in Z_g}
  \overline{\widetilde{\chi}^g_{\mu}(h)}
  \widetilde{\chi}^g_{\nu}(h)
  =\delta_{\mu,\nu},\\
  &\frac{|C^A|}{|Z_g|}
  \sum_{\mu}\,
  \overline{\widetilde{\chi}^g_{\mu}(a)}\,
  \widetilde{\chi}^g_{\mu}(b)
  =
  \left\{
  \begin{array}{ll}
    1, & a\text{ conjugate to }b\\
    0, & \text{otherwise}
  \end{array}
  \right.,
\end{aligned}
\ee
where $|Z_g|$ is the order of the subgroup $Z_g$, and $|C^A|$ is the cardinality of the conjugacy class $C^A$ containing $a$ in the subgroup $Z_g$.
By Eq. \eqref{eq:characterrelation}, one can verify that the basis in Eq. \eqref{eq:AmuBasis} is orthonormal. Moreover, if $h$ is not $\beta_{g}$--regular, then $\widetilde{\chi}^g_{\mu}(h)=0$, which is the very Proposition \ref{prop:chiIs0} proven in Appendix \ref{app:solModularTrans}.

We can now justify that $\ket{A,\mu}$ is indeed a ground state by its invariance under the action of the ground state projection operator $P^0$ defined in Eq. \eqref{eq:torusProjector}.
\begin{align}
& P^0\ket{A,\mu}\nonumber\\
=&\frac{1}{|G|}\sum_{x\in G}A^x\ket{A,\mu}\nonumber\\
=&\frac{1}{\sqrt{|G|^3}}\sum_{x\in G}\sum_{\substack{g\in C^A \\ h\in Z_g}}\widetilde{\chi}^g_{\mu}(h)\eta^g(h,x)\ket{xgx^{-1}, xhx^{-1}}\nonumber\\
=&\frac{1}{\sqrt{|G|^3}}\sum_{x\in G}\sum_{\substack{g\in C^A \\ h\in Z_g}}\widetilde{\chi}^{xgx^{-1}}_{\mu}(xhx^{-1})\ket{xgx^{-1}, xhx^{-1}}\nonumber\\
=&\frac{1}{\sqrt{|G|^3}}\sum_{g'\in C^A,  h'\in Z_g}\widetilde{\chi}^{g'}_{\mu}(h')\ket{g', h'}\sum_{x\in G}1\nonumber\\
=&\frac{1}{\sqrt{|G|}}\sum_{g\in C^A,h\in Z_g}
  \widetilde{\chi}^g_{\mu}(h)\,\ket{g,h}=\ket{A,\mu},\label{eq:P0Amu}
\end{align}
where Eqs. \eqref{eq:torusAxBeta} and \eqref{eq:CharacterInConjugacyClasses} are used respectively in the second and third equalities, while substitutions $g'=xgx^{-1}$ and $h=xhx^{-1}$ are made to get the fourth equality but renamed back to $g$ and $h$ in the end. Therefore, we conclude that the set of $\ket{A,\mu}$ does furnish an orthonormal basis of the ground states, i.e.,
\be\label{eq:H0basis}
\Hil^0=\mathrm{span}\{\ket{A,\mu}:A=1\dots r(G),\ \mu=1\dots r(Z^A,\beta_{g^A})\}
\ee

This uncovers the mathematical structure that classifies the topological degrees of freedom in the ground states via representation theory. We start with our model specified by a 3--cocycle $\alpha$ over $G$ and end up with the result that the topological degrees of freedom are determined by the 2--cocycles $\beta_g$ over $Z_g$.

the ground--state basis vectors $(A,\mu)$ label the set of all inequivalent irreducible representation spaces of the twisted quantum double $D^{\alpha}(G)$, which plays a central role in the orbifolds by a symmetry group $G$ of a holomorphic conformal field theory. We may understand the term ``twisted'' as twisting linear representations to projective representations. We are not going to explain the details of the twisted quantum double, which is beyond the concern of this paper. But for completeness, we note here the multiplication law in the twisted quantum double $D^{\alpha}(G)$:
\be
  \label{mult}
  (P_a\otimes x)(P_b\otimes y)=\beta_a(x,y) \delta_{a,xbx^{-1}}(P_a\otimes xy),
\ee
for all $a,b,x,y\in G$, where $P_a$ projects out $a$ while $x$ obeys the usual group multiplication with a projective phase factor.

In particular,  as to be shown in Section \ref{sec:Kitaev}, the untwisted version of our model (i.e., when $\alpha=\alpha^0$) turns out to be Kitaev's quantum double model (or, the toric code model), the GSD of which agrees with the number of irreducible representations of the quantum double $D(G)$ of the finite group $G$. Therefore, our model can be viewed as a deformation of the quantum double model by a twisting with the $\beta_a$ in Eq. \eqref{eq:betafunction}, which twists the linear representations of a group to the projective representations. This is mainly why we christen our model \textbf{twisted quantum double (TQD) model}.

\section{Fractional Topological numbers}\label{sec:fractionaltopologicalnumbers}

In the previous section, we studied the GSD as the simplest topological observable of our model. But topological phases are only partially characterizes GSD. It is possible that two models specified by two inequivalent 3--cocycles have the same GSD but in the mean time, give rise to distinct topological phases.

Hence, a natural question is how to differentiate two distinct topological phases if they bear the same GSD. It is known that the emergent fractional topological numbers in the elementary excitations can differentiate such distinct topological phases.

In this section, we shall first construct on the subspace $\Hil^{B_f=1}$ the topological observables then solve their eigen--problems to acquire the expected fractional topological numbers. These fractional topological numbers are related to the fractional statistics of quasiparticles in the elementary excitations. Actually, there is believed to exist a correspondence between the topological degrees of freedom in the ground states of the system on a torus and the local degrees of freedom of the quasiparticles in the elementary excitations. We shall come back to address this correspondence in Section \ref{sec:relation2DW}.
\subsection{Topological observables as $\text{SL}(2,\Z)$ generators}\label{sub:TopoObservableSL2Z}

Consider the graph $\Gamma$ on which the model is defined. In Section \ref{sec:topoOb}, we constructed the mutation transformations that can change the local structure of the graph but preserve the graph topology, i.e., the topology that $\Gamma$ triangulates. Under such mutations, the topological degrees of freedom of the ground states are intact. All such transformations are local. The ground--state projector
$\prod_v A_v$ can also be constructed from such mutations.

Here, on the other hand, we look into the large transformations that alter the graph structure globally but still preserve the graph topology and lead to richer topological observables.

Again, since we are not interested in the local transformations of the graph, we need only to work on the simplest triangulation of torus as in Fig. \ref{fig:torus}.

The transformations that change the topology are the familiar modular transformations, which form the group $SL_2(\Z)$ that is generated by
\be  \label{eq:SL2Z}
  \str=\bpm
     0 & -1 \\
     1 & 0
    \epm,
    \quad
  \ttr=\bpm
     1 & 1 \\
     0 & 1
    \epm,
\ee
satisfying relations $(\str\ttr)^3=\str^2$ and $\str^4=1$.

To cast the modular transformations in the form of $3$--cocycles, let us redraw the torus in Fig. \ref{fig:torus}
in the coordinate frame in Fig. \ref{fig:SL2Ztransformation},
which illustrates the $\str$ and $\ttr$ transformations on the torus.
\begin{figure}[h!]
\centering
  \begin{align*}
  &\str: \bmm\SLtransformationSone\emm
  \mapsto
  \bmm\SLtransformationStwo\emm\\
  &\ttr: \bmm\SLtransformationTone\emm
  \mapsto
  \bmm\SLtransformationTtwo\emm
  \end{align*}
\caption{$\str$ and $\ttr$ transformations of a torus.}
\label{fig:SL2Ztransformation}
\end{figure}
The $\str$ and $\ttr$ transformations on the subspace $\Hil^{B_f=1}$
are constructed as follows.

We leave the details of the construction to Appendix \ref{app:modularTrans} but claim here that
\begin{align}
  \label{eq:SxtransformationEnumerations}
  &\str^x \BLvert\torusgraphST{1}{}{2}{}{3}{}{4}{}{0}\Brangle
  \nonumber\\
  =& \elf{2'}{1}{3}{4}{\elf{2'}{1}{2}{4}}^{-1}
  \nonumber\\
   \times& {\elf{2'}{4'}{1}{3}}^{-1}{\elf{1'}{2'}{1}{2}}^{-1}
   \nonumber\\
   \times& \elf{2'}{3'}{4'}{1}{\elf{1'}{2'}{3'}{1}}^{-1}
   \nonumber\\
   \times& {\elf{1'}{2'}{3'}{4'}} \BLvert \torusgraphST{3}{'}{1}{'}{4}{'}{2}{'}{1}\Brangle
\end{align}
where we set the order of the enumerations by $1'<2'<3'<4'<1<2<3<4$ such that the orientation of the two boundary edges are taken consistently. One sees that the wave function transforms oppositely in Fig. \ref{fig:SL2Ztransformation}.

Taking $[12]=[34]=h$, $[13]=[24]=g$, and $[3'1]=[1'2]=[4'3]=[2'4]=x$, $\str^x$ is casted explicitly in terms of the group elements as
\begin{align}
  \label{eq:SxtransformationGroupElements}
  &\str^x\ket{g,h}
  \nonumber\\
  =&\alpha(xg^{-1}h^{-1},g,h)\alpha(xg^{-1}h^{-1},h,g)^{-1}
  \nonumber\\
  \times &\alpha(xh^{-1}x^{-1},xg^{-1},g)^{-1}\alpha(xgx^{-1},xg^{-1}h^{-1},h)^{-1}
  \nonumber\\
  \times &\alpha(xg^{-1}h^{-1}x^{-1},xgx^{-1},xg^{-1})
  \alpha(xgx^{-1},xg^{-1}h^{-1}x^{-1},x)^{-1}
  \nonumber\\
  \times &\alpha(xgx^{-1},xg^{-1}h^{-1}x^{-1},xgx^{-1})\ket{xh^{-1}x^{-1},xgx^{-1}},
\end{align}
where $\ket{[1'3'],[1'2']}\defeq\ket{xh^{-1}x^{-1},xgx^{-1}}$.

Similarly, we claim that $\ttr^x$  behaves as
\begin{align}
  \label{eq:Txtransformation}
  & \ttr^x \BLvert\torusgraphST{1}{}{2}{}{3}{}{4}{}{0}\Brangle
  \nonumber\\
  =& \elf{2'}{1}{3}{4}{\elf{2'}{1}{2}{4}}^{-1}
  \nonumber\\
  \times &\elf{3''}{2'}{1}{3} {\elf{4''}{2'}{1}{2}}^{-1}
  \nonumber\\
  \times &\elf{1''}{3''}{2'}{1} {\elf{1''}{4''}{2'}{1}}^{-1}
  \nonumber\\
  \times &\elf{1''}{3''}{4''}{2'}
  \BLvert \torusgraphST{1}{''}{4}{''}{3}{''}{2}{'}{1}\Brangle
\end{align}
where we set $1''<3''<4''<2'<1<2<3<4$ as the order of enumerations, and when explicitly expressed in term of group elements, becomes
\begin{align}
  \label{eq:TxtransformationGroupElements}
  &\ttr^x\ket{g,h}
  \nonumber\\
  =&\alpha(xg^{-1}h^{-1},g,h)\alpha(xg^{-1}h^{-1},h,g)^{-1}
  \nonumber\\
  \times &\alpha(xhx^{-1},xg^{-1}h^{-1},g)\alpha(xgx^{-1},xg^{-1}h^{-1},h)^{-1}
  \nonumber\\
  \times &\alpha(xgx^{-1},xhx^{-1},xg^{-1}h^{-1})\alpha(xhx^{-1},xgx^{-1},xg^{-1}h^{-1})^{-1}
  \nonumber\\
  \times &\alpha(xgx^{-1},xhg^{-1}x^{-1},xgx^{-1})\ket{xgx^{-1},xg^{-1}hx^{-1}}
\end{align}

The $s$  and $\ttr$ transformation are defined by
\be
\label{eq:STequalSumSTx}
\str=\frac{1}{|G|}\sum_x \str^x,
\quad \ttr=\frac{1}{|G|}\sum_x \ttr^x.
\ee

The operators $\str,\ttr$ in Eq. \eqref{eq:STequalSumSTx} are a representation of the $\str$ and $\ttr$ matrices in Eq. \eqref{eq:SL2Z} on the subspace $\mathcal{H}^{B_f=1}$ of the model. Indeed, a direct evaluation by the 3--cocycle conditions verifies $(\str\ttr)^3=\str^2$ and $\str^4=1$.

On a torus, the vertex operator $A^x$ and the modular transformation operators $\str^x$ and $\ttr^x$ comprise an interesting algebraic structure, namely,
\be\label{eq:algSTandA}
\begin{aligned}
&\str^x A^y=\str^{xy}=A^x S^y,\\
&\ttr^x A^y=\ttr^{xy}=A^x\ttr^y.
\end{aligned}
\ee
We shall not prove this here, as it can be done straightforwardly by manipulating the $3$--cocycles in the above equations. This algebraic structure results in the following important reexpression of the $\str$ and $\ttr$ operators.
\be\label{eq:STandP0}
P^0\ttr^e=\ttr=\ttr^e P^0,\quad P^0\str^e=\str=\str^e P^0,
\ee
the first of which is proven as follows.
\begin{align*}
\ttr=\frac{1}{|G|}\sum_{x\in G}\ttr^x\xeq[]{Eq. \eqref{eq:algSTandA}} &\frac{1}{|G|}\sum_{x\in G}\ttr^e A^x\\
=&\ttr^e(\frac{1}{|G|}\sum_{x\in G}A^x)\\
=&\ttr^e P^0=P^0\ttr^e,
\end{align*}
where $e$ is the identity element of $G$. The proof of the second relation in Eq. \eqref{eq:STandP0} follows likewise. This indicates that  the operators $\str$ and $\ttr$ are indeed topological observables and symmetries in $\Hil^0$.

We can lay the ground states in the basis composed of the eigenvectors $\{\Phi_k\}$ of $\ttr$,
\be
\label{eq:GSintermsofT}
\ttr\ket{\Phi_k}=\theta_k \ket{\Phi_k}
\ee
where $\theta_k$ is a $U(1)$ phase, and $k=1,2,...,\text{GSD}$ labels the degenerate ground states. These eigenvectors will be identified with $\ket{A,\mu}$ in the next subsection.

We remark that $\ttr$ also has other eigenvectors, whose eigenvalues are zero, which is implied by the first relation in Eq. \eqref{eq:STandP0}. These zero eigenvectors are actually the excited states of the model; however, we are not going to dwell on them in this paper.

Hence, one can regard the eigenvalues $\theta_k$ of $\ttr$ as a set of topological numbers of the model. In fact, from Fig. \ref{fig:SL2Ztransformation}, $\ttr$ can be viewed as a global twisting of the system, and thus its eigenvalues $\theta_k$ can be regarded as the topological spins of the topological sectors $\ket{\Phi_k}$.

Another set of topological numbers are the $\str$--matrix of the topological sectors,
\be
\label{eq:smatrix}
s_{ij}=\left\langle\Phi_i \right| \str\ket{\Phi_j}.
\ee
where $i,j=1,2,...,\text{GSD}$. This matrix is orthonormal:
\be \label{eq:Smatrixcondition}
  \sum_{j}s_{ij}\overline{s_{jk}}=\delta_{ik}.
\ee

Above all, apart from GSD, we obtain two more sets of topological numbers,  $\{\theta_k\}$, and $\{s_{ij}\}$, to characterize the topological phases in our model.

We remark that we have presented here a novel derivation of the modular $\str$ and $\ttr$ matrices, which is purely based on our model and in terms of $3$--cocycles of $G$, without resorting to any theory of group representations.

\subsection{$\str$ and $\ttr$ Matrices}

We now offer concrete solutions of the topological numbers $\{\theta_k,s_{ij}\}$, which are tied to the projective representation theory.
We emphasize that the topological observables $\str$ and $\ttr$ are defined on the subspace $\Hil^{B_f=1}$, whereas the solutions to their eigen--problems are to be obtained on $\Hil^0\subset \Hil^{B_f=1}$.

In the following we diagonalize the $\ttr$ matrix in Eq. \eqref{eq:STequalSumSTx}. One should bear in mind that the transformation $\ttr$ acts non--vanishingly only on the ground states. In Section \ref{subsec:GSbasis}, we see that the ground states are spanned by the orthonormal basis $\ket{A,\mu}$ defined in Eq. (\ref{eq:AmuBasis}), with $A$ running over all conjugacy classes of $G$ and $\mu$ over the irreducible $\beta_{g^A}$--representations of $Z^A$. It turns out that $\ket{A,\mu}$ are eigenvectors of $\ttr$ as we demand, Here we sketch the proof. In the $\ket{A,\mu}$ basis, the action of $\ttr$ becomes
\begin{align}
  \label{eq:TonAmu}
  &\ttr\ket{A,\mu}
  \nonumber\\
  =&\ttr^e P^0\ket{A,\mu}=\ttr^e\ket{A,\mu}
  \nonumber\\
  =&\frac{1}{|G|}\sum_{\substack{g\in C^A\\ h\in Z_g,\nu}}{\widetilde{\chi}^g_{\mu}(h)} \overline{\widetilde{\chi}^g_{\nu}(g^{-1}h)}\ket{A,\nu}\nonumber\\
  =&\frac{\widetilde{\chi}^{g^A}_{\mu}(g^A)}{\text{dim}_{\mu}}\ket{A,\mu},
\end{align}
where $\text{dim}_{\mu}$ is the dimension of the representation $\mu$, the second row uses Eq. (\ref{eq:STandP0}) and Eq. \eqref{eq:P0Amu}, and in the fourth equality use is made of the inverse transformation
\be\label{eq:ghToAmu}
  \ket{g,h}
  =\frac{1}{\sqrt{|G|}}\sum^{r(Z^B,\beta_g)}_{\nu=1}\,
  \overline{\widetilde{\chi}^g_{\nu}(h)}\,\ket{B,\nu},
\ee
which is defined in $\Hil^0$ only, with $g\in C^B$ is assumed. Appendix \ref{app:solModularTrans} proves Eq. \eqref{eq:TonAmu} step by step.

Therefore, the basis vectors $\ket{A,\mu}$ are indeed the eigenvectors of $\ttr$, with the eigenvalues
\be\label{eq:Teigenvalue}
\theta^A_{\mu}=\frac{\widetilde{\chi}^{g^A}_{\mu}(g^A)}{\text{dim}_{\mu}}.
\ee

Clearly, the projective characters are the ground--state wave functions of the system, collapsed in the basis vectors that are the eigenvectors of the $\ttr$ matrix.

The above calculation guides us to interpret $\theta^A_{\mu}$ from  the representation theory  as an invariant that characterizes the representation $\widetilde{\rho}^g_\mu$. More precisely, for any $g\in C^A$, the matrix $\widetilde{\rho}^{g}_{\mu}(g)$ commutes with all other matrices $\widetilde{\rho}^{g}_{\mu}(h)$ for $h\in Z_{g}$,
as\be
\label{eq:rhoggCommutation}
\widetilde{\rho}^{g}_{\mu}(g)\widetilde{\rho}^{g}_{\mu}(h)
=\frac{\beta_g(g,h)}{\beta_{g}(h,g)}
\widetilde{\rho}^{g}_{\mu}(h)\widetilde{\rho}^{g}_{\mu}(g)
=\widetilde{\rho}^{g}_{\mu}(h)\widetilde{\rho}^{g}_{\mu}(g),
\ee
where the second equality can be checked by a direct evaluation in terms of 3--cocycles. From Schur's lemma, the matrix $\widetilde{\rho}^{g}_{\mu}(g)$ is a multiple of the identity matrix
\be
\label{eq:rhoggSchurLemma}
\widetilde{\rho}^{g}_{\mu}(g)=\frac{\widetilde{\chi}^g_{\mu}(g)}{\text{dim}_{\mu}}\mathds{1}.
\ee
Moreover, by setting $h=g$ in Eq. \eqref{eq:CharacterInConjugacyClasses} and using Eq. (\ref{eq:etaggIs1}),
we find that $\chi^{xgx^{-1}}(xgx^{-1})=\chi^g(g)$, indicating that the topological number $\theta^A_{\mu}=\widetilde{\chi}^g_{\mu}(g)/\text{dim}_{\mu}$ is indeed an invariant on $C^A$, associated with the representation $\widetilde{\rho}^{g}_{\mu}$ of $Z^A$.

The topological spin $\theta^A_{\mu}$ are given in terms of 3--cocycles $\alpha$ by
\be\label{eq:thetaomegaA}
(\theta^A_{\mu})^{p_A}=\omega_A
\ee
with
\begin{align}
  \label{eq:omegaA}
  \omega_A\defeq\prod_{n=0}^{p_A-1}\alpha(g,g^n,g)
\end{align}
for conjugacy class $C^A$ of $G$, and $p_A$ the degree, i.e., the least integer such that $g^{p_A}=e$, where $g^{n}$ is the power of $g$. The $\omega_A$ is independent of the choice of $g\in C^A$ and thus a conjugacy class function. This relation is verified by applying Eq. \eqref{eq:betarepresentation} to $(\widetilde{\rho}^A_{\mu})^{p_A}=(\theta^A_{\mu})^{p_A}\mathds{1}$. Therefore, the topological spin $\theta^A_{\mu}$ takes values in the $p_A$--th roots of $\omega_A$. Moreover, each of the $p_A$ distinct $p_A$--th roots appears precisely $r(Z^A,\beta_{g^A})/p_A$ times in $\{\theta^A_{\mu}\}$ for all $\mu$, where the ratio $r(Z^A,\beta_{g^A})/p_A$ is an integer. To see this, one observes that each element in $Z_g$ can be uniquely written as ${g^n}h$ with $n=0,1,...,p_A$ for some $h\in Z_g$, and there are $p_A$ 1--dimensional representation of $Z_g$ by $\rho^j({g^n}h)=\exp(2\pi\ii{j}{n}/p_A)$. Then for each representation $\mu$ there exist a $\mu'$ such that $\widetilde{\rho}^A_{\mu'}(g^nh)=\rho^j(g^nh)\widetilde{\rho}^A_{\mu}(g^nh)$, and thus that $\theta^A_{\mu'}=\exp(2\pi\ii{n}/p_A)\theta^A_{\mu}$.

Similarly, the $\str$--matrix can also be evaluated in terms of the projective characters. We record as follows the final formula for the $\str$--matrix while detail the proof in Appendix \ref{app:solModularTrans}. The $\str$--matrix reads
\be  \label{eq:SonAmuBnu}
\begin{aligned}
s_{(A\mu)(B\nu)}=&\left\langle A,\mu\right|\str \ket{B,\nu}\\
  =&\frac{1}{|G|}\sum_{\substack{{g\in C^A,h\in C^B}\\{gh=hg}}}
  \overline{\widetilde{\chi}^{g}_{\mu}(h)}
  \overline{\widetilde{\chi}^{h}_{\nu}(g)}.
\end{aligned}
\ee
Again, we take the projective representations $\widetilde{\rho}^g_{\mu}$ such that the projective characters are related by Eq. \eqref{eq:CharacterInConjugacyClasses}.
This general result of the $\str$--matrix actually offers an answer to one of the open questions listed in Ref\cite{Coste2000}.

The mathematical significance of the $\str$ and $\ttr$--matrices is the following. They are the invariants carried by the projective representations of $Z^A$ for all conjugacy classes $C^A$, in which the $\beta_g$ functions play a crucial role. All these ingredients are well--organized by the representation theory of a twisted quantum group (or, a twisted Hopf algebra), called the twisted quantum double $D^{\alpha}(G)$ of the finite group $G$, which is parameterized by a 3--cocycle $\alpha$.
All the irreducible $\beta_{g^A}$--representations for all conjugacy class representatives $g^A$ form the linear irreducible representations of $D^{\alpha}(G)$.

Usually, the irreducible representations of a (twisted) quantum group classify the anyonic quasiparticle species. The invariants of each irreducible representation identifies the fractional topological quantum numbers of the corresponding quasiparticle. The $\str$--matrix has the origin as a braiding operation that exchanges any two of these quasiparticles, while the $\ttr$--matrix contains the statistical spins of the corresponding quasiparticles which are determined by the braiding operation. For the discussion of the $\str$ and $\ttr$--matrices for the twisted quantum double $D^{\alpha}(G)$, see Ref\cite{Dijkgraaf1991,Propitius1995,Coste2000}.
We also remark here that twisted quantum double has been used to classify confinement phases in planar physics\cite{Bais2002}.

In this section, we have reproduced from our $\str$ and $\ttr$ operators in terms of $3$--cocycles the familiar $\str$ and $\ttr$--matrices in terms of projective characters for the twisted quantum double $D^{\alpha}(G)$, which were originally obtained from representation theory, according to a braiding operation. Our calculations are carried purely in the ground--state subspace and root in the large transformation of the spatial graph of the system on a torus. We expect that the quasiparticles in the elementary excitations will be classified by the same topological numbers in the way that the GSD equals the number of the quasiparticle species, the $\str$ and $\ttr$--matrices on the ground states are the same as those of the quasiparticles.

\section{Topological Numbers and Topological Phases}\label{sec:classification}
We believe that the topological phases are classified by the topological numbers $\{\mathrm{GSD},\theta^A_{\mu},s_{(A\mu),(B\nu)}\}$. In all examples discussed in Section \ref{sec:examples}, we observe that they are classified by the third cohomology classes of $\alpha$, i.e., any two models $H_{G,\alpha}$ and $H_{G,\alpha'}$ have the same topological numbers if and only if $\alpha$ and $\alpha'$ are equivalent.

In this section, we study how the topological numbers depends on the cohomology classes of $\alpha$.

\subsection{When $3$--Cocycle is cohomologically trivial}\label{subsec:trivial3cocycle}

In Section \ref{subsec:equivModel}, we have shown that two equivalent $3$--cocycles define equivalent twisted quantum double models, which consequently should describe the same topological phase. We now study this topological phase in more details.

We begin with a special case, where the 3--cocycle of our model belongs to class of the trivial $3$--cocycle $\alpha^0$. Such a $3$--cocycle can take the form of a 3--coboundary:
\be
\label{eq:3couboundary}
\alpha(x,y,z)=\delta\beta(x,y,z)=\frac{\beta(y,z)\beta(x,yz)}{\beta(xy,z)\beta(x,y)},
\ee
where $\beta(x,y)$ is any normalized 2--cochain, i.e., any function $G\times G\rightarrow U(1)$ that satisfies $\beta(e,x)=1=\beta(x,e)$ for all $x\in G$. To be seen in Section \ref{sec:Kitaev}, such a model is equivalent to Kitaev's quantum double model.

The corresponding twisted $2$--cocycle $\beta_g$, defined in Eq. (\ref{eq:betafunction}), is automatically trivial such that it has the freedom to be written as a twisted $2$--coboundary:\be\label{eq:equivTbetaDue2Equiv3cocycle}
\beta_g(x,y)=\widetilde{\delta}\epsilon_g(x,y),
\ee
where
\be\label{eq:epsilonDue2Equiv3cocycle}
\epsilon_g(x)=\beta(x,x^{-1}gx)\beta(g,x)^{-1}
\ee
is a twisted $1$--cochain, whose twisted $2$--coboundary reads
\be\label{eq:twisted2coboundaryDue2Equiv3cocycle}
\widetilde{\delta}\epsilon_g(x,y)=\epsilon_g(x)\epsilon_g(xy)^{-1} \epsilon_{x^{-1}gx}(y),
\ee
for all $g,x,y \in G$.
From the relation \eqref{eq:epsilonDue2Equiv3cocycle}, we inevitably notice the following constant:
\be\label{eq:epsilonConst}
\epsilon_g(h)\epsilon_h(g)\equiv 1,\, \forall h\in Z_g, \text{ and }\epsilon_g(g)\equiv1,\, \forall g,h\in G.
\ee
This is indeed a constant as it is clearly independent of which $\alpha$ is picked in its equivalent class.

By the form of $\beta_g$ in Eq. \eqref{eq:equivTbetaDue2Equiv3cocycle}, the irreducible $\beta_g$--representations $\widetilde{\rho}^g_{\mu}$ of $Z_g$ are in one--to--one correspondence to the irreducible linear representations $\rho_{\mu}$, by
\be
\label{eq:BetaRepresToLinearRepres}
\widetilde{\rho}^g_{\mu}(h)=\epsilon_g(h)\rho^g_{\mu}(h)
\ee
for all $h\in Z_g$, by which one can directly check the definition property \eqref{eq:betarepresentation}.

The ground states in Eq. (\ref{eq:AmuBasis}) now become
\be
\label{eq:AmubaisTrivialAlpha}
\ket{A,\mu}=\sqrt{\frac{1}{|G|}}\sum_{g\in C^A,h\in Z_g}\epsilon_g(h)\chi^g_{\mu}(h)\ket{g,h},
\ee
where $\chi^g_{\mu}=\mathrm{tr}\rho^g_{\mu}$ is the usual character.

Eq. \eqref{eq:twisted2coboundaryDue2Equiv3cocycle} implies that $\beta_g(x,y)=\beta_g(y,x)$ for all $x,y\in Z_g$ with $xy=yx$. This means that all elements in $Z_g$ are $\beta_g$--regular and that $\eta_g(h,x)\equiv 1$ for all $x\in Z_{g,h}$. Hence, the GSD in Eq. (\ref{eq:GSDbeta}) now reads
\begin{align}
\label{eq:GSDforTrivialAlpha}
  \text{GSD}=\sum_{g\in G}\sum_{h\in Z_g}\sum_{x\in Z_{g,h}}\frac{1}{|G|} =\left|\frac{\text{Hom}(\pi_1(T^2),G)}{conj}\right|,
\end{align}
where the quotient means the equivalence $(g,h)\sim (xgx^{-1},xhx^{-1})$ for any $x\in G$.

By \eqref{eq:BetaRepresToLinearRepres} and the constraint \eqref{eq:epsilonConst} of the $\epsilon_a$, the topological numbers $\theta^A_{\mu}$ and the $\str$--matrix are expressed by
\be
\label{eq:TTrivialAlpha}
\theta^A_{\mu}
=\frac{\widetilde{\chi}^{g^A}_{\mu}(g^A)}{\text{dim}_{\mu}}
=\frac{\chi^{g^A}_{\mu}(g^A)}{\text{dim}_{\mu}}
\ee
and
\be  \label{eq:STrivialAlpha}
\begin{aligned}
s_{(A\mu)(B\nu)}=\frac{1}{|G|}\sum_{\substack{{g\in C^A,h\in C^B}\\{gh=hg}}}
  \overline{\chi^{g}_{\mu}(h)\chi^{h}_{\nu}(g)}.
\end{aligned}
\ee

When the 3--cocycle is $\alpha^0=1$, the ground states are labeled by the usual irreducible linear representations of all the centralizers $Z^A\subseteq G$. For $\alpha\in [\alpha^0]$ but $\alpha\neq \alpha^0$, the ground states are labeled by projective representations, which are related to the corresponding linear representations by merely a phase, of all the centralizers, since $\beta_g\neq 1$; however, all topological numbers are the same as those in the case of $\alpha^0$, as they should be.

\subsection{When twisted $2$--cocycle is cohomologically trivial}\label{subsec:trivialBeta}

When the $3$--cocycle $\alpha\notin [\alpha^0]$, it could still be ``trivial" at a lower level, in the mathematical sense that the $2$--cocycle $\beta_a$ it defines in Eq. (\ref{eq:betafunction}) is cohomologically trivial,
i.e., this $\beta_a$ is actually a twisted $2$--coboundary:
\be\label{eq:BetaToEpsilon}
  \beta_a(x,y)=\widetilde{\delta}\epsilon_a(x,y),
\ee
for all $a,x,y\in G$. Note that however, the twisted $1$--cochain $\epsilon_a$ in this case does not necessarily have the closed form in Eq. \eqref{eq:epsilonDue2Equiv3cocycle} in general because $\alpha$ is not cohomologically trivial; hence, $\epsilon_g(g)\neq 1$ in general. The twisted $2$--cocycle condition in  Eq. \eqref{eq:betacondition} yields
\be\label{eq:EtaForTrivialBeta}
  \eta^g(h,x)=\frac{\epsilon_{xgx^{-1}}(xhx^{-1})}{\epsilon_g(h)}
\ee
for all $h\in Z_g$ and $x \in G$, which is unity for all $x\in Z_{g,h}$.

Similar to the previous case, the ground--state subspace in the current case are also spanned by the basis vectors $\ket{A,\mu}$ of the form in Eq. \eqref{eq:AmubaisTrivialAlpha}, where $\mu$ labels the $\beta_g$--representations $\widetilde{\rho}^g_{\mu}$ of $Z_g$, which are again related to the usual linear representations $\rho_{\mu}$ by Eq. \eqref{eq:BetaRepresToLinearRepres}.

Since Eq. \eqref{eq:EtaForTrivialBeta} renders all elements in $Z_g$ $\beta_g$--regular, as before, the GSD in this case copies that in Eq. \eqref{eq:GSDforTrivialAlpha}.

By the form \eqref{eq:BetaRepresToLinearRepres}, the topological numbers $\theta^A_{\mu}$ and the $s$ matrix are related to $\epsilon_g(h)$ by,
\be
\label{eq:TTrivialBeta}
\theta^A_{\mu}
=\epsilon_{g^A}(g^A)\frac{\chi^{g^A}_{\mu}(g^A)}{\text{dim}_{\mu}}
\ee
and
\be  \label{eq:STrivialBeta}
\begin{aligned}
s_{(A\mu)(B\nu)}=\frac{1}{|G|}\sum_{\substack{{g\in C^A,h\in C^B}\\{gh=hg}}}
  \overline{\chi^{g}_{\mu}(h)\chi^{h}_{\nu}(g)}\,\overline{\epsilon_g(h)\epsilon_h(g)}.
\end{aligned}
\ee

The GSD is the same as the one in the $[\alpha^0]$--model, a result of the cohomologically trivial $\beta_g$.
Nevertheless, the topological numbers $\theta^A_{\mu}$ and the $\str$--matrix characterize the difference between the current model and the untwisted model.
The Eq. \eqref{eq:TTrivialBeta} exhibits the physical relevance of the phases $\epsilon$: They endow each basis ground state $\ket{A,\mu}$ an extra spin factor $\epsilon_{g^A}(g^A)$ in addition to the topological spins in Eq. \eqref{eq:TTrivialAlpha}.

\subsection{General properties of the topological numbers}\label{subsec:GeneralCase}
Since our model is defined in terms of $3$--cocycles $\alpha$, all of the topological numbers discovered must depend on $\alpha$ at the bottom level. Notwithstanding this, the special cases belabored above imply that in general, some topological numbers may have a higher--level dependence of $\alpha$ via the quantities derived from $\alpha$, such as the equivalence class of $\alpha$ and the twisted $2$--cocycles. We now list as follows the general characteristic properties of the topological numbers $\{\text{GSD},\theta^A_{\mu},s_{(A\mu)(B\nu)}\}$ of our model on torus.
\begin{enumerate}
\item
The set $\{\text{GSD},\theta^A_{\mu},s_{(A\mu)(B\nu)}\}$ of all topological numbers depends on the equivalence class $[\alpha]\in H^3(G,U(1))$ of the $\alpha$ that defines the model.
\item The GSD depends only on the equivalence classes $[\beta_{g^A}]\in H^2(Z^A,U(1))$ for $g^A\in R_C,\ A=1,\dots,r(G)$, independent of the representatives $g^A$.
\item The topological spins $\{\theta^A_{\mu}\}$ are classified by $\{r(Z^A,\beta_A),\omega_A\}$ for all conjugacy classes $C^A$.
\end{enumerate}
We now elaborate on the two properties above in order.

For property 1, one can check that any two equivalent $3$--cocycles $\alpha'$ and $\alpha$ related by $\alpha'=\alpha\delta\beta$, as in Eq. (\ref{eq:equiv3cocycle}), give rise to two equivalent twisted $2$--cocycles $\beta'_a$ and $\beta_a$ related by a twisted $2$--coboundary as follows.
\be\label{eq:equivTbetaDue2Equiv3cocycle}
\beta'_a(b,c)=\beta_a(b,c)\widetilde{\delta}\epsilon_a(b,c),
\ee
where the twisted 2--coboundary $\widetilde{\delta}\epsilon$ happen to be those defined in Eq. \eqref{eq:epsilonDue2Equiv3cocycle} and Eq. (\ref{eq:twisted2coboundaryDue2Equiv3cocycle}).
The constraint  in Eq. (\ref{eq:twisted2coboundaryDue2Equiv3cocycle}) implies that $\text{GSD}'=\text{GSD}$, ${\theta'}^A_{\mu}={\theta}^A_{\mu}$, and $s'_{(A\mu)(B\nu)}=s_{(A\mu)(B\nu)}$.

Property 2 is manifest in Eq. (\ref{eq:GSDconjugacy}), in which the GSD is a sum of the numbers $r(Z^A,\beta_{g^A})$ over all conjugacy classes $C^A$ of $G$. One further confirms this by looking at Eq. (\ref{eq:GSDbeta}), where the GSD is a sum of the phases $\rho^g$ as a function of $\beta_g$ defined in Eq. (\ref{eq:rhoEta}), which is a $1$--dimensional representation of $Z_{g,h}$ with $h\in Z_g$.   Suppose two 3--cocycles $\alpha$ and $\alpha'$, not necessarily inequivalent, that define $\beta_a$ and $\beta'_a$ equivalent for all $a\in G$, in the sense that there exists functions $\epsilon_a:G\rightarrow U(1)$ parameterized by $a\in G$, which satisfies $\epsilon_a(e)=1$, such that $\beta_a$ and $\beta'_a$ are different by merely a twisted $2$--coboundary, namely,
\be\label{eq:BetaBetaEpsilon}
  \beta'_a(x,y)=\frac{\epsilon_a(x)\epsilon_{x^{-1}ax}(y)}{\epsilon_a(xy)}\beta_a(x,y),
\ee
for all $a,x,y \in G$. In the restriction to $x\in Z_a$ and $y\in Z_{a,x}$, one sees that $\beta'_a$ and $\beta_a$ are equivalent up to a usual 2--coboundary over $Z_a$. By Eq. (\ref{eq:BetaBetaEpsilon}) and Eq. \eqref{eq:GSDbeta}, we obtain
\[
\rho'^g(h,x)=\rho^g(h,x),
\]
which verifies that $\mathrm{GSD}'=\mathrm{GSD}$. Furthermore, Eq. (\ref{eq:SumBetaToDelta}) states that the sum of $\rho^g(h,x)$ over $Z_{g,h}$ is either unity or zero, regardless of which representative $g\in C^A$ is chosen. Therefore, property 1 holds, as expected from the analysis in Section \ref{subsec:equivModel}.

Property 3 is straightforward. As discussed in Section \ref{sec:fractionaltopologicalnumbers}, each of the $p_A$ distinct $p_A$--roots appears precisely $r(Z^A,\beta^{g^A})/p_A$ times in $\{\theta^A_{\mu}\}$ for all $\mu$. Then $\omega_{\alpha}=\omega_{\alpha'}$ yields $\theta^A_{\mu}={\theta'}^A_{\mu}$ up to a relabeling of $\mu=0,1,...,r(Z^A,\beta_A)-1$, assuming $r(Z^A,\beta_A)=r(Z^A,\beta'_A)$.

\section{Examples}\label{sec:examples}
In this section, we explicitly compute various examples of our model, making contact with the structure discussed in the previous sections. We declare that some parts of the results here are adapted from certain known results\cite{Dijkgraaf1989,Dijkgraaf1990,Propitius1995,Coste2000} of $3$--cocycles and projective representations that were otherwise discovered in studies of conformal field theory by means of representation theory, which now, however, as we show, become applicable to describing topological phases, owing to the lucid connection revealed by our model between topological phases and group cohomology.
In all examples to be discussed in this section, any two models $H_{G,\alpha}$ and $H_{G,\alpha'}$ have the same GSD and satisfy the condition in Eq. \eqref{eq:epsilonConst} if and only if $\alpha$ and $\alpha'$ are equivalent via Eq. \eqref{eq:equiv3cocycle}.

\subsection{$G=\Z_m$}
When $G$ is the cyclic group $\Z_m$ of order $m$, it is known that the cohomology group is $H^3\left(\Z_m,U(1)\right)\cong \Z_m$, and hence there are $m$ inequivalent classes of 3-cocycles\cite{Moore1989}.

We denote by $a\in\{0,1,...,m-1\}$ the elements of $\Z_m$. The multiplication in $\Z_m$ is $a\cdot b=a+b\mod{m}$.

The $m$ cohomology classes of $3$-cocycles are generated by\cite{Moore1989,Wakui1992,Coste2000}
\be\label{eq:ZmCocycle}
\begin{aligned}
  &\alpha:\Z_m\times \Z_m \times \Z_m\rightarrow U(1)\\
  &\alpha(a,b,c)=\exp\left\{\frac{2{\pi}\ii}{m^2}a\left[b+c-\langle b+c\rangle)\right]\right\},
\end{aligned}
\ee
where $a,b,c\in\Z_m$, and $\langle a\rangle$ is the residue of $a \mod{m}$. By ``generated'' we mean that the $m$ classes of 3-cocycles can be represented by the powers of $\alpha$ in Eq. \eqref{eq:ZmCocycle} as
\be\label{eq:PowerCocycle}
  \{\alpha^k(a,b,c)\vert\; k=0,1,...,m-1\}.
\ee
One verifies that $\alpha^m=1$ and $\alpha^k\alpha^l=\alpha^{k+l}$.

The $\beta_g$ for each $\alpha^k$ has the form \eqref{eq:BetaToEpsilon}, with
\be
  \epsilon_a(b)=\exp\{\frac{2\pi\ii}{m^2}k \times a\times b\}.
\ee
Hence each $\alpha^k$ gives rise to $\text{GSD}=m^2$.

The linear characters are $\chi_{\mu}(a)=\exp(2\pi\ii \mu a/m)$, for $\mu=0,1,...,m-1$. Applying  this to Eqs. \eqref{eq:TTrivialBeta} and \eqref{eq:STrivialBeta} yields
\begin{align}
  &\theta^a_{\mu}=\exp\left[2\pi\ii (k a^2+ma\mu)/m^2 \right]
  \nonumber\\
  &s_{(a\mu),(b\nu)}=\frac{1}{m}\exp\left\{-2\pi\ii[2kab+m(a\nu+b\mu)]/m^2\right\}
\end{align}

\subsection{$G=\Z_2$ and $\Z_3$}\label{subsec:ExZ2Z3}
In the special case where $G$ is the simplest finite group $\Z_2$, there are two classes of 3-cocycles. The first is the trivial one,
namely,
\begin{align}
  \label{eq:Z2Cocycle0}
  \alpha^0(a,b,c)=1,
\end{align}
where $a,b,c=0,1$ are elements of $\Z_2$.

The second one is given as follows, according to Eq. \eqref{eq:ZmCocycle},
\be\label{eq:Z2Cocycle1}
\begin{aligned}
  &\alpha^1(1,1,1)=-1,\\
  &\alpha^1(a,b,c)=1, \text{  for all other }a,b,c
\end{aligned}
\ee

We recognize the $\alpha^0$ model as Kitaev's toric code model, or dual to the Levin-Wen model with the $6j$ symbols determined by irreducible representations of $\Z_2$.
The $\alpha^1$ model is dual to the Levin-Wen model with the $6j$ symbols determined by the semisimple irreducible representations of the quantum group $\mathcal{U}_q(sl(2,\mathds{C}))$ for $q=\exp(\ii\pi/3)$, up to a local unitary transformation (see Section \ref{sec:LW}). These form the complete solutions to Levin--Wen models with $\Z_2$ fusion rule.
The topological spins $\theta^x_{\mu}$ for ground states $\left(\substack{x\\ \mu}\right)$ are given in Table \ref{tab:thetaZ2}.
\\
\begin{table}
\begin{center}
\begin{tabular}{lccccc}
\hline\hline
      & & $\left(\substack{0\\ 0}\right)$ & $\left(\substack{0\\1}\right)$ &  $\left(\substack{1\\ 0}\right)$ &  $\left(\substack{1\\ 1}\right)$ \\\hline
$\alpha^0$  & & 1 & 1 & 1 & $-1$\\
$\alpha^1$  & & 1 & 1 & $\ii$  & $-\ii$\\
\hline
\end{tabular}\caption{ $\theta^x_{\mu}$ for models with $G=\Z_2$}\label{tab:thetaZ2}
\end{center}
\end{table}

When $G=\Z_3$, there are three classes of $3$--cocycles, denoted by $\alpha^0, \alpha^1$, and $\alpha^2$. The $\alpha^0$ model is dual to the Levin--Wen model with the $6j$ symbols determined by the irreducible representations of $\Z_3$ (or equivalent to Kitaev's model with $\Z_3$), whereas the $\alpha^1$ and $\alpha^2$ models are not dual to any Levin--Wen models with the $\Z_3$ fusion rule, as there is only one Levin--Wen model in the circumstance. The topological spins $\theta^x_{\mu}$ are also tabulated in Table \ref{tab:thetaZ3}.
\\
\begin{table}
\begin{center}
\begin{tabular}{lcccccccccc}
\hline\hline
       & & $\left(\substack{0\\ 0}\right)$ & $\left(\substack{0\\1}\right)$ &  $\left(\substack{0\\ 2}\right)$ &  $\left(\substack{1\\ 0}\right)$ & $\left(\substack{1\\ 1}\right)$ & $\left(\substack{1\\ 2}\right)$ & $\left(\substack{2\\ 0}\right)$ & $\left(\substack{2\\ 1}\right)$ & $\left(\substack{2\\ 2}\right)$ \\\hline
$\alpha^0$  & & 1 & 1 & 1 & 1 & $\text{e}^{\frac{2\pi\ii}{3}}$ & $\text{e}^{-\frac{2\pi\ii}{3}}$ & 1 & $\text{e}^{-\frac{2\pi\ii}{3}}$ & $\text{e}^{\frac{2\pi\ii}{3}}$ \\
$\alpha^1$  & & 1 & 1 & 1  & $\text{e}^{\frac{2\pi\ii}{9}}$ & $\text{e}^{\frac{8\pi\ii}{9}}$ & $\text{e}^{-\frac{4\pi\ii}{9}}$ & $\text{e}^{\frac{8\pi\ii}{9}}$ & $\text{e}^{\frac{2\pi\ii}{9}}$ & $\text{e}^{-\frac{4\pi\ii}{9}}$ \\
$\alpha^2$  & & 1 & 1 & 1  & $\text{e}^{\frac{4\pi\ii}{9}}$ & $\text{e}^{-\frac{8\pi\ii}{9}}$ & $\text{e}^{-\frac{2\pi\ii}{9}}$ & $\text{e}^{-\frac{2\pi\ii}{9}}$ & $\text{e}^{-\frac{8\pi\ii}{9}}$ & $\text{e}^{\frac{4\pi\ii}{9}}$ \\
\hline
\end{tabular}\caption{ $\theta^x_{\mu}$ for models with $G=\Z_3$}\label{tab:thetaZ3}
\end{center}
\end{table}

Consider the complex conjugation $\mathcal{K}$. The models $H_{\Z_2,\alpha^0}$, $H_{\Z_2,\alpha^1}$ and $H_{\Z_3,\alpha^0}$ are invariant under $\mathcal{K}$, whereas the models $H_{\Z_3,\alpha^1}$ and $H_{\Z_3,\alpha^2}$ are not.  The complex conjugation $\mathcal{K}$ transforms the topological spins in the $H_{\Z_3,\alpha^1}$ model to those in the $H_{\Z_3,\alpha^2}$ model, as seen in Table \ref{tab:thetaZ3}. The $H_{\Z_3,\alpha^1}$ and $H_{\Z_3,\alpha^2}$ models are the simplest models that break the complex conjugation symmetry which persists in the Levin--Wen models.

\subsection{$G=\Z_m \times \Z_m$}
The simplest non--cyclic Abelian group is $G=\Z_m^2$. But
nothing is really new compared with the case where $G=\Z_m$.
The cohomology group $H^3(\Z_m^2,U(1))=\Z_m^3$ has three generators.
We will label the group elements in $G=\Z_m^2$ as pairs $a=(a_1,a_2)$
with $a_1,a_2=0,1,...,m-1$. The multiplication is the obvious one
$(a_1,a_2)(b_1,b_2)=(a_1+b_1\text{ mod }m,a_2+b_2\text{ mod }m)$.
The three 3--cocycle generators are
\be  \label{eq:ZmZmgenerators}
\begin{aligned}
  &\alpha_I^{(1)}(a,b,c)=\exp\{\frac{2{\pi}\ii}{m^2}a_1(b_1+c_1-\langle b_1+c_1\rangle)\},\\
  &\alpha_I^{(2)}(a,b,c)=\exp\{\frac{2{\pi}\ii}{m^2}a_2(b_2+c_2-\langle b_2+c_2\rangle)\},\\
  &\alpha_{II}^{(12)}(a,b,c)=\exp\{\frac{2{\pi}\ii}{m^2}a_1(b_2+c_2-\langle b_2+c_2\rangle)\},
\end{aligned}
\ee
where $\langle x\rangle=x \text{ mod }m$ is the residue of $x$.
The $m^3$ classes of 3--cocycles are the products of powers of these three
generators.
The $\beta_a$ function for all these three generators has the form \eqref{eq:BetaToEpsilon}, with
\be  \label{eq:ZmZmEpsilon}
\begin{aligned}
  &{(\epsilon_I^{(1)})}_{a}(b)=\exp\{\frac{2{\pi}\ii}{m^2}a_1 \times b_1\},\\
  &{(\epsilon_I^{(2)})}_{a}(b)=\exp\{\frac{2{\pi}\ii}{m^2}a_2 \times b_2\},\\
  &{(\epsilon_{II}^{(12)})}_{a}(b)=\exp\{\frac{2{\pi}\ii}{m^2}a_1 \times b_2\}.
\end{aligned}
\ee
Therefore the associated $\beta_a$ for all 3--cocycles are equivalent,
and correspond to the trivial element in $H^2(\Z_m^2,U(1))=\Z_m$,
though the second cohomology group itself is non-trivial.

We conclude the models specified by all 3--cocycles have $\text{GSD}=m^4$.

\subsection{$G=\Z_m \times \Z_m \times \Z_m$}

When it comes to the case of $G=\Z_m^3$, things become more interesting.
We label the group elements by triples $a=(a_1,a_2,a_3)$ with
$a_1,a_2,a_3=0,1,...,m-1$.
The cohomology group $H^3(\Z_m^3,U(1))=\Z_m^7$ has seven generators,
\be\label{eq:ZmZmZmgenerators}
\begin{aligned}
  &\alpha_I^{(i)}(a,b,c)=\exp\{\frac{2{\pi}\ii}{m^2}a_i(b_i+c_i-\langle b_i+c_i\rangle)\}, \\
  &\alpha_{II}^{(ij)}(a,b,c)=\exp\{\frac{2{\pi}\ii}{m^2}a_i(b_j+c_j-\langle b_j+c_j\rangle)\}, \\
  &\alpha_{III}(a,b,c)=\exp\{\frac{2{\pi} \ii}{m}a_1 b_2 c_3\},
\end{aligned}
\ee
where $1\leq i\leq 3$ and $1\leq i\leq j\leq 3$ are assumed respectively in the first two lines, and $\langle x\rangle$ is the residue of $x \mod{m}$.

The $\beta_a$ function for the first two types has the form of Eq. \eqref{eq:BetaToEpsilon}, with
\be  \label{eq:ZmZmZmEpsilon}
\begin{aligned}
  &{(\epsilon_I^{(j)})}_{a}(b)=\exp\{\frac{2{\pi}\ii}{m^2}a_j \times b_j\},\\
  &{(\epsilon_{II}^{(jk)})}_{a}(b)=\exp\{\frac{2{\pi}\ii}{m^2}a_j \times b_k\},
\end{aligned}
\ee

But $\alpha_{III}$ cannot be decomposed as in Eq. \eqref{eq:BetaToEpsilon}.
This provides further classification of the models defined
by the $m^7$ 3--cocycle representatives.
Precisely speaking, each 3--cocycle class has a representative of the form
of a product of powers of the seven generators.
They are classified into $m$ classes, depending on the power of $\alpha_{III}$
they contain, namely the form $\alpha_{III}^q$ for $q=0,1,...,m-1$.
The 2--cocycle for each $\alpha_{III}^q$ is given by
\begin{align}
  \label{eq:ZmZmZmIIIbeta}
  \beta_a(b,c)=\exp\{\frac{2{\pi}\ii}{m}q(a_1b_2c_3-b_1a_2c_3+b_1c_2a_3)\}.
\end{align}
If $\alpha_1$ and $\alpha_2$
belong to two different classes with the powers $q_1$ and $q_2\neq q_1$,
the corresponding models have inequal GSD in general.

Then the GSD in Eq. (\ref{eq:GSDalpha}) for the model specified by $\alpha$ depends
on the class only, i.e., the power $q$ of $\alpha_{III}$.
Specifically, the GSD is determined by Eq. \eqref{eq:GSDbeta}, with
$\beta_a$ given by Eq. \eqref{eq:ZmZmZmIIIbeta}.
According to the analysis in Section \ref{subsec:GSDtorus}, computing the GSD
amounts to count the number of $\beta_a$--regular conjugacy classes of
$Z_a$, for each conjugacy class representative $a$ of $G$, then sum over all
$a$. The GSD is
\begin{align}
  \label{eq:GSDZmZmZm}
  \text{GSD}=\frac{m^6}{f^3}
  \prod_p \left[(p^{k_p}-1)(1+p^{-1}+p^{-2})+1\right],
\end{align}
where $f=m/\gcd(q,m)$ is the greatest common divisor of $q$ and $m$, $p$ and $k_p$
are the prime number and the corresponding power in the prime decomposition
$f=\prod_p p^{k_p}$. The GSD in $m=2,3,4,5$ cases are given in Table \ref{tab:GSDZmZmZm},
which can be computed either from the original formula Eq. \eqref{eq:GSDalpha}
or the final reduced formula Eq. \eqref{eq:GSDZmZmZm}.
\\
\begin{table}
\begin{center}
\begin{tabular}{lcccccc}
\hline\hline
       & & $q=0$ & $q=1$ &  $q=2$ &  $q=3$ &  $q=4$ \\\hline
$m=2$  & & 64 & 22 & & &\\
$m=3$  & & 729 & 105 & 105 & &\\
$m=4$  & & 4096 & 400 & 1408 & 400 &\\
$m=5$  & & 15625 & 745 & 745 & 745 & 745\\
\hline
\end{tabular}\caption{ GSD with $G=\Z_m\times\Z_m\times\Z_m$}\label{tab:GSDZmZmZm}
\end{center}
\end{table}

The models with $q=0$, i.e. with only $\alpha_I$ and $\alpha_{II}$ as the defining $3$--cocycles, possess similar topological numbers $\{\text{GSD},\theta,s\}$ as those in the previous examples for $G=\Z_m$ and $G=\Z_m\times\Z_m$. The $\theta$ and $S$ are derived from the linear characters
\be
\label{eq:ZmZmZmCharacters}
\chi_{\mu}(x)=\exp\left\{\frac{2\pi \ii}{m}(\mu_1 x_1+\mu_2 x_2+\mu_3 x_3)\right\},
\ee
with $\mu_1,\mu_2,\mu_3=0,1,...,m-1$ labeling the irreducible representations of $G=\Z_m\times\Z_m\times\Z_m$.

The models with $q\neq 0$, i.e., involving $\alpha_{III}$ in the defining $3$--cocycle, possess more interesting topological numbers. Though the finite group $G$ is Abelian, the topological charges of the ground states are non-Abelian.

\subsection{$G=\Z_m^n$}
We now study the Abelian non-cyclic group $\Z_m^n$ for some integer $m$ and $n$ more generally, whose special cases where $n\leq 3$ were investigated in previous sub sections. For $n>3$,
things are similar to $G=\Z_m^3$ case. The second and the third
cohomology groups are
\begin{align}
\label{ZmnCohomologyGroups}
&H^2(\Z_m^n,U(1))\simeq \Z_m^{n(n-1)/2},
\nonumber\\
&H^3(\Z_m^n,U(1))\simeq \Z_m^{n+n(n-1)/2+n(n-1)(n-2)/3!}.
\end{align}
Like in the $\Z_m^3$ case, there are three types of 3--cocycles
taking the following form
\be\label{eq:Zmngenerators}
\begin{aligned}
  &\alpha_I^{(i)}(a,b,c)=\exp\{\frac{2{\pi}\ii}{m^2}a_i(b_i+c_i-\langle b_i+c_i\rangle)\}, \\
  &\alpha_{II}^{(ij)}(a,b,c)=\exp\{\frac{2{\pi}\ii}{m^2}a_i(b_j+c_j-\langle b_j+c_j\rangle)\}, \\
  &\alpha_{III}^{(ijk)}(a,b,c)=\exp\{\frac{2{\pi} \ii}{m}a_i b_j c_k\},
\end{aligned}
\ee
where $1\leq i\leq 3$, $1\leq i\leq j\leq 3$, and $1\leq i\leq j\leq k\leq 3$ are assumed respectively in the three lines. the number of three types of generators are $n$, $n(n-1)/2$, and $n(n-1)(n-2)/3!$,
corresponding to the number of generators in $H^3(\Z_m^n,U(1))$.
Topological phases are classified by the elements in $H^3(\Z_m^n,U(1))$.
The 2--cocycles $\beta_a$ obtained from type--III generators correspond
to non-trivial elements in $H^2(\Z_m^n,U(1))$.

\subsection{$G=D_m$ for odd $m$}

The simplest non-Abelian finite groups are the dihedral groups $D_m$.
Specifically, $D_3$ (equivalent to the permutation group $S_3$)
is the simplest non-Abelian group.
We will only consider odd $m$  here, in which all 3--cocycles
can be decomposed as in \eqref{eq:BetaToEpsilon}.

We will label the elements in $D_m$ by pairs $(A,a)$
for $A=0,1$ and $a=0,1,...,m-1$. The multiplication law
takes the form
\be\label{eq:DmMultiplication}
  (A,a)(B,b)=({\langle A+B\rangle}_2,{\langle (-1)^B a+b\rangle}_m)
\ee
where ${\langle x\rangle}_2=x \text{ mod }2$ and ${\langle x\rangle}_m=x \text{ mod }m$ means taking the residue.

The cohomology group $H^3(D_m,U(1))=\Z_{2m}$ has only one generator of the 3--cocycles:
\begin{widetext}
\begin{align}
  \label{eq:Dmgenerator}
  \alpha\bigl((A,a),(B,b),(C,c)\bigr)
  =\exp\left\{\frac{2{\pi}\ii}{m^2}
  \left[(-1)^{B+C}a\left[(-1)^C b+c-{\langle(-1)^C b+c\rangle}_m\right] +\frac{m^2}{2}ABC\right]\right\}.
\end{align}
\end{widetext}

The representatives of each 3--cocycle class takes the form $\alpha^p$ for $p=0,1,...,2m-1$.

The $\beta_a$ for this 3--cocycle generator takes the form \eqref{eq:BetaToEpsilon}, with
\begin{align}
  \label{eq:DmEpsilon}
  &\epsilon_{(A,a)}\bigl((B,b)\bigr)
  \nonumber\\
  =&\exp\left(\frac{2{\pi}\ii}{m^2}\left\{b\left[(-1)^B a+2Ab\right]-Ab^2+\frac{m^2}{4}AB\right\}\right),
\end{align}
for all $(A,a)$ and $(B,b)$ that satisfy $(A,a)(B,b)=(B,b)(A,a)$.
The GSD is the same for all 3--cocycles, and are given by $\text{GSD}=\tfrac{m^2+7}{2}$.

\section{Kitaev's quantum double model: $\alpha$ is trivial}\label{sec:Kitaev}

In this section, we show that in the special case where the
3-cocycle is trivial, our model becomes the Kitaev's quantum
double(QD) model. By ``trivial'' we mean that the 3--cocycle takes the constant value 1,
\begin{equation}
  \label{TrivialCocycle}
  \alpha^0(x,y,z)=1, \text{ for all }x,y,z\in G
\end{equation}

With this $\alpha^0$, the definition
\eqref{eq:Avg} of $A_v$ operator is reduced to
\be\label{BphAlpha0}
A_{v_3}\BLvert \threeTriangles{v_1}{v_2}{v_3}{v_4}{2} \Brangle
=\BLvert \threeTriangles{v_1}{v_2}{v'_3}{v_4}{2} \Brangle,
\ee

Then the model defined by \eqref{eq:Hamiltonian} and \eqref{IndexByLink}
becomes the familiar Kitaev's quantum double model on triangle graphs. With a nontrivial 3-cocycle $\alpha$, our model can be
viewed as the twisted version of Kitaev's QD model, where
the twisting is specified by the 3-cocycle $\alpha$.
We will explain the twisting
in more detail in the next section.

To gain more intuition, we would like to briefly
review Kitaev's quantum double
model in the language of gauge theory.

To set up a gauge theory on the graph $\Gamma$,
we need to specify the connections and the gauge transformations.
Each basis vector in
Eq. \eqref{IndexByLink} corresponds to
a connection, namely, an assignment $g:E\rightarrow{G}$ to each edge $e$ of $\Gamma$
a group element $g_e$ of $G$.
A gauge transformation $h$ on $\Gamma$ is an
assignment to each vertex $v$ a group element $h_v$ of $G$.
The action $L(h)$ of a gauge transformation $h$ on a connection $g$
is given by $[L(h) g]_e=g_{t(e)}g_e{g^{-1}_{s(e)}}$ for each $e$,
where $s(e)$ and $t(e)$ are the starting and ending vertices
of the edge $e$.
For example, on one edge $e$ orienting from $v_1$ to $v_2$,
the action of a gauge transformation $h$
is
\begin{align}
  \label{GaugeTransform}
  L(h): \BLvert\HorGaugeTransformYY{0.5}{v_1}{v_2}{g_e}\Brangle
  \rightarrow \BLvert\HorGaugeTransformYY{0.5}{v_1}{v_2}{h_{v_2}g_e h_{v_1}^{-1}}\Brangle
\end{align}

The action of any gauge transformation can be decomposed into
local operators defined at each vertex.
We denote by $L_v(h_v)$ the action of a
local gauge transformation of at vertex $v$, which is defined as,
\begin{align}
  \label{LocalGaugeTransform}
  L_v(h_v): \BLvert\TriVertex{g_1}{g_2}{g_3}\Brangle
  \rightarrow \BLvert\TriVertex{h_v g_1}{h_v g_2}{h_v g_3}\Brangle
\end{align}

The Hamiltonian of Kitaev's QD model is
\be
H=-\sum_{v\in{V}}A_v-\sum_{f\in{F}}B_f
\ee

It includes two types of local operators $A_v$ and $B_f$.
The operator $A_v$ at vertex $v$ defined by
\be A_v=|G|^{-1}\sum_{h_v\in{G}}L_v(h_v),\ee
is an average of all local gauge transformations at $v$.
This is the same as Eq. \eqref{BphAlpha0}
By checking that
$L_v(h'_v)A_v=|G|^{-1}\sum_{h_v\in{G}}L_v(h'_vh_v)=|G|^{-1}\sum_{h_v\in{G}}L_v(h_v)=A_v$,
we see this is the projector that projects onto states that are invariant
under the local gauge transformation $L_v(h'_v)$ at vertex $v$ for any $h'_v\in{G}$.
Therefore, $A_v$ prefers gauge symmetry at vertex $v$. While the gauge symmetry broken
states are allowed, it costs a energy of 1 to break the gauge symmetry.
An important consequence of this gauge symmetry breaking is that
a quantum number emerges at vertex $v$, and it is classified by
the representations of the gauge group $G$.
This quantum number identifies a quasiparticle at vertex $v$.
The group element $g_e$ represent the action on the states of
the parallel transport of this emergent quasiparticle
along the edge $e$ of the graph.

The operator $B_f$ on face $f$ is defined via
\begin{align}
  \label{ActionOfBfInKitaevModel}
  B_f\BLvert \oneTriangleInKitaevModel \Brangle
  =\delta_{g_1\cdot g_2 \cdot g_3}
  \BLvert \oneTriangleInKitaevModel\Brangle,
\end{align}
which is the same as Eq. \eqref{eq:actionOfBf}.

Here $g_1 g_2 g_3$ is the holonomy around the face $f$, and the delta function $\delta_{a}=1$ if the group element $a$ equals the identity element $e$ in $G$ and $0$ otherwise. The delta function can be expanded in terms of characters,
\be
\delta_{g}=|G|^{-1}\sum_{\rho\in{\mathrm{Irrep}(G)}}\dim_{\rho}\chi^{\rho}(g),
\nonumber
\ee
where $\mathrm{Irrep}(G)$ is the set of all irreducible representations of $G$,
$\dim_{\rho}$ the dimension of the representation $\rho$,
and $\chi^{\rho}$ the character of the $\rho$.
Thus $B_f$ is a projector that measures whether the holonomy around the face $f$ is trivial or not.

Returning to the cases where $\alpha$ is in general nontrivial,
our model \eqref{eq:Hamiltonian} may be viewed as the twisted version
of Kitaev's QD model.
In this interpretation,
$A_v^g$ is the action of the twisted gauge transformation at $v$, and
the $A_v$ is the average of all local twisted gauge transformations.
To make this interpretation precise, we need to study the
algebra of all local operators, which is the main task of the next section.


\section{Relation to Dijkgraaf--Witten Topological gauge theory and conformal field theory}\label{sec:relation2DW}
In this section, we dwell on the relation between our model, a lattice realization due to Dijkgraaf and Witten  of  topological Chern--Simons gauge theories, and conformal field theories.

We begin with a quick review of the gist of the part of Dijkgraaf--Witten gauge theories that is relevant to our model. In Ref\cite{Dijkgraaf1990}, Dijkgraaf and Witten established a correspondence between the three dimensional Chern--Simons gauge theories with a compact gauge group $G$ and the two dimensional sigma models with Wess-Zumino interactions of the group $G$, in the sense that there is a natural map from the cohomology group $H^4(BG,\Z)$, which classifies the Chern--Simons theories, and the group $H^3(G,\Z)$, which classifies the Wess--Zumino interactions. The classifying space of the group $G$ is denoted by $BG$. In general, the prescription of the topological action of a three dimensional Chern--Simons gauge theory is rather abstract; however, in view of that $H^4(BG,\Z)$ is isomorphic to $H^3(BG,U(1))$ when $G$ is finite, Dijkgraaf and Witten constructed a concrete lattice realization of the topological action in the case of finite gauge groups. From now on in this section, we restrict the discussion to finite groups only.

So, more precisely, consider a topological gauge theory defined in a three dimensional manifold $M$, with a finite gauge group $G$, the lattice realization is defined on a three-skeleton, i.e., a triangulation $\mathbf{T}$, of $M$, with a group element of $G$ living on each $1$--simplex, which is oriented, of the triangulation (see Fig. \ref{fig:tetrahedron}). The topological partition function of such a lattice gauge theory reads
\be\label{eq:DWpartition}
\mathcal{Z}(\mathbf{T}(M),G)=\frac{1}{|G|}\prod_i W(\mathbf{T}_i)^{\varepsilon_i},
\ee
where the product runs over all of the tetrahedra $\mathbf{T}_i$ in the triangulation $\mathbf{T}(M)$, and $\varepsilon_i$ is a sign, $+1$ or $-1$, depending on whether the four vertices of the corresponding tetrahedron are in a right--handed arrangement or left--handed arrangement. It is shown\cite{Dijkgraaf1990} that the $W(\mathbf{T}_i)$ associated with each tetrahedron $\mathbf{T}_i$ is a $3$--cocycle over $G$. For example, for the tetrahedron in Fig. \ref{fig:tetrahedron}, $W(T)=\alpha(g,h,k)$.

\begin{figure}[h!]
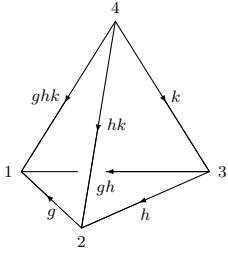

\centering
\tetrahedron{1}{2}{3}{4}{g}{h}{k}
\caption{A tetrahedron whose edges are graced with group elements; group multiplication rule applies to each of the four triangles. The corresponding $3$--cocycle is $\alpha(g,h,k)$.}
\label{fig:tetrahedron}
\end{figure}

Note that $\alpha$ is an equivalence class, any two representatives of the class are related by a $3$--coboundary. If the manifold $M$ is closed, the value of $\mathcal{Z}(\mathbf{T}(M),G)$ does not depend on the choice of the representative of an equivalent class of $3$--cocycles. The partition function is also invariant under the Pachner moves that connect two simplicial triangulations of $M$. Another remark is that the labeling of the vertices of the tetrahedron in Fig. \ref{fig:tetrahedron} is fixed once for all.
At this point, the partition function does not have tetrahedral symmetry.

On closed manifolds, the partition function in Eq. (\ref{eq:DWpartition}) also has a gauge invariance. Consider the single tetrahedron in Fig. \ref{fig:tetrahedron} as an example, the gauge transformation that acts on vertex $1$ transforms the topological action as follows.
\be\label{eq:gaugeTransDW}
W(\mathbf{T})\rightarrow W'(\mathbf{T})=\frac{\alpha(c,g,h)\alpha(c,gh,k)}{\alpha(c,g,hk)}W(\mathbf{T}),
\ee
where $c\in G$ is the gauge parameter. This gauge transformation can be understood topologically as in Fig. \ref{fig:gaugeTransform}.
\begin{figure}[h!]
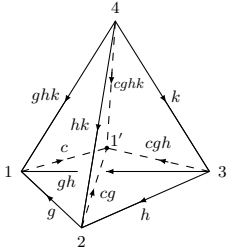

\centering
\tetrahedronSplit{1}{2}{3}{4}{g}{h}{k}{1}{c}
\caption{Gauge transformation acting on the vertex $1$ in Fig. \ref{fig:tetrahedron}; a new vertex $1'$ is created at the barycenter, accompanied by a new group element $c$ on the edge $[1'1]$. Also understood as the $1\rightarrow 4$ Pachner move.}
\label{fig:gaugeTransform}
\end{figure}
The tetrahedron $[1234]$ is associated with the original topological action $W(\mathbf{T})$. The gauge transformation acting on vertex $1$ creates a new vertex $1'$ inside the tetrahedron (can be thought as being at the barycenter) and thus created four new tetrahedra, of which the tetrahedron $[0'123]$ is associated with the new topological action $W'(\mathbf{T})=\alpha(cg,h,k)$. There are five tetrahedra all told in Fig. \ref{fig:gaugeTransform}, associated with which the five $3$--cocycles satisfy the $3$--cocycle condition and thus lead to Eq. (\ref{eq:gaugeTransDW}). On closed manifolds, the topological action is invariant under the gauge transformation Eq. (\ref{eq:gaugeTransDW}) because the factors on the RHS of Eq. (19) can be canceled by those produced by the gauge transformation on the neighbouring tetrahedra. Topologically, the gauge transformation behaves like a $1\rightarrow 4$ Pachner move that splits a tetrahedron at it barycenter into four tetrahedra. Such a Pachner move can be visualized only in four--dimension, whereas Fig. \ref{fig:gaugeTransform} is the three--dimensional projection of a $4$--simplex whose five boundary $3$--simplices are the five tetrahedra in the figure.

To gain a deeper understanding of the gauge transformation Eq. (\ref{eq:gaugeTransDW}), let us rewrite the equation in terms of $3$--cocycles only as follows.
\be\label{eq:gaugeTransform3cocycle}
\begin{aligned}
\alpha(g,h,k) &\rightarrow \alpha'(g,h,k)\\
&=\alpha(cg,h,k)=\gamma(c,g,h,k)\alpha(g,h,k),
\end{aligned}
\ee
where we define
\[\gamma(c,g,h,k)=\frac{\alpha(c,g,h)\alpha(c,gh,k)}{\alpha(c,g,hk)}.
\]
In general, however, the new object $\alpha'(g,h,k)$ is not a $3$--cocycle any longer because one can check that it does not meet the $3$--cocycle condition Eq. (\ref{3CocycleCondition}).
Nonetheless, that $\alpha'$ is not a $3$--cocycle makes it possible to choose a convenient gauge such that the prescription of the topological partition function becomes simpler. Indeed, according to Ref\cite{Dijkgraaf1990}, depending on the divisibility of $|G|$, the following gauge of the $3$--cocycles may be imposed.
\be\label{eq:normalization}
\alpha'(g,g^{-1},h)=\alpha'(g,h,h^{-1})=1.
\ee
Under this gauge, the ordering of the vertices of a tetrahedron is irrelevant; in other words, the topological action $W'(\mathbf{T})$ acquires tetrahedral symmetry, in the sense that it is invariant under the change of the labeling of the vertices. Therefore, the $3$--cocycle condition and the gauge in Eq. \eqref{eq:gaugeTransform3cocycle} are incompatible unless the $3$--cocycle under consideration is equivalent to the trivial one, namely $\alpha\in[\alpha^0]$.

If the manifold $M$ has a boundary (open or closed), however, the gauge transformation in Eq. (\ref{eq:gaugeTransDW}) ceases to apply because a boundary condition must be imposed on $M$, which fixes the boundary value of the embedding of $M$ into the classifying space $BG$ and hence forbids the Pachner moves that involve the boundary simplices. Effectively, there are now degrees of freedom that cannot be gauged away on the two--dimensional boundary $\partial M$ of $M$. As such, the three--dimensional partition function turns out to be the wave function of the corresponding boundary state at certain time.
Since the Dijkgraaf--Witten Chern--Simons theory is a topological gauge theory, there is no nonvanishing Hamiltonian due to Legendre transform that can enable the notion of ground and excited states for the boundary states. Instead, here we have only gauge--invariant and noninvariant states. In particular, however, if the manifold $M$ has a closed boundary, e.g., a solid torus with its boundary a 2--torus, the remaining boundary states are automatically only those gauge--invariant ones. This can be understood by the standard technique of "gluing" and "sewing" if a topological quantum field theory, which in the current case gives the size of the Hilbert space on the boundary. As such, the Dijkgraaf--Witten partition function becomes the dimension of the Hilbert space of the gauge--invariant states on the closed boundary of $M$.

At this point, one may ask if it is possible to construct a Hamiltonian on the closed boundary of $M$ whose ground states happen to be the gauge--invariant boundary states of the Dijkgraaf--Witten theory in $M$. Yes, the Hamiltonian of our TQD model turns out to be a positive answer to this question, as explained as follows.

Staring at Fig. \ref{fig:tetrahedron} again as if it is a triangulation of a $3$--ball, then the the right--to--left projection of the four triangles comprising the boundary $2$--sphere to the paper plane is the very graph in Fig. \ref{fig:3cocycleB}, a basis graph of our twisted quantum double model. Note that in this case, there are four triangles in Fig. \ref{fig:3cocycleB}, including the triangle $[124]$ in the back. When the corresponding state of the graph in Fig. \ref{fig:3cocycleB} is a ground state, the $B_f$ operator is unity acting on any of the triangles, complying with the group multiplication rule on each of the triangles in Fig \ref{fig:tetrahedron}.

By comparing $A^c_3$ with $c\in G$ at vertex 3 on the graph in Fig. \ref{fig:3cocycleB} to Eq. \ref{eq:gaugeTransDW}, one may see the vertex operator $A^g_v$ in the Hamiltonian of our model is formally identical to the gauge transformation on the partition function in Dijkgraaf--Witten theory. In our model, the operators $A_v$ evolve the states; however, we have $A_v=\mathds{1}$ on the ground states, which implies that they are invariant under the Dijkgraaf--Witten gauge transformation in Eq. (\ref{eq:gaugeTransDW}). If we embed the graph of the TQD model on a torus or any surface homeomorphic to it, say, $\partial X\times S^1$, where $X$ is homeomorphic to a disk, we soon see that the ground states of the TQD model are the gauge--invariant boundary states of the Dijkraaf--Witten theory defined on $M=X\times S^1$. Therefore, we can conclude that
\[
\mathcal{Z}_{CS}\left(\mathbf{T}(X\times S^1),G\right)=\mathrm{GSD}_{TQD} \left(T(\partial X\times S^1)\right).
\]
The results obtained in Sections \ref{sec:GSD} and \ref{sec:fractionaltopologicalnumbers} fall into this latter case, which can be verified by comparing them to the corresponding results in Ref\cite{Dijkgraaf1990,Coste2000}.

Now that we have gone through the logic of this correspondence, we can claim that our twisted quantum double models may indeed be regarded as a valid Hamiltonian extension of the Dijkgraaf--Witten discrete Chern--Simons theories. Although the discussion so far is restricted to $3$--dimension, the correspondence described above may be readily generalizable to higher dimensions.

In light of this correspondence, the gauge transformation in Eq. \eqref{eq:gaugeTransform3cocycle} implies that all twisted quantum double models but the one defined by $[\alpha^0]$ do not have tetrahedral symmetry.

On the other hand, the CS theories in a $3$--manifold $M$ also correspond to other two--dimensional theories on the boundary $\partial M$, namely the rational conformal field theories (RCFT), such as Wess--Zumino--Witten (WZW) models. The salient point of this correspondence is that the CS Hilbert space on $\partial M$ is isomorphic by a canonical identification to the space of holomorphic conformal blocks of the RCFT on $\partial M$, while the CS wave--function reproduces the fusion algebra of the holomorphic sector of the RCFT. Note that this correspondence is level by level, in the sense that the CS theory and the WZW interaction of the corresponding RCFT are both at the same level, say, $k$.

What follows naturally is a correspondence between our twisted quantum double models  with a finite group $G$ and a type of RCFTs, namely the CFTs as orbifolds by the group $G$ of a holomorphic CFT. The modular data of a $G$--orbifold is twisted by a $3$--cocycle over $G$. This correspondence is also level by level, in the following sense. The third cohomology group over $G$, $H^3(G,U(1))$, is a discrete group; in particular it is $\Z$ if $G$ is a compact Lie group. Hence, one can label the equivalence classes as the elements of $H^3(G,U(1))$ by integers, say, $[k]$. The situation is similar when $G$ is finite, in which case $H^3(G,U(1))$ also becomes finite, indicating that there are only finite number of levels available. Hence, a twisted quantum double model defined by a $3$--cocycle $\alpha\in [k]$ corresponds to a $G$--orbifold twisted by the same $\alpha$, which has a twisted modular data at level $k$ too. As such, when $\alpha=\alpha^0\in [0]$, the twisted quantum double model is actually untwisted, which is equivalent to Kitaev model, and thus corresponds to the usual untwisted $G$--orbifold.

Let $H_{G,\alpha}$ be a twisted quantum double model on a torus and $\mathcal{C}_{G,\alpha}$ a twisted toroidal orbifold of a holomorphic CFT $\mathcal{C}$. One can check that they correspond to each other in the respects tabulated row by row as follows.
\begin{table}[h!]
\begin{center}
\begin{tabular}{|c|c|}\hline
$H_{G,\alpha}$ & $\mathcal{C}_{G,\alpha}$ \\\hline
States: $\ket{g,h}$ & conformal blocks: $\cb{h}{g}$ \\\hline
Ground State: $\ket{A,\mu}$ & 1--loop characters: $\varkappa^A_{\mu}$ \\\hline
Ground state degeneracy & Number of primary fields \\\hline
$\str$ \& $\ttr$ matrices in $\ket{A,\mu}$ basis & $\str$ \& $\ttr$ matrices in $\varkappa^A_{\mu}$ basis \\\hline
\end{tabular}
\caption{Correspondence between a TQD and a twisted $G$--orbifold.}\label{tab:TQDorbifold}
\end{center}
\end{table}

A few remarks on the table are in order. The equality between the GSD of the TQD model and the number of primary fields of the corresponding orbifold is not surprising, as each primary field is associated with a highest--weight vector of an irreducible representation of the Virasoro algebra that is annihilated by the positive modes of the algebra and thus can be thought as a ``ground state". At this moment, such a relation may appear to be abstract; however, if a TQD model has a boundary, it may be possible to  construct a boundary CFT whose number of primary fields matches the GSD of the TQD model. An example is shown numerically for a $(2+1)$ dimensional Haldane model and its boundary CFT\cite{Cincio2012}. Since each 1--loop character counts all the descendants of a primary field, including the primary field itself, it naturally corresponds to a unique ground state of the TQD model. Like the TQD ground states $\ket{A,\mu}$, The 1--loop characters $\varkappa^A_{\mu}$ form an orthonormal basis, in which the $\ttr$ operator is diagonal. The conformal blocks $\cb{h}{g}$ projects onto this basis as
\[
\cb{h}{g}=\sum_{\mu}\overline{\widetilde{\chi}^{g^A}_{\mu}(h)} \varkappa^A_{\mu},
\]
which is precisely how a TQD state $\ket{g,h}$ projects onto the ground state basis, as in Eq. (\ref{eq:ghToAmu}).

Although we have been talking about the fractional topological numbers and statistics of the quasiparticles of our model, we do not have in hand the operators that can create or annihilate these quasiparticles, nor do we know the exact wave functions of these quasiparticles. Nevertheless, as a ramification of the correspondence with the orbifold CFTs, that we can study the topological numbers and statistics of the possible quasiparticles of our model by using only the modular matrices can be expected. This ramification is further propped by a similar correspondence between the $(2+1)$--dimensional Hamiltonian formulation of fractional quantum Hall effect (FQHE) systems and two--dimensional RCFT, which maps the holomorphic wave functions of the quasiparticles of the FQHE system to the conformal blocks of the CFT.
\section{Relations to Levin--Wen models}\label{sec:LW}
In this section, we discuss the relation between Levin--Wen models and our TQD models. In particular, we demonstrate a duality map of a class of Levin--Wen models into certain TQD models.

To begin with, let us briefly review Levin--Wen Models. Levin--Wen models, also known as string--net models, were proposed to generate the ground states that exhibit the phenomenon of string--net condensation as a physical mechanism for the time reversal invariant topological phases.
They are believed to be a Hamiltonian formulation of the Turaev--Viro topological field theories\cite{Hu2012,Wang2010,Turaev1994,KADAR2009}, analogous to that our TQD models are a Hamiltonian extension of topological Chern--Simons theories, as belabored in the previous section. Levin--Wen models are usually defined on the honeycomb lattice.

String degrees of freedom reside on the edges of the honeycomb lattice, each link of which is graced with one of $N+1$ string types. In the most general setting, the $N+1$ string types form a finite set $I$ equipped with a duality map $*:I \rightarrow I$ such that $j^{**}=j$ for all $j\in I$. These abstract string types are usually considered to label the irreducible representations of certain group or algebra (e.g., a quantum group).

A Levin-Wen model is specified by a triple of the input data $\{d,\delta,F\}$. Quantum numbers $d_j$ are called quantum dimensions and are complex numbers associated with the group elements $j\in G$, satisfying $d_j=d_{j^*}$. There are in principle two ways of setting up the fusion rules $\delta$. First, one can let the tensor product rules of the irreducible representations labeled by the string types as the fusion rules. But we do not consider this case here. Second, which is the case to be discussed in this section, one can use the product rule of certain group $G$ as the fusion rule, which is in fact a Kronecker $\delta$--function associated with each triple of string types $\{i,j,k\}$ respectively on the three links meeting at a vertex, such that $\delta_{ijk}$ equals $1$ if the group multiplication $ijk$ is the identity element $e$, and $0$ otherwise. The quantum dimensions and fusion rules must satisfy
\begin{align}
  \label{eq:ddelta}
  d_i d_j= \delta_{ijk^*}d_k.
\end{align}
The dual string type $j^*=j^{-1}$ can also represent the corresponding inverse group element in $G$. Fig. \ref{fig:snVertex} illustrates the fusion rule on the honeycomb lattice, by showing just one vertex. A link $a$ of the lattice is graced with a string type $s_a$ and endowed with an orientation, specified by an arrow. Such a string of type $s_a$  can also be represented by a flipped arrow, but with the conjugate string type
$s_a^*=(s_a)^{-1}$.
\begin{figure}[h!]
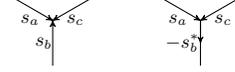

\centering
   \Ygraph{s_a}{s_b}{s_c}\qquad \Ygraph[0]{s_a}{s^*_b}{s_c}
\caption{A string--net vertex.}\label{fig:SNvertices}
\label{fig:snVertex}
\end{figure}

The $6j$ symbol $F$ are complex numbers that obey the following self--consistency conditions
\be
\begin{cases}
\label{eq:6jcond}
F^{ijm}_{kln}=F^{mij}_{nk^{*}l^{*}}\frac{v_m v_n}{v_j v_l}
=F^{klm^{*}}_{ijn^{*}}=F^{jim}_{lkn^*}=\overline{F^{j^*i^*m^*}_{l^*k^*n}},\vspace{0.5em}\\
F^{mlq}_{kp^{*}n}F^{jip}_{mns^{*}}F^{js^{*}n}_{lkr^{*}}
=F^{jip}_{q^{*}kr^{*}}F^{riq^{*}}_{mls^{*}},\vspace{0.5em}\\
F^{mlq}_{kp^{*}n}F^{l^{*}m^{*}i^{*}}_{pk^{*}n}
=\delta_{iq}\delta_{mlq}\delta_{k^{*}ip},
\end{cases}
\ee
where $v_j=\sqrt{d_j}$ (and $v_e=1$ for the identity group element $e$). The first line is a symmetry over the indices of $F$, where the last equality is meant for the Hamiltonian to be Hermitian, the second line the pentagon identity, and the last line the orthogonality condition. We remark here that it is the solution of the $F$--symbols and the quantum dimensions to Eq. \eqref{eq:6jcond} that dictates for which the abstract string types label the irreducible representations. Often, the string types turn out to label the irreducible representations of a rather complicated algebra $\mathfrak{A}$ although the group $G$ that supplies the fusion rules is very simple. See Section \ref{subsec:ExZ2Z3} for an example.

The usual string--net Hamiltonian takes the form
\be\label{eq:snH}
H=-\sum_v \hat{A}_v-\sum_p \hat{B}_p,
\ee
where $\hat{A}_v\Bigl\vert\bmm\scalebox{0.6}{\Ygraph{i}{j}{k}}\emm\Bigr\rangle =\delta_{ijk}\Bigl\vert\bmm\scalebox{0.6}{\Ygraph{i}{j}{k}}\emm\Bigr\rangle$ is the vertex operator defined at each string--net vertex, and $\hat{B}_p=\sum_{s=0}^N a_s \hat{B}^s_p$, with $a_s=d_s/D$ and $D=\sum_{i=0}^N d^2_i$, is the "magnetic--flux" operator defined for each hexagonal plaquette of the string--net lattice. Each operator $\hat{B}^s_p$ in $\hat{B}_p$ acts on a plaquette as follows\cite{Hung2012}.
\be\label{eq:BpsLW}
\hat{B}^s_p \BLvert\scalefont{0.6}\Psix[1]{l}{e}\Brangle
=\prod\limits_{a=1}^{6}F^{l_a e^*_a e_{a-1}}_{s^*{e'}_{a-1}{e'}_a^*}\BLvert\scalefont{0.6}\Psix[1]{l}{e'}\Brangle,
\ee
where ${e'}_a=e_a s$.

The vertex operators $\hat{A}_v$ are projectors. The parameter $a_s$ ensures that the operators $\hat{B}_p$ are also projectors. It can be shown that $\{\hat{A}_v,\hat{B}_p|\forall v,p\}$ is a set of commuting operators, whose common eigenstates span the Hilbert space of the model. The ground states of the model are thus the $+1$ eigenstates of $\hat{A}_v$ and $\hat{B}_p$, which are known as the string--net condensed states.

We emphasize that in Levin--Wen models, the plaquette operator $\hat{B}_p$ is identically zero outside of the subspace of $\hat{A}_v=1$ because the $F$--symbols are automatically zero. This is in contrast to our TQD models in which the vertex operators are well--defined and nontrivial outside $\Hil^{B_f=1}$.

The honeycomb lattice has as its dual lattice the triangular lattice, which is a regular case of the triangle graph $\Gamma$ on which our TQD models are defined.  We further notice that the way we enumerate the vertices and assign group elements on the edges of $\Gamma$ does induce orientations on the links of the dual honeycomb lattice, as seen in Eq. \eqref{eq:orientationOfLink}, in which only the dual honeycomb plaquettes are shown. Bear in mind that the enumerations of the vertices of $\Gamma$ now label the plaquettes.
\begin{align}
  \label{eq:orientationOfLink}
  \YYa{0.5}{v_1}{v_2}{v_3}{v_4}\quad \YYb\quad\YYc
\end{align}
The pair $[v_1 v_2]$ indicates that the plaquette $v_1$ is on the left of their common edge,
while $v_2$ is on the right. Hence, the group element on the edge can be denoted as $g_{[v_1v_2]}$ and obviously
satisfies $g_{[v_1 v_2]}=g_{[v_2 v_1]}^{-1}$.
This rule can be applied to the entire graph.

These observations imply that there may exist a kind of duality between the concerned type of Levin--Wen models and certain TQD models, in fact an inclusion of the former into the latter, as we now explore.

We claim that any $F$--symbol that solves Eq. \eqref{eq:6jcond} can be mapped to a $3$--cocycle $\alpha$ (up to a $3$--coboundary) that defines a TQD model. That said, a Levin--Wen model with such an $F$ can be identified with a TQD model with the corresponding $\alpha$. Indeed, a tensor $F$ has three independent indices and can be expressed as
\begin{align}
  \label{eq:Falpha}
  F^{ijm}_{kln}=\alpha(i,j,k)\delta_{m,(ij)^{-1}}\delta_{n,jk}\delta_{l,(ijk)^{-1}}
\end{align}
where $\alpha\in \mathds{C}^{\times}=\mathbb{C}\setminus\{0\}$ is a function of $i,j,k\in G$. The self--consistency conditions in Eq. $\eqref{eq:6jcond}$ become
\be
  \label{eq:6jcondalpha}
\begin{cases}
\begin{aligned}
\alpha(i,j,k) &=\alpha\left((ij)^{-1},i,jk\right)\frac{v_{ij}v_{jk}}{v_j v_{ijk}} \\
  &=\alpha\left(k,(ijk)^{-1},i\right)\\
  &={\alpha\left(j,i,{ijk}^{-1}\right)}=\overline{{\alpha\left(j^{-1},i^{-1},ijk\right)}},
\end{aligned}\vspace{0.5em}\\
\dfrac{\alpha(l,k,j)}{\alpha(ml,k,j)}\dfrac{\alpha(m,lk,j)}{\alpha(m,l,kj)} \alpha(m,l,k)=1,\vspace{0.5em}\\
\alpha(m,l,k)\alpha(l^{-1},m^{-1},mlk)=1,
\end{cases}
\ee
respectively. The pentagon identity in the second line is readily the $3$--cocycle condition of this $\alpha$ that turns out to be a $3$--cocycle in $H^3(G,\mathds{C}^{\times})$. It does not harm to assume that $\alpha$ is $U(1)$--valued in $H^3(G,U(1))\cong H^3(G,\mathds{C}^{\times})$. Then we see that the orthogonality condition is identified with one equality in the symmetry condition.

The quantum dimension $d_j$ then takes the definition
\begin{equation}
  \label{eq:djalpha}
  d_i=\alpha(i^{-1},i,i^{-1}),
\end{equation}
as is verified by setting $j=i^{-1}$ and $k=i$ in the symmetry condition and by the $3$--cocycle condition.

The conditions $d_i d_k=d_{ik}$ and $d_i=d_{i^{-1}}=\pm 1$ are immediate consequences of the symmetry condition in Eq. \eqref{eq:6jcondalpha}. By setting $j$ to the identity element $e$ in the first equality of the symmetry condition \eqref{eq:6jcondalpha} and by applying the $3$--cocycle condition $\delta\alpha\left((ij)^{-1},i,j,k\right)=1$, we obtain
\be
\frac{v_i v_k}{v_{ik}}=\frac{\alpha(e,e,k)}{\alpha(i^{-1},i,k)\alpha(i^{-1},i,e)}.
\ee
Since $\alpha(i^{-1},i,k)^2=1$ from the symmetry condition for all $i,k\in G$, taking a square of the above equation yields $d_i d_k=d_{ik}$. The condition $d_i=d_{i^{-1}}=\pm 1$ is due to $\alpha(i^{-1},i,i)^2=1$ and $d_i d_{i^{-1}}=d_e=1$.

Substituting the $\alpha$ in Eq. (\ref{eq:Falpha}) into the $\hat{B}^s_p$ in Eq. (\ref{eq:BpsLW}), we find  that $A_p^s|_{B_f=1}$ can be identified with $d_s\hat{B}^s_p$ for all $p$ up to a common unitary transformation that depends on the enumeration of the vertices and the choice of $\{v_i\}$, where $A_p^s|_{B_f=1}$ is the TQD model vertex operator acting on the vertex $p$ dual to the plaquette $p$ on the honeycomb lattice, with the action restricted to the subspace $\Hil^{B_f=1}$. Notice that the enumeration dependence of the transformation is related to the enumeration dependence of the definition of $A_p^s$.

Therefore, we infer that the Levin--Wen model with the fusion rule $\delta_{ijk}$ determined by the group multiplication law of $G$ and the $F$--symbols meeting Eq. (\ref{eq:6jcond}) can be identified as the TQD model restricted to  $\Hil^{B_f=1}$ that is defined by the (not necessarily normalized) $U(1)$--valued $3$--cocycle $\alpha$ that satisfies
\begin{align}
  \label{eq:AlphaForLW}
  \alpha(i,j,k)
  &=\alpha\left((ij)^{-1},i,jk\right)\frac{v_{ij}v_{jk}}{v_j v_{ijk}}
  =\alpha\left(k,(ijk)^{-1},i\right)
  \nonumber\\
  &={\alpha\left(j,i,{ijk}^{-1}\right)}={\alpha\left(j^{-1},i^{-1},ijk\right)}^{-1}
\end{align}
where $v_i=\sqrt{d_i}$ and $d_i=\alpha(i,i^{-1},i)$.

Nevertheless, not all $3$--cocycles satisfy the conditions \eqref{eq:AlphaForLW}, in which case the TQD models do not correspond to any Levin-Wen models. For example, with $\Z_2$ fusion rule, there are two Levin--Wen models, namely the toric code model and the double semion model, which can be identified with the TQD models respectively with the $[\alpha^0]$ and $[\alpha^1]$ given in Eq. \eqref{eq:ZmCocycle}, which correspond to $d_1=1$ and $d_1=-1$ respectively, according to Eq. \eqref{eq:djalpha}. There is, however, only one Levin-Wen model with $\Z_3$ fusion rule that has a dual TQD model---the one with the $[\alpha^0]$ given in Eq. (\ref{eq:PowerCocycle}). Note that here, we do not distinguish the TQD models defined by equivalent $3$--cocycles because they describe the same topological phases, as explained in Section \ref{subsec:equivModel}.
Details of the TQD models with $\Z_2$ and $\Z_3$ are found in Section \ref{subsec:ExZ2Z3}.

The study in this section partially answers the question when and how Levin--Wen models can be characterized by group cohomology, which was raised in Ref\cite{Hung2012}.
\section{Discussions and Outlook}\label{sec:disc}
In this very last section, we shall summarize our major results along with discussions on a few questions bonded to these results that are yet not fully answered in this paper but deserve future exploration.

First of all, we fabricated a new model---the Twisted Quantum Double model---of 2d topological phases by a $3$--cocycle $[\alpha]\in H^3(G,U(1))$ of a finite group $G$ on a graph composed of triangles, each edge of which is decorated by an element of $G$. This model constitutes a very rich class of topological phases, which are otherwise missing in some other models, such as the Kitaev model and Levin--Wen model. The topological properties of the TQD model are reflected in the topological numbers---GSD, topological spin, etc---associated with the topological observables of the model.

We further classified these topological numbers, which either directly depend on the defining $3$--cocycle of the model or indirectly via a twisted $2$--cocycle determined by the $3$--cocycle. Two TQD models defined by two equivalent $3$--cocycles are shown to bear the same topological phase. We thus expect that the classification of the topological numbers does the job as well for the topological phases described by the TQD models. We expect but do not affirm this yet because we have not been able to explicitly prove that two inequivalent $3$--cocycles never yield the same topological phase. This and detailed studies of the topological phases certainly calls for more efforts in future works.

Second, our TQD model appears to be certain generalization of the Kitaev model in the following sense. A TQD model is precisely a Kitaev model when its defining $3$--cocycle is trivial. In this situation, the Hamiltonian consists of local gauge transformations and local flux projections. As a collective effect, the ground states are classified by the irreducible representations of quantum double of the finite group $G$, and this is expected to be true also for the quasiparticle excitations. When the defining $3$--cocycle is non--trivial, the Hamiltonian can be viewed as consisting of local twisted gauge transformations and local flux projections. Similarly the ground states are classified by the irreducible representations of the twisted quantum double of the finite group $G$.

Third, we relate our TQD models to the Dijkgraaf--Witten (DW) topological Chern--Simons theories, by viewing ours as a Hamiltonian extension of the latter. In fact, we have shown that the GSD of a TQD model defined by some $3$--cocycle on the boundary of a $3$--manifold coincides with the partition function of the DW topological Chern-Simons theory in the bulk, whose topological action is given by the same $3$--cocycle.

This connection motivates a correspondence between our TQD models on a torus and the RCFTs that are the toric orbifolds by a finite group of a holomorphic CFT and are twisted by nontrivial $3$--cocycles. This correspondence identifies the ground states, the GSD, and the modular matrices of a TQD model, respectively, with the holomorphic characters, the  number of primary fields, and also the modular matrices of the corresponding RCFT. Provided with the description of fractional quantum Hall effect by CFT, we are encouraged to expect that the statistical and topological properties of the quasiparticle excitations and hence the topological phase of a TQD model can be investigated in terms of the modular matrices of the model.

Fourth, to echo the fact that our model is partly motivated by the Levin--Wen model, we studied the relation between TQD models and the type of Levin--Wen models where the fusion rules coincide with the multiplication laws of finite groups: we demonstrated that each such Levin--Wen model on a graph can be directly translated to a TQD model of the type on the dual graph. The reverse is not true, however, indicating that TQD models embodies more topological phases than this type of Levin--Wen models. In our study of this in Section \ref{sec:LW}, we adopted simply the original settings of the Levin--Wen model\cite{Levin2004}.

Furthermore, as pointed out in Section \ref{sec:LW} and in Ref\cite{Hung2012}, the fusion rules in a Levin--Wen model can in principle be identified with the tensor product rules of the irreducible representations of certain group or algebra. If this is the case, the pentagon identity in Eq. \eqref{eq:6jcond} of the $F$--symbols contains a summation over the index $n$ on the LHS due to the summation that would appear in the fusion rules, which comprises the interpretation of the $F$--symbols as $3$--cocycles and the pentagon identity as the corresponding $3$--cocycle condition. Then clearly, this type of Levin--Wen models, apart from the special cases where the representations are restricted so as to remove the summation, cannot be dual to our TQD models. This type of Levin--Wen models are believed to be classified in terms of tensor categories. This distinction between the two types of Levin--Wen models is related to the question when and how Levin--Wen model can be classified by group cohomology, which is raised in Ref\cite{Hung2012} that is inspired by the duality between certain SPT phases and long range entangled topological phases described by certain Levin--Wen models\cite{Chen2011e,Levin2012,Hung2012}. The duality found between our TQD models and Levin--Wen models then partially answered this question.

Above all, a main purpose of this paper is to reveal the topological properties of the ground states. We propose the following topological properties of the elementary excitations. The number of quasiparticle species in the elementary excitations is equal to the GSD on a torus. Moreover, the topological charge that identifies the quasiparticles are classified by the twisted quantum double of the finite group $G$, and The $\str$ and $\ttr$ statistical matrices are the same as the modular $\str$ and $\ttr$ matrices derived from the topological observable in the ground states on a torus. Work is in progress in this direction.
In general, two inequivalent 3 cocycles may yield the same topological phase, because of possible relabeling of the quasiparticles. For example, 
$H^3(\Z_2\times \Z_2, U(1))$ has eight equivalence classes of $3$-cocycles. If we assume that the set $\{GSD, \str, \ttr\}$ gives the number of topological phases, then it is verified that there are only 4 independent sets of {GSD, S, T} in the case of $\Z_2 \times \Z_2$. At this moment, we are lack of a systematic understanding of the underlying principle and general pattern, and shall try to address this issue in our future work.

\section*{Acknowledgements}

YH and YW thank Department of Physics, Fudan University, where 
this work was initiated, for warm hospitality they received 
during a visit in summer 2012. YT thanks Dr. Z.C. Gu for his helpful discussions. YW appreciates Prof. Mikio 
Nakahara, Prof. Seigo Tarucha, and Prof. Rod Van Meter for their 
generosity and tolerance. YW's gratitudes also go to Dr. Ling--Yan Hung, 
Dr. Peng Gao, Prof. Guifre Vidal for their helpful discussions, and in particular Juven Wang for his proof reading of many sections of the manuscript. 
YS Wu thanks Dr. Z.C. Gu for discussions. He is supported in part 
by US NSF through grant No. PHY-1068558.

\begin{appendix}
\section{Basics of $H^n(G,U(1))$}\label{app:HnGU1}
The $3$-cocycles concerned in this paper correspond to the topological actions in the Dijkgraaf-Witten Chern-Simons theory realized on simplicial triangulations of $3$-manifolds. This physics is introduced in Section\ref{sec:relation2DW}. Here in this appendix, to be self-contained, we briefly catalog basic definitions of cohomology groups $H^n(G,U(1))$ of finite groups $G$.

The $n$-th \textit{cochain group} $C^n(G,U(1))$ of a finite group $G$ is an Abelian group of $n$-\textit{cochains} $c(g_1,\dots,g_n): G^{\times n}\to U(1)$, where $g_i\in G$, with the group multiplication: $c(g_1,\dots,g_n)c'(g_1,\dots,g_n)=(cc')(g_1,\dots,g_n)$. There is a natural derivation from $C^n$ to $C^{n+1}$, namely the \textit{coboundary operator} $\delta$ defined as follows.
\begin{align*}
\delta &:C^n\to C^{n+1}\\
&: c(g_1,\dots,g_n)\mapsto \delta c(g_0,g_1\dots,g_n),
\end{align*}
where
\begin{align*}
&\delta c(g_0,g_1\dots,g_n)\\
=& \prod_{i=0}^{n+1}c(\dots,g_{i-2},g_{i-1}g_i,g_{i+1},\dots)^{(-1)^i},
\end{align*}
where it is understood as when $i=0$, the arguments start at $g_0$, and when $i=n+1$, the arguments end at $g_{n-1}$. Equation (\ref{3CocycleCondition}) is the example for $n=3$. It is easy to verify that $\delta^2c=1$, the nilpotency of $\delta$, by which the following exact sequence is established:
\[
\cdots C^{n-1}\stackrel{\delta}{ \to} C^n\stackrel{\delta}{ \to} C^{n+1}\cdots ,
\] 
where the $n$-cochains in $\mathrm{im}(\delta:C^{n-1}\to C^n)$ are called $n$-\textit{coboundaries}, and those in $\ker(\delta:C^n\to C^{n+1})$ are called $n$-\textit{\textbf{cocycles}}, i.e. those satisfying the \textit{cocycle condition} $\delta c=1$. Again, Eq. (\ref{3CocycleCondition}) is the example for $n=3$. This exact sequence gives rise to the definition of the cohomology group
\[
H^n(G,U(1)):=\frac{\ker(\delta:C^n\to C^{n+1})}{\mathrm{im}(\delta:C^{n-1}\to C^n)},
\]
which is the Abelian group of equivalence classes of $n$-cocyles that defer from each other by merely an $n$-coboundary.
Trivial $n$-cocycles are those in the equivalence class with the unit 1.

\section{Algebra of local operators}\label{app:algAvBf}

In this appendix, we show that the local operators in the Hamiltonian Eq. (\ref{eq:Hamiltonian}) forms the following algebra.
\begin{align}
  \label{AlgebraAB}
  &\mathrm{(i)}
  &[B_{f'},B_f]=0,
  \quad [B_f,A^g_v]=0,
  \nonumber\\
  &\mathrm{(ii)}
  &[A^g_v,A^h_w]=0\text{ if } v\neq w,
  \nonumber\\
  &\mathrm{(iii)}
  &A^g_{v'}A^h_v = A^{g\cdot h}_v,
\end{align}
where in (iii), $[v'v]=g$ is understood. The equality (iii) in the above implies that $A_v=\sum_gA^g_v/|G|$ is a projector.

\smallskip

\noindent(i).
That $[B_{f'},B_f]=0$ follows immediately from the definition of $B_f$ in Eq. \eqref{eq:actionOfBf}.
Since $A_v$ affects only the group elements on the boundary of the vertex $v$, $[B_f,A_v^g]=0$ holds obviously when the vertex $v$ is not on the boundary of $f$. In the case where $v$ is right on the boundary of $v$, without loss of generality, let us consider two actions, $A_{v_2}^g B_f$ and $B_f A_{v_2}^g$, on the basis vector
\[
  \BLvert \oneTriangle{v_1}{v_2}{v_3} \Brangle.
\]
We have
\begin{align}
&A_{v_2}^g B_f\BLvert \oneTriangle{v_1}{v_2}{v_3} \Brangle \\ =&\dots \frac{\dlt{v_1}{v'_2}{v_3}}{\elf{v_1}{v'_2}{v_2}{v_3}}\BLvert \oneTriangle{v_1}{v'_2}{v_3} \Brangle\nonumber\\
=&\dots \frac{\dlt{v_1}{v_2}{v_3}}{\elf{v_1}{v'_2}{v_2}{v_3}}\BLvert \oneTriangle{v_1}{v'_2}{v_3} \Brangle\nonumber\\
=&\dots \frac{\dlt{v_1}{v_2}{v_3}}{\elf{v_1}{v'_2}{v_2}{v_3}}\BLvert \oneTriangle{v_1}{v'_2}{v_3} \Brangle\nonumber\\
=&B_f A_{v_2}^g\BLvert \oneTriangle{v_1}{v_2}{v_3} \Brangle,
\end{align}
where $[v'_2v_2]=g$ is understood, the dots $\dots$ collects all other factors irrelevant and thus omitted, and the second equality follows from applying the chain rule $[v_1v'_2]\cdot[v'_2v_3]=[v_1v'_2]\cdot[v'_2v_2] \cdot[v_2v'_2]\cdot[v'_2v_3]=[v_1v_2]\cdot[v_2v_3]$. Hence,
we conclude that $[B_f,A_v^g]=0$ holds for all $f,v\in\Gamma$.



\smallskip

\noindent(ii).
It is clear by the definition of $B_v$ that if $v_1$ and
$v_2$ are not connected by any edge, $[A_{v_1}^g ,A_{v_2}^h]=0$ is true. We then need only to check the case where $v_1$ and $v_2$ are neighboring
to each other. Let us first check the following action of $A_{v_1}^g A_{v_2}^h$ on a relevant basis vector.
\begin{align}
  \label{B1B2}
  &A_{v_1}^g A_{v_2}^h \BLvert \twoTriangles{v_1}{v_3}{v_2}{v_4}{1} \Brangle
  \nonumber\\
  =&\Bigl(
  \elf{v_1}{v'_2}{v_2}{v_4}^{-1}
  \nonumber\\
  &
  \elf{v_1}{v'_2}{v_2}{v_3}
  ...\Bigr) A_{v_1}^g \BLvert \twoTriangles{v_1}{v_3}{v'_2}{v_4}{1} \Brangle
  \nonumber\\
  =&\Bigl(
  \elf{v_1}{v'_2}{v_2}{v_4}^{-1}
  \elf{v_1}{v'_2}{v_2}{v_3}...\Bigr)
  \nonumber\\
  &\Bigl(
  \elf{v'_1}{v_1}{v'_2}{v_4}
  \nonumber\\
  &\elf{v'_1}{v_1}{v'_2}{v_3}^{-1}...\Bigr)
  \BLvert \twoTriangles{v'_1}{v_3}{v'_2}{v_4}{1} \Brangle
\end{align}
with $[v'_2v_2]=h,[v'_1v_1]=g$. Only those 3-cocycles corresponding
to the two common boundary vertices are written down.

By using twice the 3-cocycle condition in \eqref{3CocycleCondition}, we have
\begin{align}
  &\elf{v'_1}{v'_2}{v_2}{v_4}
  \alpha\left([v'_1v_1]\cdot[v_1v'_2],[v'_2v_2],[v_2v_4]\right)^{-1}
  \nonumber\\
  &\alpha\left([v'_1v_1],[v_1v'_2]\cdot[v'_2v_2],[v_2v_4]\right)
  \nonumber\\
  &\alpha\left([v'_1v_1],[v_1v'_2],[v'_2v_2]\cdot[v_2v_4]\right)^{-1}\times
  \nonumber\\
  &\alpha\left([v'_1v_1],[v_1v'_2],[v'_2v_2]\right)=1
\end{align}
and
\begin{align}
  &\alpha\left([v'_1v'_2],[v'_2v_2],[v_2v_3]\right)
  \nonumber\\
  &\alpha\left([v'_1v_1]\cdot[v_1v'_2],[v'_2v_2],[v_2v_3]\right)^{-1}
  \nonumber\\
  &\alpha\left([v'_1v_1],[v_1v'_2]\cdot[v'_2v_2],[v_2v_3]\right)
  \nonumber\\
  &\alpha\left([v'_1v_1],[v_1v'_2],[v'_2v_2]\cdot[v_2v_3]\right)^{-1}
  \nonumber\\
  &\alpha\left([v'_1v_1],[v_1v'_2],[v'_2v_2]\right)=1
  \nonumber\\
\end{align}
together with the chain rule Eq. (\ref{eq:ChainRuleInBph}), we find the action of $A_{v_1}^g A_{v_2}^h$ is the same as
$A_{v_2}^h A_{v_1}^g$:
\begin{align}
  \label{B2B1}
    &A_{v_2}^h A_{v_1}^g \BLvert \twoTriangles{v_1}{v_3}{v_2}{v_4}{1} \Brangle
    \nonumber\\
    =&\Bigl(\alpha\left([v'_1v_1],[v_1v_2],[v_2v_4]\right)
    \nonumber\\
    &\alpha\left([v'_1v_1],[v_1v_2],[v_2v_3]\right)^{-1}...\Bigr)
    B_{v_2}^h \BLvert \twoTriangles{v'_1}{v_3}{v_2}{v_4}{1} \Brangle
    \nonumber\\
    =&\Bigl(\alpha\left([v'_1v_1],[v_1v_2],[v_2v_4]\right)
    \alpha\left([v'_1v_1],[v_1v_2],[v_2v_3]\right)^{-1}...\Bigr)
    \nonumber\\
    &\Bigl(\alpha\left([v'_1v'_2],[v'_2v_2],[v_2v_4]\right)^{-1}
    \nonumber\\
    &\alpha\left([v'_1v'_2],[v'_2v_2],[v_2v_3]\right)...\Bigr)
    \BLvert \twoTriangles{v'_1}{v_3}{v'_2}{v_4}{1} \Brangle
\end{align}

Notice that the chain rule in Eq. (\ref{eq:ChainRuleInBph})
guarantees that each group element indexed by the same pair
of enumerations is the same in the above evaluations.
Therefore we arrive at $A_{v_1}^g A_{v_2}^h =A_{v_2}^h A_{v_1}^g$.

\smallskip

\noindent(iii). We show $A^g_{v'} A^h_v = A^{g\cdot h}_v$, where $v$ is assumed to become $v'$ after the action of $A^h_v$ with $[v'v]=h$, while $v$ is turned to be $v''$ with $[v''v]=g\cdot h$. We begin with the action of $A^g_{v'} A^h_v $ on the vertex $v_2$of  the basis vector as follows.
\begin{align}
  \label{BgBh}
  &A^g_{v'_2} A^h_{v_2}\BLvert \threeTriangles{v_1}{v_3}{v_2}{v_4}{3}\Brangle
  \nonumber\\
  =&\elf{v_1}{v'_2}{v_2}{v_3}
  \nonumber\\
  &\elf{v_1}{v'_2}{v_2}{v_4}^{-1}
  \nonumber\\
  &\elf{v'_2}{v_2}{v_3}{v_4}^{-1}
  A^g_{v'}\BLvert \threeTriangles{v_1}{v_3}{v'_2}{v_4}{3} \Brangle
  \nonumber\\
  =&\Bigl(\elf{v_1}{v'_2}{v_2}{v_3}
  \elf{v_1}{v'_2}{v_2}{v_4}^{-1}
  \nonumber\\
  &\elf{v_2}{v'_2}{v_3}{v_4}^{-1}  \Bigr)
  \nonumber\\
  &\Bigl(\elf{v_1}{v''_2}{v'_2}{v_3}
  \elf{v_1}{v''_2}{v'_2}{v_4}^{-1}
  \nonumber\\
  &\elf{v''_2}{v'_2}{v_3}{v_4}^{-1}
  \Bigr)\BLvert \threeTriangles{v_1}{v_3}{v''_2}{v_4}{3} \Brangle.
\end{align}

Using three times the 3-cocycle condition \eqref{3CocycleCondition},
\begin{align}
  &\elf{v_1}{v'_2}{v_2}{v_3}
  \elf{v_1}{v''_2}{v'_2}{v_3}
  \nonumber\\
  =&\elf{v_1}{v''_2}{v'_2}{v_2}
  \elf{v_1}{v''_2}{v_2}{v_3}
  \nonumber\\
  &\elf{v''_2}{v'_2}{v_2}{v_3}
\end{align}

\begin{align}
  &\elf{v_1}{v''_2}{v'_2}{v_4}^{-1}
  \elf{v''_1}{v'_2}{v_2}{v_4}^{-1}
  \nonumber\\
  =&\elf{v_1}{v''_2}{v'_2}{v_2}^{-1}
  \elf{v_1}{v''_2}{v_2}{v_4}^{-1}
  \nonumber\\
  &\elf{v''_2}{v'_2}{v_2}{v_4}^{-1}
\end{align}

\begin{align}
  &\elf{v'_2}{v_2}{v_3}{v_4}^{-1}
  \elf{v''_2}{v'_2}{v_3}{v_4}^{-1}
  \nonumber\\
  =&\elf{v''_2}{v'_2}{v_2}{v_3}^{-1}
  \elf{v''_2}{v'_2}{v_2}{v_4}
  \nonumber\\
  &\elf{v''_2}{v_2}{v_3}{v_4}
\end{align}
we obtain
\begin{align}
  &\elf{v_1}{v''_2}{v_2}{v_3}
  \elf{v_1}{v''_2}{v_2}{v_4}^{-1}
  \nonumber\\
  &\elf{v''_2}{v_2}{v_3}{v_4}
  \BLvert \threeTriangles{v_1}{v_3}{v''_2}{v_4}{3} \Brangle
\end{align}

According to the chain rule in Eq. (\ref{eq:ChainRuleInBph}),
we have $[v''_2v_2]=[v''_2v'_2]\cdot[v'_2v_2]$.
The above action is identified as action of $A_{v_2}^{g\cdot h}$.
Therefore we conclude that
\[
  A^g_{v'} A^h_v = A^{g\cdot h}_v
\]
Though the above proof is done on a triangle plaquette,
the general proof on a plaquette of any other shape is straightforward.

\section{Mutations are Unitary Symmetry transformations}

\subsection{Symmetry}

When restricted to ground states
in $\Hil^0_{\Gamma}\subset\Hil^{B_f=1}_{\Gamma}$, we can impose the following \textit{chain rule} on all triangles in $\Gamma$:
 \begin{equation}
  \label{ChainRuleInGroundState}
   [v_i v_j]=[v_iv_j]\cdot[v_jv_k],
 \end{equation}
 where $v_i$, $v_j$, and $v_k$ are the three vertices of any triangle. Since the mutation operators are defined on the subspace $\Hil^{B_f=1}_{\Gamma}$ for all $\Gamma$, to show that $T_iP_{\Gamma}=P_{\Gamma'}T_i$, where $\Gamma'=T_i(\Gamma)$, we can neglect the face operators in the projectors $P_{\Gamma}$ and $P_{\Gamma'}$ defined in Eq. (\ref{eq:GSDprojector}). We thus need to show that
\be\label{eq:TBcommute}
 T_i\prod_{v\in\Gamma}A_v=\prod_{v\in\Gamma'}A_v T_i
\ee
holds for all mutation operators $T_i$, $i=1,2,3$, and any state in $\Hil^{B_f=1}_{\Gamma}$ on any graph $\Gamma$. Since a $T_i$ acts on at most four three triangles and does not affect any other triangle in the same graph, $T_i$ certainly commutes with any $A_v$ at any vertex $v$ that does not lie on the boundary of the triangles on which the $T_i$ acts. Hence, in Eq. (\ref{eq:TBcommute}) we can neglect the part of $\Gamma$ out of the scope of the action of $T_i$ and hence the $A_v$ acting on this part. We now pursue the proof in the following equations respectively for $T_1$, $T_2$, and $T_3$.
In the sequel, $[ijk]$ denotes a triangle whose vertices are $i$, $j$, and $k$ counter--clockwise, and \textquotedblleft$\cdots$" represents all the irrelevant factors, which are thus omitted.

\begin{align}
  & T_1([1'2'4'],[2'3'4'])\prod_{i=1}^4A_i\BLvert\twoTriangles{1}{2}{3}{4}{0}\Brangle\nonumber\\
= &\frac{1}{|G|^4}  \sum_{\substack{[1'1],[2'2],\\ [3'3],[4'4]\in G}}  \cdots\elf{1'}{1}{2}{4} \nonumber\\
&\times \elf{1'}{2'}{2}{4}^{-1}  \elf{2'}{2}{3}{4}\nonumber\\
&\times\elf{2'}{3'}{3}{4}^{-1}\elf{1'}{2'}{4'}{4}\nonumber\\ &\times\elf{2'}{3'}{4'}{4}T_1   \BLvert\twoTriangles{1'}{2'}{3'}{4'}{0}\Brangle\nonumber\\
= &\frac{1}{|G|^4}\sum_{\substack{[1'3'],[1'1],[2'2],\\ [3'3],[4'4]\in G}}  \cdots\elf{1'}{1}{2}{4} \nonumber\\
&\times \elf{1'}{2'}{2}{4}^{-1}  \elf{2'}{2}{3}{4}\nonumber\\ &\times \elf{2'}{3'}{3}{4}^{-1}\elf{1'}{2'}{4'}{4}\nonumber\\ &\times\elf{2'}{3'}{4'}{4}\elf{1'}{2'}{3'}{4'}\nonumber\\
&\times\BLvert\twoTriangles{1'}{2'}{3'}{4'}{1}\Brangle\label{eq:T1Bv}
\end{align}
\begin{align}
  & \prod_{i=1}^4A_iT_1([124],[234])\BLvert\twoTriangles{1}{2}{3}{4}{0}\Brangle\nonumber\\
= &  \sum_{[13]\in G} \elf{1}{2}{3}{4}\prod_{i=1}^4A_i     \BLvert\twoTriangles{1}{2}{3}{4}{0}\Brangle\nonumber\\
= &\frac{1}{|G|^4}\sum_{\substack{[13],[1'1],[2'2],\\ [3'3],[4'4]\in G}}  \cdots\elf{1'}{1}{2}{3}\nonumber\\
&\times \elf{1'}{1}{3}{4}  \elf{1'}{2'}{2}{3}^{-1}\nonumber\\
&\times \elf{1'}{2'}{3'}{3}\elf{1'}{3'}{3}{4}\nonumber\\
&\times \elf{1'}{3'}{4'}{4}\elf{1}{2}{3}{4}\nonumber\\
&\times\BLvert\twoTriangles{1'}{2'}{3'}{4'}{1}\Brangle\label{eq:BvT1}
\end{align}

It is straightforward to show that the RHS of Eq. (\ref{eq:T1Bv}) is equal to that of Eq. (\ref{eq:BvT1}) by knowing that $\sum_{[1'1],[13],[3'3]}=\sum_{[1'1],[1'3'],[3'3]}$ because of the chain rules $[1'1]\cdot[13]=[1'3]$ and $[1'3']\cdot[3'3]=[1'3]$, and by applying the following four $3$--cocycle conditions in order.
\begin{widetext}
\begin{align*}
& \elf{1'}{2'}{4'}{4}\elf{2'}{3'}{4'}{4}=\elf{1'}{3'}{4'}{4}   \elf{1'}{2'}{3'}{4}\elf{1'}{2'}{3'}{4'}^{-1},\\
& \elf{1'}{1}{2}{4}=\elf{1}{2}{3}{4}\elf{1'}{2}{3}{4}^{-1}     \elf{1'}{1}{3}{4}\elf{1'}{1}{2}{3},\\
& \elf{2'}{2}{3}{4}\elf{1'}{2}{3}{4}^{-1}=\elf{1'}{2'}{3}{4}^{-1}   \elf{1'}{2'}{2}{4}\elf{1'}{2'}{2}{3},\\
& \elf{2'}{3'}{3}{4}^{-1}\elf{1'}{2'}{3'}{4}\elf{1'}{2'}{3}{4}^{-1}   =\elf{1'}{3'}{3}{4}^{-1}\elf{1'}{2'}{3'}{3}.
\end{align*}
\end{widetext}
Thus, Eq. (\ref{eq:TBcommute}) holds for $T_1$.

The case of $T_2$ is a bit trickier. Let us write down how the LHS and RHS of Eq. (\ref{eq:TBcommute}) act on a state as follows.
\begin{widetext}
\begin{align}
&T_2([2'3'4'])\prod_{i=1}^3 A_i\BLvert\oneTriangle{2}{3}{4}\Brangle\nonumber\\
=& \frac{1}{|G|^3}\sum_{\substack{[2'2],[3'3],\\ [4'4]\in G}}\elf{2'}{2}{3}{4} \elf{2'}{3'}{3}{4}^{-1}\elf{2'}{3'}{4'}{4} T_2\BLvert\oneTriangle{2'}{3'}{4'}\Brangle\nonumber\\
=& \frac{1}{|G|^3}\sum_{\substack{[2'2],[3'3],[4'4],\\ [12'],[13'],[14'] \in G}}\elf{2'}{2}{3}{4} \elf{2'}{3'}{3}{4}^{-1}\nonumber\\
&\qquad\qquad\qquad\quad\times  \elf{2'}{3'}{4'}{4}\elf{1}{2'}{3'}{4'} \BLvert\threeTriangles{2'}{3'}{1}{4'}{4}\Brangle\label{eq:T2Bv}
\end{align}
\begin{align}
&\prod_{i=1}^3A_i T_2([234])\BLvert\oneTriangle{2}{3}{4}\Brangle= \sum_{\substack{[12],[13],\\ [14]\in G}}\elf{1}{2}{3}{4}^{-1} \prod_{i=1}^4A_i \BLvert\threeTriangles{2}{3}{1}{4}{4}\Brangle\nonumber\\
=& \frac{1}{|G|^4}\sum_{\substack{[12],[13],[14],[1'1],\\ [2'2],[3'3],[4'4]\in G}}\elf{1}{2}{3}{4}\elf{1'}{1}{3}{4}\elf{1'}{1}{2}{4}^{-1}\elf{1'}{1}{2}{3}\nonumber\\
&\qquad\qquad\qquad\quad \times\elf{1'}{2'}{2}{3}^{-1} \elf{1'}{2'}{2}{4}\elf{1'}{2'}{3'}{3}\nonumber\\
&\qquad\qquad\qquad\quad \times\elf{1'}{3'}{3}{4}^{-1}\elf{1'}{2'}{4'}{4}^{-1} \elf{1'}{3'}{4}{4}\BLvert\threeTriangles{2'}{3'}{1'}{4'}{4}\Brangle\nonumber\\
=& \frac{1}{|G|^3}\sum_{\substack{[2'2],[3'3],[4'4],\\ [1'2'],[1'3'],[1'4'] \in G}}\elf{2'}{2}{3}{4} \elf{2'}{3'}{3}{4}^{-1}\nonumber\\
&\qquad\qquad\qquad\quad\times  \elf{2'}{3'}{4'}{4}\elf{1'}{2'}{3'}{4'} \BLvert\threeTriangles{2'}{3'}{1'}{4'}{4}\Brangle,\label{eq:BvT2}
\end{align}
\end{widetext}
where the last equality is obtained by first plugging into the second row the following four $3$--cocycle conditions
\begin{align*}
&\tfrac{\elf{1'}{1}{3}{4}\elf{1'}{1}{2}{3}}{\elf{1'}{1}{2}{4}}=\tfrac{ \elf{1'}{2}{3}{4}}{\elf{1}{2}{3}{4}}\\
&\tfrac{\elf{1'}{2'}{2}{4}}{\elf{1'}{2'}{2}{3}}=\tfrac{\elf{2'}{2}{3}{4} \elf{1'}{2'}{3}{4}}{\elf{1'}{2}{3}{4}}\\
&\tfrac{\elf{1'}{2'}{3'}{3}}{\elf{1'}{3'}{3}{4}}=\tfrac{\elf{1'}{2'}{3'}{4}}{\elf{1'}{2'}{3}{4} \elf{2'}{3'}{3}{4}}\\
&\tfrac{\elf{1'}{3'}{4'}{4}}{\elf{1'}{2'}{4'}{4}}=\tfrac{\elf{2'}{3'}{4'}{4} \elf{1'}{2'}{3'}{4'}}{\elf{1'}{2'}{3'}{4}},
\end{align*}
and then by applying to the summations these chain rules:
\begin{align*}
& [1'1]\cdot[12]\cdot[2'2]^{-1}=[1'2'],\\
& [1'1]\cdot[13]\cdot[3'3]^{-1}=[1'3'],\\
& [1'1]\cdot[14]\cdot[4'4]^{-1}=[1'4'],
\end{align*}
which are guaranteed by the restriction of $B_f=1$ and the properties of the vertex operators $A_i$. The summations over $[2'2]$, $[3'3]$, and $[4'4]$ can be replaced by those over $[1'2']$, $[1'3']$, and $[1'4']$ respectively. Since $[1'1]$ does not appear in the $3$--cocycles any more, $\sum_{[1'1]}$ contributes a factor a $|G|$. Clearly, the vertex enumerated by $1'$ on the RHS of the last equality of Eq. (\ref{eq:BvT2}) is now dummy and thus can be re--enumerated by $1$ without altering its order relative the enumerations of the other three vertices. As such, we can see that the RHS of Eq. (\ref{eq:BvT2}) and that of Eq. (\ref{eq:T2Bv}) are actually identical. That is, Eq. (\ref{eq:TBcommute}) is true for $T_2$.

When it comes to the case of $T_3$, we have the action of the LHS and that of the RHS of Eq. (\ref{eq:TBcommute}) on the same state respectively being
\begin{widetext}
\begin{align}
  & T_3([1'2'3'4'])\prod_{i=1}^4A_i\BLvert\threeTriangles{1}{3}{2}{4}{3}\Brangle\nonumber\\
= & \frac{1}{|G|^4} \sum_{\substack{[1'1],[2'2],\\ [3'3],[4'4]\in G}}  \cdots\elf{1'}{1}{2}{4}  \elf{1'}{1}{2}{3}^{-1} \elf{1'}{2'}{2}{3}\nonumber\\
  &\qquad\quad\;\;\times\elf{1'}{2'}{2}{4}^{-1}\elf{2'}{2}{3}{4}\elf{1'}{2'}{3'}{3}^{-1}\nonumber\\ &\qquad\quad\;\;\times\elf{2'}{3'}{3}{4}^{-1}\elf{1'}{2'}{4'}{4}\elf{2'}{3'}{4'}{4}T_3     \BLvert\threeTriangles{1'}{3'}{2'}{4'}{3}\Brangle\nonumber\\
= & \frac{1}{|G|^4} \sum_{\substack{[1'1],[2'2],\\ [3'3],[4'4]\in G}}  \cdots\elf{1'}{1}{2}{4}  \elf{1'}{1}{2}{3}^{-1} \elf{1'}{2'}{2}{3}\elf{1'}{2'}{2}{4}^{-1}\nonumber\\
  &\qquad\quad\;\;\times\elf{2'}{2}{3}{4}\elf{1'}{2'}{3'}{3}^{-1}\elf{2'}{3'}{3}{4}^{-1}\nonumber\\ &\qquad\quad\;\;\times\elf{1'}{2'}{4'}{4}\elf{2'}{3'}{4'}{4}\elf{1'}{2'}{3'}{4'}\BLvert\oneTriangle{1'}{3'}{4'}\Brangle,\label{eq:T3Bv}
\end{align}
\be\label{eq:BvT3}
\begin{aligned}
&\prod_{i=1}^3A_iT_3([1234])\BLvert\threeTriangles{1}{3}{2}{4}{3}\Brangle =\elf{1}{2}{3}{4}\prod_{i=1}^3A_i\BLvert\oneTriangle{1}{3}{4}\Brangle\\
&=\frac{1}{|G|^3}\elf{1}{2}{3}{4}\sum_{\substack{[1'1],[3'3],\\ [4'4]\in G}}\elf{1'}{1}{3}{4}\elf{1'}{3'}{3}{4}^{-1}\elf{1'}{3'}{4'}{4} \BLvert\oneTriangle{1'}{3'}{4'}\Brangle.
\end{aligned}
\ee
\end{widetext}
The RHS of Eq. (\ref{eq:T3Bv}) and that of Eq. (\ref{eq:BvT3}) can be identified by applying the four $3$--cocycle conditions to the corresponding $3$--cocycles in Eq. (\ref{eq:T3Bv}).
\begin{widetext}
\be
\begin{aligned}
&\elf{1'}{1}{2}{4}\elf{1'}{1}{2}{3}^{-1}=\elf{1}{2}{3}{4}\elf{1'}{1}{3}{4}\elf{1'}{2}{3}{4}^{-1}\\
&\elf{1'}{2'}{2}{3}\elf{1'}{2'}{2}{4}^{-1}\elf{2'}{2}{3}{4}=\elf{1'}{2}{3}{4}\elf{1'}{2'}{3}{4}^{-1}\\
&\elf{1'}{2'}{3'}{3}^{-1}\elf{2'}{3'}{3}{4}^{-1}=\elf{1'}{3'}{3}{4}^{-1}\elf{1'}{2'}{3}{4}\elf{1'}{2'}{3'}{4}^{-1}\\
&\elf{2'}{3'}{4'}{4}\elf{1'}{2'}{4'}{4}=\elf{1'}{3'}{4'}{4}\elf{1'}{2'}{3'}{4}\elf{1'}{2'}{3'}{4'}^{-1}.
\end{aligned}
\ee
\end{widetext}
These $3$--cocycle conditions render the $3$--cocycles resulted independent of the $\sum_{[2'2]}$ in Eq. (\ref{eq:T3Bv}), which then contributes a factor of $|G|$. Therefore, we conclude that Eq. (\ref{eq:TBcommute}) holds for all mutation operators and any state in $\Hil^{B_f=1}_{\Gamma}$, which means that the mutation operators preserve the space of  ground states.

\subsection{Unitarity}

Now we need to show that all mutation transformations are unitary, i.e.,
they satisfy Eq. \eqref{unitary}. It is sufficient to show all generators
$T_1$, $T_2$, and $T_3$ are unitary. We show them respectively in the following.

We first demonstrate that $T_1$ is unitary not only on the ground states but also so over the entire subspace $\Hil^{B_f=1}$. Let us consider the action of $T^2_1$ on a generic basis state as follows, in which only the relevant part of the graph is shown.
\begin{align*}
  & T^2_1\BLvert \twoTriangles{v_1}{v_2}{v_3}{v_4}{0}\Brangle\\=
  &\sum_{[v_1v_3]\in G}\elf{v_1}{v_2}{v_3}{v_4}
   T_1\BLvert \twoTriangles{v_1}{v_2}{v_3}{v_4}{1} \Brangle\\
  =&\sum_{[v_1v_3]\in G}\sum_{[v_2v_4]'\in G}\elf{v_1}{v_2}{v_3}{v_4}\\ &\times\elf{v_1}{v_2}{v_3}{v_4}^{-1}\BLvert \twoTriangles{v_1}{v_2}{v_3}{v_4}{0} \Brangle\\
  =&\sum_{[v_1v_3]\in G}\sum_{[v_2v_4]'\in G}\delta_{[v_2v_4]',[v_2v_4]} \BLvert \twoTriangles{v_1}{v_2}{v_3}{v_4}{0} \Brangle\\
  =&\sum_{[v_1v_3]\in G}
  \BLvert \twoTriangles{v_1}{v_2}{v_3}{v_4}{0} \Brangle\\
  =&\BLvert \twoTriangles{v_1}{v_2}{v_3}{v_4}{0} \Brangle\quad,
\end{align*}
which deserves some explanation. In the second equality, the action of $T_1$ on two triangles sharing a horizontal edge is understood from rotating by $\pi/2$ either way the action of $T_1$ on two triangles with a vertical common edge as in the first equality because there is a global rotation symmetry on the graph. The inverse $3$--cocycle $\alpha^{-1}$ follows the rule described below Eq. (\ref{eq:T1move}). The group element $[v_2v_4]'$ that is summed over in the second equality is brought by the action of the second $T_1$, as the edge $[v_2v_4]$ in the second equality is a new edge relative to the $[v_2v_4]$ on the LHS of the equation, despite the same end vertices, and thus should be graced with a different group element $[v_2v_4]'$, which is however constrained by $B_f=1$ on the four triangles to be equal to $[v_2v_4]$. The last equality is also due to the restriction to the subspace $\Hil^{B_f=1}$.
As a result, $T_1=T_1^{-1}$, such that we can define that $T_1^{\dag}=T_1$ and infer that $T_1$ is unitary over the entire $\Hil^{B_f=1}$.

Next, we show that $T_3T_2=1$ on the entire subspace $\Hil^{B_f=1}$. We consider the action of $T_3T_2$ on a generic basis state as follows, in which only the relevant part of the graph is shown.
\begin{align*}
&T_3 T_2 \BLvert \oneTriangle{v_1}{v_2}{v_3} \Brangle\\
=&\sum_{\substack{[v_1q],[v_2q],\\ [v_3q]\in G}}\elf{q}{v_1}{v_2}{v_3}T_3\BLvert \threeTriangles{v_1}{v_2}{q}{v_3}{4} \Brangle\\
=&\sum_{[v_1q],[v_2q],[v_3q]\in G}\elf{q}{v_1}{v_2}{v_3}\\
&\times \elf{q}{v_1}{v_2}{v_3}^{-1}\BLvert \oneTriangle{v_1}{v_2}{v_3}\Brangle\\
=&\BLvert \oneTriangle{v_1}{v_2}{v_3}\Brangle,
\end{align*}
where the inverse $3$--cocycle $\alpha^{-1}$ follows from the rule introduced below Eq. (\ref{eq:T3move}). At this point, one may think that $T_3$ is the inverse of $T_2$ on $\Hil^{B_f=1}$. But this is not true because $T_2T_3\neq 1$ in general. Nevertheless, fortunately, as we now show, $T_2T_3=1$ on the ground states $\Hil^0$.

Since $T_2T_3P^0=T_2P^0T_3=P^0T_2T_3$ on $\Hil^0$, we have
\begin{align}
&T_2T_3P^0(1,2,3,4)\BLvert \threeTriangles{1}{3}{2}{4}{3} \Brangle\nonumber\\
=& T_2 P^0(1,2,3,4) T_3\BLvert \threeTriangles{1}{3}{2}{4}{3} \Brangle\nonumber\\
=& \elf{1}{2}{3}{4}T_2 P^0(1,2,3)\BLvert \oneTriangle{1}{3}{4}\Brangle\nonumber\\
=& \elf{1}{2}{3}{4}P^0(1,2,3) T_2\BLvert \oneTriangle{1}{3}{4}\Brangle\nonumber\\
=& P^0(1,2,3)\sum_{\substack{[12'],[2'3],\\ [2'4]\in G}}\elf{1}{2}{3}{4}\nonumber\\ \times & \elf{1}{2'}{3}{4}^{-1}\BLvert \threeTriangles{1}{3}{2'}{4}{3} \Brangle\nonumber\\
=& P^0(1,2,3)\sum_{\substack{[12'],[2'3],\\ [2'4]\in G}}\elf{1}{2'}{2}{3}\nonumber\\
\times & \elf{2'}{2}{3}{4}\elf{1}{2'}{2}{4}^{-1}\BLvert \threeTriangles{1}{3}{2'}{4}{3} \Brangle\nonumber\\
=& P^0(1,2,3)\frac{1}{|G|}\sum_{[2'2]\in G}\elf{1}{2'}{2}{3}\nonumber\\
\times & \elf{2'}{2}{3}{4}\elf{1}{2'}{2}{4}^{-1}\BLvert \threeTriangles{1}{3}{2'}{4}{3} \Brangle\nonumber\\
=& P^0(1,2,3)A_2\BLvert \threeTriangles{1}{3}{2}{4}{3} \Brangle\nonumber\\
=& P^0(1,2,3,4)\BLvert \threeTriangles{1}{3}{2}{4}{3} \Brangle, \label{eq:T2T3is1onH0}
\end{align}
where $P^0(1,2,3)$ and $P^0(1,2,3,4)$ are projectors acting on the vertices of the corresponding basis graph. The fifth equality in the equation above is obtained by applying to the two $3$--cocycles in the fourth row the following $3$--cocycle condition.
\begin{align*}
&\frac{\elf{1}{2}{3}{4}}{\elf{1}{2'}{3}{4}}=
\frac{\elf{1}{2'}{2}{3}\elf{2'}{2}{3}{4}}{\elf{1}{2'}{2}{4}}.
\end{align*}
We now explain why this $3$--cocycle condition is applicable. The action of $T_3$ on the basis graph in the first row of Eq. \eqref{eq:T2T3is1onH0} followed by an action of $T_2$ on the resulted basis vector in the third row of the equation can be viewed (see Fig. \ref{fig:T2T3isA} and assume the line $[2'2]$ nonexisting for the moment) as if the original tetrahedron $[1 2 3 4]$ is first flattened to be the triangle $[1 3 4]$ and then lifted again to a new tetrahedron $[1 2' 3 4]$. These two tetrahedra share a face, i.e., the triangle $[1 3 4]$, so there are seven triangles all told.The restriction to the subspace $\Hil^{B_f=1}$ imposes a chain rule of the three group elements on each of the seven triangles, such that the following identities on the group elements hold.
\begin{align*}
&[2' 1]\cdot [1 2]=[2' 3]\cdot [3 2],\\
&[2' 3]\cdot [3 2]=[2' 4]\cdot [4 2],\\
&[2' 4]\cdot [4 2]=[2' 1]\cdot [1 2].
\end{align*}
This enables one to add the line $[2' 2]$ and assign $[2' 2]= [2' 1]\cdot [1 2]=[2' 3]\cdot [3 2]=[2' 4]\cdot [4 2]$. Hence there are now five tetrahedra in Fig. \ref{fig:T2T3isA} with $B_f=1$ on any of the ten triangles, justifying the $3$--cocycle condition above. Furthermore, $[2' 2]$ determines $[12']$, $[2'3]$ and $[2'4]$, the summation in the fifth equality in Eq. \eqref{fig:T2T3isA} can be replaced by $\sum_{[2'2]}/|G|$, where the factor of $1/|G|$ arises because $[2'2]\in G$ is a new element to be summed over.
\begin{figure}[h!]
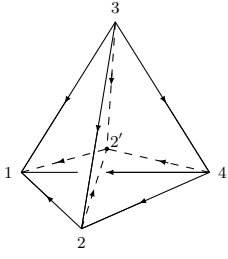

\centering
\BvTri{1}{2}{3}{4}{2}
\caption{The topology of the action of $T_2T_3$.}
\label{fig:T2T3isA}
\end{figure}

We infer from Eq. \eqref{fig:T2T3isA} that $T_2T_3=1$ on the ground states, which together with $T_3T_2=1$ on $\Hil^{B_f=1}$, implies that $T_2=T_3^{\dag}$ and $T_3=T_2^{\dag}$ on $\Hil^0$. That is, $T_2$ and $T_3$ are unitary on the ground states.

An interesting byproduct of this proof is that on the subspace $\Hil^{B_f=1}$,
\be
T_2T_3=A_v,
\ee
where $v$ is the vertex annihilated by the action of $T_3$.

\section{Ground state operators}
We prove Eq. \eqref{torusAxRewrite} in this appendix.
We start with the coefficient in Eq. \eqref{eq:torusAx},
which can be expressed by $I^x(a,b)/I^x(b,a)$, where $I^x(a,b)$ is
\begin{align}
  \label{Axab}
  I^x(a,b)=
  \alpha(a,bx^{-1},x)
  \alpha(bx^{-1},x,ax^{-1})
  \alpha(x,ax^{-1},xbx^{-1}).
\end{align}

We use the 3--cocycle condition \eqref{3CocycleCondition}
together with the normalization condition \eqref{NormalizationCondition}
to rewrite the three 3--cocycles above. By
\begin{align}
  \label{dtdtdt}
  &\delta\alpha(a,b,x^{-1},x)=1
  \nonumber\\
  &\delta\alpha(b,x^{-1},x,ax^{-1})=1
  \nonumber\\
  &\delta\alpha(x,x^{-1},xax^{-1},xbx^{-1})=1
  \nonumber\\
  &\delta\alpha(x^{-1},x,x^{-1},xax^{-1})=1
  \nonumber
\end{align}
and have
\begin{align}
  &\alpha(a,bx^{-1},x)=\frac{\alpha(ab,x^{-1},x)}{\alpha(b,x^{-1},x)\alpha(a,b,x^{-1})}
  \nonumber\\
  &\alpha(bx^{-1},x,ax^{-1})=\frac{\alpha(x^{-1},x,ax^{-1})\alpha(b,x^{-1},x)}{\alpha(b,x^{-1},xax^{-1})}
  \nonumber\\
  &\alpha(x,ax^{-1},xbx^{-1})=
  \nonumber\\
  &\quad\quad \frac{1}{\alpha(x^{-1},xax^{-1},xbx^{-1})\alpha(x,x^{-1},xax^{-1})}
  \nonumber\\
  &\alpha(x^{-1},x,ax^{-1})=\alpha(x,x^{-1},xax^{-1})\alpha(x^{-1},x,x^{-1})
  \nonumber
\end{align}

These four identities leads to
\begin{align}
  &I^x(a,b)=
  \nonumber\\
  &\frac{\alpha(ab,x^{-1},x)\alpha(x^{-1},x,x^{-1})}
  {\alpha(a,b,x^{-1})\alpha(b,x^{-1},xax^{-1})\alpha(x^{-1},xax^{-1},xbx^{-1})}
\end{align}
Using $ab=ba$, We evaluate $I^x(a,b)/I^x(b,a)$ directly and arrive at Eq. \eqref{torusAxRewrite}.

\section{Modular Transformations}\label{app:modularTrans}
Here we show how we construct the modular transformation operators $\str^x$ and $\ttr^x$.
\begin{align}
 & \str^x:  \BLvert\torusgraphST{1}{}{2}{}{3}{}{4}{}{0}\Brangle
  \nonumber\\
\mapsto &\elf{2'}{1}{3}{4}{\elf{2'}{1}{2}{4}}^{-1} \BLvert \torusgraphST{1}{}{2}{}{3}{}{2}{'}{0}\Brangle
  \nonumber\\
\mapsto & \elf{2'}{1}{3}{4}{\elf{2'}{1}{2}{4}}^{-1}
  \nonumber\\
 \times &{\elf{2'}{4'}{1}{3}}^{-1} \BLvert\torusgraphST{1}{}{2}{}{4}{'}{2}{'}{0}\Brangle
    \nonumber\\
\mapsto & \frac{\elf{2'}{1}{3}{4}}{\elf{2'}{1}{2}{4}\elf{2'}{4'}{1}{3}}
  \nonumber\\
 \times  & {\elf{1'}{2'}{1}{2}}^{-1} \BLvert\torusgraphST{1}{}{1}{'}{4}{'}{2}{'}{0}\Brangle
      \nonumber\\
\mapsto & \elf{2'}{1}{3}{4}{\elf{2'}{1}{2}{4}}^{-1}
  \nonumber\\
 \times  & {\elf{2'}{4'}{1}{3}}^{-1}{\elf{1'}{2'}{1}{2}}^{-1}
   \nonumber\\
 \times  & \frac{\elf{2'}{3'}{4'}{1}}{\elf{1'}{2'}{3'}{1}} \BLvert \torusgraphST{3}{'}{1}{'}{4}{'}{2}{'}{0}\Brangle
         \nonumber\\
\mapsto & \elf{2'}{1}{3}{4}{\elf{2'}{1}{2}{4}}^{-1}
  \nonumber\\
 \times  & {\elf{2'}{4'}{1}{3}}^{-1}{\elf{1'}{2'}{1}{2}}^{-1}
   \nonumber\\
 \times  & \elf{2'}{3'}{4'}{1}{\elf{1'}{2'}{3'}{1}}^{-1}
   \nonumber\\
 \times  & {\elf{1'}{2'}{3'}{4'}} \BLvert \torusgraphST{3}{'}{1}{'}{4}{'}{2}{'}{1}\Brangle
\end{align}
\begin{align}
 & \ttr^x :\BLvert\torusgraphST{1}{}{2}{}{3}{}{4}{}{0}\Brangle
  \nonumber\\
\mapsto &\elf{2'}{1}{3}{4}{\elf{2'}{1}{2}{4}}^{-1} \BLvert\torusgraphST{1}{}{2}{}{3}{}{2}{'}{0}\Brangle
  \nonumber\\
\mapsto &\elf{2'}{1}{3}{4}{\elf{2'}{1}{2}{4}}^{-1}
  \nonumber\\
 \times  & \elf{3''}{2'}{1}{3} \BLvert\torusgraphST{1}{}{2}{}{3}{''}{2}{'}{0}\Brangle
  \nonumber\\
\mapsto &\elf{2'}{1}{3}{4}{\elf{2'}{1}{2}{4}}^{-1}
  \nonumber\\
 \times  & \frac{\elf{3''}{2'}{1}{3} }{\elf{4''}{2'}{1}{2}} \BLvert \torusgraphST{1}{}{4}{''}{3}{''}{2}{'}{0}\Brangle
  \nonumber\\
\mapsto &\elf{2'}{1}{3}{4}{\elf{2'}{1}{2}{4}}^{-1}
  \nonumber\\
 \times  & \elf{3''}{2'}{1}{3} {\elf{4''}{2'}{1}{2}}^{-1}
  \nonumber\\
 \times  &\frac{ \elf{1''}{3''}{2'}{1}}{\elf{1''}{4''}{2'}{1}} \BLvert \torusgraphST{1}{''}{4}{''}{3}{''}{2}{'}{0}\Brangle
  \nonumber\\
\mapsto &\elf{2'}{1}{3}{4}{\elf{2'}{1}{2}{4}}^{-1}
  \nonumber\\
 \times  & \elf{3''}{2'}{1}{3} {\elf{4''}{2'}{1}{2}}^{-1}
  \nonumber\\
 \times  & \elf{1''}{3''}{2'}{1} {\elf{1''}{4''}{2'}{1}}^{-1}
  \nonumber\\
 \times  & \elf{1''}{3''}{4''}{2'}
  \BLvert \torusgraphST{1}{''}{4}{''}{3}{''}{2}{'}{1}\Brangle
\end{align}

\section{Solutions for $\str$ and $\ttr$ matrices}\label{app:solModularTrans}
\subsection{Projective Characters.}

The centralizer subgroups $Z_g$ are isomorphic for all elements $g$ of the conjugacy class $C^A$. Therefore the corresponding projective representations are also isomorphic. Given the representation $\widetilde{\rho}^g_{\mu}$ for a fixed $g\in C^A$, we construct $\widetilde{\rho}^{xgx^{-1}}_{\mu}$ as follows. For all $x\in G$, the elements $xhx^{-1}$ runs over all elements in $xZ_gx^{-1}$ while $h$ runs over all elements in $Z_g$. We can define a projective representation $\widetilde{\rho}^{xgx^{-1}}_{\mu}$ of $xZ_gx^{-1}$ from a given representation $\widetilde{\rho}^{g}_{\mu}$ of $Z_g$, by
\begin{align}
  \label{eq:rhoxgx}
  \widetilde{\rho}^{xgx^{-1}}_{\mu}(xhx^{-1})
  =\frac{\beta_g(x^{-1},xhx^{-1})}{\beta_g(h,x^{-1})}\widetilde{\rho}^g_{\mu}(h)
\end{align}

We verify that $\widetilde{\rho}^{xgx^{-1}}_{\mu}$ is indeed a $\beta_{xgx^{-1}}$ representation. We check the the multiplication rule for all $h_1,h_2\in Z_g$
\begin{align}
  \label{eq:betaxgxrepresentation}
  &\widetilde{\rho}^{xgx^{-1}}_{\mu}(xh_1x^{-1}) \widetilde{\rho}^{xgx^{-1}}_{\mu}(xh_2x^{-1})
  \nonumber\\
  =&\frac{\beta_g(x^{-1},xh_1x^{-1})}{\beta_g(h_1,x^{-1})}\frac{\beta_g(x^{-1},xh_2x^{-1})}{\beta_g(h_2,x^{-1})}
  \widetilde{\rho}^g_{\mu}(h_1)\widetilde{\rho}^g_{\mu}(h_2)
  \nonumber\\
  =&\frac{\beta_g(x^{-1},xh_1x^{-1})}{\beta_g(h_1,x^{-1})}\frac{\beta_g(x^{-1},xh_2x^{-1})}{\beta_g(h_2,x^{-1})}
  \beta_g(h_1,h_2)\widetilde{\rho}^g_{\mu}(h_1h_2)
  \nonumber\\
  =&\frac{\beta_g(x^{-1},xh_1x^{-1})}{\beta_g(h_1,x^{-1})}\frac{\beta_g(x^{-1},xh_2x^{-1})}{\beta_g(h_2,x^{-1})}
  \beta_g(h_1,h_2)
  \nonumber\\
  &\times
  \left[\frac{\beta_g(x^{-1},xh_1h_2x^{-1})}{\beta_g(h_1h_2,x^{-1})}\right]^{-1}
  \widetilde{\rho}^{xgx^{-1}}_{\mu}(xh_1h_2x^{-1})
  \nonumber\\
  =&\frac{\eta^g(h_1,x)\eta^g(h_2,x)}{\eta^g(h_1h_2,x)}\beta_g(h_1,h_2) \widetilde{\rho}^{xgx^{-1}}_{\mu}(xh_1h_2x^{-1})\nonumber\\
  =&\beta_{xgx^{-1}}(xh_1x^{-1},xh_2x^{-1})\widetilde{\rho}^{xgx^{-1}}_{\mu}(xh_1h_2x^{-1})
\end{align}
where the last equality is obtained by using the following relation:
\be\label{eq:3etaRel}
\frac{\eta^g(h_1,x)\eta^g(h_2,x)}{\eta^g(h_1h_2,x)}=\frac{\beta_{xgx^{-1}}(xh_1x^{-1}, xh_2x^{-1})}{\beta_g(h_1,h_2)}
\ee
which can be verified by directly applying the twisted 2--cocycle conditions $\widetilde{\delta}\beta_g=1$ in \eqref{eq:betacondition} successively to the
triples $(h_1,h_2,x^{-1})$, $(h_1,x^{-1},xh_2x^{-1})$, $(x^{-1},xh_1x^{-1},xh_2x^{-1})$ and make the corresponding substitutions.

An immediate consequence of the above isomorphism is the relation between the projective characters,
\be
\label{eq:CharacterInConjugacyClassesAppendix}
\widetilde{\chi}^{xgx^{-1}}_{\mu}(xhx^{-1})=\eta^g(h,x)\widetilde{\chi}^g_{\mu}(h),
\ee
which is the very Eq. (\ref{eq:CharacterInConjugacyClasses}). This relation leads to the following proposition.
\begin{proposition}\label{prop:chiIs0}
If $h\in Z_g$ is not $\beta_g$--regular, $\widetilde{\chi}^g_{\mu}(h)=0$.
\end{proposition}
\begin{proof}
The proof is straightforward. If $h\in Z_g$ is not $\beta_g$--regular, there must exist $k\in Z_{g,h}$, such that $\beta_g(h,k)\neq\beta_g(k,h)$. By setting $x=k^{-1}$ in Eq. \eqref{eq:CharacterInConjugacyClassesAppendix} we have
\[
\widetilde{\chi}^{g}_{\mu}(h)= \frac{\beta_g(k,h)} {\beta_g(h,k)} \widetilde{\chi}^g_{\mu}(h)\Longrightarrow\widetilde{\chi}^g_{\mu}(h)=0.
\]
\end{proof}

\subsection{Some Proofs}
We prove the equalities in Eq. \eqref{eq:TonAmu} in particular the last one as follows.
We first write down the action of $\ttr\ket{A,\mu}$ explicitly.
\begin{align*}
\ttr\ket{A,\mu}=&\ttr^e P^0\ket{A,\mu}=\ttr^e\ket{A,\mu}\\
  =&\frac{1}{|G|}\sum_{g\in C^A,h\in Z_g}\sum_{\nu=1}^{r(Z^A,\beta_{g^A})}\\
  \times &\overline{\widetilde{\chi}^g_{\mu}(h)}\alpha(g^{-1}h^{-1},g,h)\alpha(g^{-1}h^{-1},h,g)^{-1}\\
  \times & \alpha(h,g^{-1}h^{-1},g)\alpha(g,g^{-1}h^{-1},h)^{-1}\\
  \times &\alpha(g,h,g^{-1}h^{-1})\alpha(h,g,g^{-1}h^{-1})^{-1}\\
  \times &\alpha(g,hg^{-1},g)
  \overline{\widetilde{\chi}^{g}_{\nu}(g^{-1}h^{-1})}\ket{A,\nu},
\end{align*}
where the inverse transformation in Eq. \eqref{eq:ghToAmu} is used. Note that $g$ must belong to precisely one conjugacy class, so in the equation above the classes $C^A$ and $C^B$
are identified. Knowing this, we rewrite the action of $\ttr$ in the equation above by grouping the $3$--cocycles into the twisted $2$--cocycles.
\begin{align}
  &\ttr\ket{A,\mu}
  \nonumber\\
  =&\frac{1}{|G|}\sum_{\substack{g\in C^A\\ h\in Z_g}}\sum_{\nu=1}^{r(Z^A,\beta_{g^A})}
  {\beta_g(g,g^{-1}h)}
  {\widetilde{\chi}^g_{\mu}(h)}
  \overline{\widetilde{\chi}^g_{\nu}(g^{-1}h)}
  \nonumber\\
  & \times\frac{\beta_g(h,g^{-1}h^{-1})}{\beta_g(g^{-1}h^{-1},h)}
  \ket{A,\nu}\\
  =&\frac{1}{|G|}\sum_{\substack{g\in C^A\\ h\in Z_g}}\sum_{\nu=1}^{r(Z^A,\beta_{g^A})}
  {\beta_g(g,g^{-1}h)}
  {\widetilde{\chi}^g_{\mu}(h)}
  \overline{\widetilde{\chi}^g_{\nu}(g^{-1}h)}
  \nonumber\\
  & \times  [\eta^g(h,gh)]^{-1}\ket{A,\nu}\nonumber\\
=&\frac{1}{|G|}\sum_{\substack{g\in C^A\\ h\in Z_g}}\sum_{\nu=1}^{r(Z^A,\beta_{g^A})}
  {\beta_g(g,g^{-1}h)}
  {\widetilde{\chi}^g_{\mu}(h)}
  \overline{\widetilde{\chi}^g_{\nu}(g^{-1}h)}\ket{A,\nu}
  \nonumber
\end{align}
where in the last step use is made of Eq. (\ref{eq:CharacterInConjugacyClasses}) to absorb the $\eta^g$ term into the projective characters $\widetilde{\chi}^g_{\mu}$ and $\widetilde{\chi}^g_{\nu}$.

Because Eq. \eqref{eq:rhoggSchurLemma} says that $\widetilde{\rho}^g_{\mu}(g)$ is a multiple of the identity matrix, we have
\begin{align*}
\beta_g(g,g^{-1}h)\widetilde{\chi}^g_{\mu}(h)
  & =\mathrm{tr}\left[\beta_g(g,g^{-1}h)\widetilde{\rho}^g_{\mu}(h)\right]\\
  &=\mathrm{tr}[\widetilde{\rho}^g_{\mu}(g)\widetilde{\rho}^g_{\mu}(g^{-1}h)]\\
  &=\mathrm{tr}[\frac{\widetilde{\chi}^{g^A}_{\mu}(g^A)}{\text{dim}_{\mu}} \mathds{1}\widetilde{\rho}^g_{\mu}(g^{-1}h)]\\
&= \frac{\widetilde{\chi}^{g^A}_{\mu}(g^A)}{\text{dim}_{\mu}} \widetilde{\chi}^g_{\mu}(g^{-1}h).
\end{align*}
%
Clearly, $\sum_{h}=\sum_{g^{-1}h}$, together with the orthogonality condition in Eq. \eqref{eq:characterrelation}, the summation evaluates to
\be
\ttr
\ket{A,\mu}=\frac{\widetilde{\chi}^{g^A}_{\mu}(g^A)}{\text{dim}_{\mu}} \ket{A,\mu},
\ee
confirming Eq. \eqref{eq:TonAmu}.

In the sequel, we derive the $\str$--matrix in Eq. (\ref{eq:smatrix}). We first act the $\str$ operator on a generic eigenvector $\ket{B,\nu}$ of the $\ttr$ operator.
\begin{align*}
\str\ket{B,\nu}=&\str^e P^0\ket{B,\nu}=\str^e\ket{B,\nu}\\
  =&\frac{1}{\sqrt{|G|}}\sum_{g'\in C^B,h'\in Z_{g'}} \widetilde{\chi}^{g'}_{\nu}(h')\\
  \times &\alpha(g'^{-1}h'^{-1},g',h')\alpha(g'^{-1}h'^{-1},h',g')^{-1}\\
  \times & \alpha(h'^{-1},g'^{-1},g')^{-1}\alpha(g',g'^{-1}h'^{-1},h')^{-1}\\
  \times &\alpha(g'^{-1}h'^{-1},g',g'^{-1})\alpha(g'^{-1},g^{-1}h^{-1},e)^{-1}\\
  \times &\alpha(g',g'^{-1}h'^{-1},g')\ket{h'^{-1},g'}.
\end{align*}
Notice that in the equation above, the first two and the fourth $\alpha$ terms define a twisted $2$--cocycle $\beta_{g'}(g'^{-1}h'^{-1},h')^{-1}$ the sixth $\alpha$ equals 1 as it is normalized, and by the $3$--cocycle condition the fifth and seventh $\alpha$ terms are equal to $\alpha(h'^{-1},g',g'^{-1})\alpha(g', g'^{-1}h'^{-1},g')^{-1}$, which together with the third $\alpha$ term, define another twisted $2$--cocycle $\beta_{g'}(h'^{-1},g'^{-1})^{-1}$. As such, the above equation becomes
\begin{align*}
\str\ket{B,\nu}=&\frac{1}{\sqrt{|G|}}\sum_{g'\in C^B, h'\in Z_{g'}} \widetilde{\chi}^{g'}_{\nu}(h')\\
\times & \beta_{g'}(g'^{-1}h'^{-1},h')^{-1}\beta_{g'}(h'^{-1},g'^{-1})^{-1} \ket{h'^{-1},g'}\\
=&\frac{1}{\sqrt{|G|}}\sum_{g'\in C^B, h'\in Z_{g'}} \overline{\widetilde{\chi}^{g'}_{\nu}(h'^{-1})} \beta_{g'}(h'^{-1},h')\\
\times & \beta_{g'}(g'^{-1}h'^{-1},h')^{-1}\beta_{g'}(h'^{-1},g'^{-1})^{-1} \ket{h'^{-1},g'}\\
=&\frac{1}{\sqrt{|G|}}\sum_{g'\in C^B, h'\in Z_{g'}} \overline{\widetilde{\chi}^{g'}_{\nu}(h'^{-1})}\ket{h'^{-1},g'},
\end{align*}
where the second equality and third equality are respectively the results of the following two identities.
\begin{align*}
&\widetilde{\rho}^{g'}_{\nu}(h')\widetilde{\rho}^{g'}_{\nu}(h'^{-1})=\beta_{g'}(h'^{-1},h') \mathds{1}\\
\Leftrightarrow \quad& \widetilde{\rho}^{g'}_{\nu}(h')=\left[\widetilde{\rho}^{g'}_{\nu}(h'^{-1})\right]^{\dag} \beta_{g'}(h'^{-1},h')\\
\Leftrightarrow \quad& \widetilde{\chi}^{g'}_{\nu}(h')=\overline{\widetilde{\chi}^{g'}_{\nu}(h'^{-1})} \beta_{g'}(h'^{-1},h'),
\end{align*}
which is due to the unitarity of the projective representation $\widetilde{\rho}$.
\begin{align*}
&\frac{\beta_{g'}(h'^{-1},h')\beta_{g'}(g'^{-1},e)}{\beta_{g'}(g'^{-1}h'^{-1},h') \beta_{g'}(g'^{-1},h'^{-1})}=1\\
\Longleftrightarrow\quad & \frac{\beta_{g'}(h'^{-1},h')}{\beta_{g'}(g'^{-1}h'^{-1},h') \beta_{g'}(h'^{-1},g'^{-1})}=1,
\end{align*}
which is a consequence of the (twisted) $2$--cocycle condition, the normalization of $\beta_{g'}$, and the fact that $\beta_{g'}(g'^{-1},h'^{-1})= \beta_{g'}(h'^{-1},g'^{-1})$ because $h'^{-1}$ must be $\beta_{g'}$--regular otherwise $\widetilde{\chi}^{g'}_{\nu}(h'^{-1})$ vanishes.

As such, the $\str$--matrix reads
\begin{align*}
&s_{(A\mu)(B\nu)}\\
=&\bbra{A,\mu}\str\ket{B,\nu}\\
=&\frac{1}{|G|}\sum_{\substack{g\in C^A\\ h\in Z_g}}\sum_{\substack{g'\in C^B\\ h'\in Z_{g'}}}\overline{\widetilde{\chi}^g_{\mu}(h)} \overline{\widetilde{\chi}^{g'}_{\nu}(h'^{-1})} \langle g,h\ket{h'^{-1},g'}\\
=&\frac{1}{|G|}\sum_{\substack{g\in C^A\\ h\in Z_g}}\sum_{\substack{h\in C^B\\ g\in Z_{h}}}\overline{\widetilde{\chi}^g_{\mu}(h) \widetilde{\chi}^h_{\nu}(g)},
\end{align*}
where $\langle g,h\ket{h'^{-1},g'}=\delta_{h'^{-1},g}\delta_{g',h}$ is understood. This proves Eq. (\ref{eq:smatrix}).
\end{appendix}

\bibliographystyle{apsrev}
\bibliography{StringNet}

\end{document}